\documentclass{article}
\usepackage[a4paper,margin=2cm]{geometry}
\usepackage{amsfonts} 
\usepackage{amsmath} 
\usepackage[utf8]{inputenc}
\usepackage[english]{babel} 
\usepackage{mathtools} 
\usepackage{latexsym}
\usepackage{authblk}
\usepackage{amsthm}
\usepackage{graphicx, color}
\usepackage{titlesec}
\titleformat{\subsection}[runin]
{\normalfont\bfseries}{\thesubsection}{.5em}{}
\titleformat{\subsubsection}[runin]
{\normalfont\itshape}{\thesubsubsection}{.5em}{}
\theoremstyle{plain}
\newtheorem{definition}{Definition}

\allowdisplaybreaks
\usepackage{amssymb}
\usepackage[round]{natbib} 
\usepackage{hyperref}

\definecolor{blendedblue}{rgb}{0.2,0.2,0.7}

\usepackage{algpseudocode}
\usepackage{algorithm,tabularx}
\usepackage{amsmath}
\usepackage{amssymb}

\DeclareMathOperator*{\argmin}{arg\,min}
\usepackage{multirow}
\usepackage{subcaption}
\usepackage{wrapfig}
\usepackage{stfloats}
\newtheorem{problem}{Problem}
\usepackage{bbm}
\usepackage[dvipsnames,table]{xcolor}
\usepackage{float}
\usepackage{graphicx}
\newtheorem{theorem}{Theorem}

\newtheorem{lemma}{Lemma}
\usepackage{booktabs} 
\definecolor{lightgray}{gray}{0.92}
\usepackage[font=small,labelfont=bf,width=0.85\textwidth]{caption}
\usepackage{multirow, makecell}
\usepackage{diagbox}
\usepackage{adjustbox}
\usepackage{afterpage}
\usepackage{pdflscape}
\usepackage[multiple]{footmisc}
\makeatletter
\newcommand{\multiline}[1]{%
	\begin{tabularx}{\dimexpr\linewidth-\ALG@thistlm}[t]{@{}X@{}}
		#1
	\end{tabularx}
}
\makeatother
\usepackage{url}
\graphicspath{ {./images/} }

\algnewcommand\algorithmicforeach{\textbf{for each}}
\algdef{S}[FOR]{ForEach}[1]{\algorithmicforeach\ #1\ \algorithmicdo}

\theoremstyle{definition}

\makeatletter
\let\@fnsymbol\@arabic
\makeatother
\title{Sparse-group SLOPE: adaptive bi-level selection with FDR-control}
\author{Fabio Feser}
\author{Marina Evangelou}
\affil{Department of Mathematics, Imperial College London}

\date{}

\begin{document}
\maketitle
	
\begin{abstract}
%There are situations where the variables of the input design matrix have a natural grouping and the aims of the analysis of such data is both to select the groups of features and to identify the features drivers of these group associations. An example of such situation is the analysis of pathways (gene-sets) of gene expression data. 
In this manuscript, a new high-dimensional approach for simultaneous variable and group selection is proposed, called sparse-group SLOPE (SGS). SGS achieves false discovery rate control at both variable and group levels by incorporating the SLOPE model into a sparse-group framework and exploiting grouping information. A proximal algorithm is implemented for fitting SGS that works for both Gaussian and Binomial distributed responses. Through the analysis of both synthetic and real datasets, the proposed SGS approach is found to outperform other existing lasso- and SLOPE-based models for bi-level selection and prediction accuracy. Further, model selection and noise estimation approaches for selecting the tuning parameter of the regularisation model are proposed and explored. 

%It is further considered as a selection tool under Gaussian designs and is found to have a higher $\text{F}_1$ score than existing lasso- and SLOPE-based models. Finally, SGS is applied to real data to predict colitis cases, achieving a peak classification accuracy of $97.4\%$, again outperforming other lasso- and SLOPE-based models.

\noindent\textbf{Code:} SGS is implemented in the repository \href{https://github.com/ff1201/sgs}{github.com/ff1201/sgs}. An \texttt{R} package will be available shortly.

\noindent\textbf{Contact:} \href{mailto:ff120@ic.ac.uk}{ff120@ic.ac.uk}
	\end{abstract}

\section{Introduction}
%Consider the typical high-dimensional setting with input features $\mathbf{X}\in \mathbb{R}^{n\times p}$, where $p>>n$. Suppose, we have a response vector, $y \in \mathbb{R}^n$, which has been generated from the linear model $y = \mathbf{X}\beta + \epsilon$, where $\beta \in \mathbb{R}^p$, and $\epsilon \sim \mathcal{N}(0,\sigma^2 >0)$. 

Exploring the relationships between a continuous response, $y \in \mathbb{R}^n$, and a design matrix, $\mathbf{X}\in \mathbb{R}^{n\times p}$, is usually done by fitting a linear regression model $y = \mathbf{X}\beta + \epsilon$, where $\beta \in \mathbb{R}^p$ and $\epsilon \sim \mathcal{N}(0,\sigma^2 >0)$. The problem of identifying the variables that have a non-zero effect on the response $y$ is called variable selection. %In mathematics this is defined as In mathematical terms variable selection is defined as $S = \{j: \beta_j \neq 0 \}\subset \{1,\dots,p\}$. 
One of the most popular approaches for variable selection when working with high-dimensional data, $p>>n$, is the \textit{least absolute shrinkage and selection operator} (\textit{lasso}) proposed by \cite{Tibshirani1996}. The lasso performs variable selection by applying the $\ell_1$ penalty, defined by the norm $\|x\|_1 = \sum_i |x_i|$, that shrinks the coefficients of the features, setting some exactly equal to zero. Over the years, a number of extensions of the lasso have been proposed in the literature for overcoming some of its limitations. The lasso was shown to be inconsistent under certain scenarios in \cite{Zou2006}, who then proposed the adaptive lasso, which achieves the oracle properties by assigning different weights to the features. Further, as a consequence of using only the $\ell_1$ penalty, the lasso can select at most $n$ variables. Thus, the elastic net extension was proposed, which combines the $\ell_1$ and $\ell_2$ penalties, and so does not suffer from this limitation \citep{Zou2005}.  

%Over the years, a number of extensions of the lasso have been proposed in the literature including the adaptive lasso (\citep{Zou2006}) where different weights are assigned to the features, and elastic net (\citep{Zou2005}, which combines the $\ell_1$ and $\ell_2$ penalties (also known as the lasso and ridge penalties). These extensions were introduced to overcome some of the limitations of the lasso, including it is inconsistent under certain scenarios and can select at most $n$ predictors.

One of the challenges of variable selection is controlling the false discovery rate (FDR), as the tests for identifying the associated variables are performed simultaneously, leading to a multiple testing problem. \cite{Bogdan2015} proposed an adaptive extension of the lasso that is considered to be a bridge between the lasso and FDR-control in multiple testing. The proposed method, named 
\textit{sorted L-one penalised estimation} (\textit{SLOPE}), applies the penalty: $\text{pen}_\text{SLOPE}(b) = \sum_{i=1}^{p}\lambda_i |b|_{(i)}$, where $\lambda_1 \geq \dotsc \geq \lambda_p$, $\left|b\right|_{(1)} \geq \dotsc \geq \left|b\right|_{(p)}$. SLOPE reduces to the lasso for $\lambda_1 = \dots = \lambda_p$. It is similar to the adaptive lasso approach, but whilst in the adaptive lasso the penalties tend to decrease with increasing magnitude of the coefficients, the opposite occurs in SLOPE \citep{Bogdan2015}. A direct link to the Benjamini-Hockberg (BH) procedure and FDR-control is found through the choice of the penalty parameters. The BH critical values are used, so that for a variable $i$, $\lambda_i = z(1-i \cdot q_v/2p)$, where $q_v\in (0,1)$ is the desired variable FDR level and $z(\cdot)$ is the quantile function of a standard normal distribution. It has been shown that SLOPE achieves FDR-control under orthogonal designs \citep{Bogdan2015}. Additional useful properties of SLOPE include that it automatically finds the minimum total squared error loss over a range of sparsity classes, which means no a priori knowledge of the degree of sparsity is required, and it is asymptotically minimax \citep{Su2016}. 

Our work proposes an approach for dealing with situations where features arise as members of groups or from different data sources, where the aim is to select the groups and the features within the groups that are associated with the response. Examples of such cases include biological pathways; groups of genes working together for a specific product. When conducting pathway (gene set) analysis of genetics data, the interest is in identifying genes and pathways associated with a change in the risk profile of a disease. \cite{Evangelou2014} illustrated how genes discovered through pathway analysis have often been missed from conventional analyses and can have important biological roles in the development of a disease. Similarly, with the advancements of technology, many studies now involve the generation of multiple data sources, each with different features that describe the samples from alternative angles. As these data sources may contain noise variables, it is imperative that they are shrunk to zero, leaving only the signal features as non-zero. To this end, \cite{Baker2020} proposed a data-integration approach based on the lasso for multi-view feature selection. 

Both the lasso and SLOPE have been extended to selecting groups of variables, rather than just individual variables. Consider some $m$-partition of the input space, $\mathcal{G} = \{ G_1, \dots, G_m\}$ of the set $\{1,\dots, p\}$, such that $G_i \cap G_j = \emptyset$ for $i\neq j$ and $\bigcup_{i=1}^m G_i = \{1,\dots,p\}$, where $p_g := |G_g|$ is the size of group $g$. Then, \textit{Group SLOPE} (\textit{gSLOPE}) is defined by applying the $\ell_2$ norm to the group effects: $\text{pen}_\text{gSLOPE}(b) = \sum_{g=1}^{m} \lambda_g \sqrt{p_g} \| b^{(g)} \|_2$, where $\lambda_1 \geq \dotsc \geq \lambda_m$, $\sqrt{p_1}\|b^{(1)} \|_2 \geq \dotsc \geq \sqrt{p_m}\|b^{(m)} \|_2$, and $b^{(g)}\in \mathbb{R}^{p_g}$ is the vector of coefficients in group $g$ \citep{Gossmann2015,Brzyski2015}. gSLOPE achieves group FDR-control under orthogonal designs \citep{Brzyski2015}. With respect to the lasso, \cite{Yuan2006} introduced the \textit{group lasso} (\textit{gLasso}) approach with penalty: $\text{pen}_\text{gLasso}(b) =\sum_{g=1}^{m}\sqrt{p_g}\|b^{(g)}\|_2$, which reduces to the lasso when each group is a singleton. It creates sparsity at a group level by shrinking whole groups exactly to zero so that each variable within a group is also zero. Further, \cite{Simon2013} introduced the \textit{sparse-group lasso} (\textit{SGL}), which combines the lasso with gLasso to create models with bi-level sparsity. SGL was found to outperform both the lasso and gLasso when applied to predicting breast cancer cases using genomics data \citep{Simon2013}. 

In this manuscript, SLOPE is combined with gSLOPE for obtaining sparse solutions at both variable and group levels. The proposed approach, sparse-group SLOPE (SGS), is presented in $\S$\ref{section:SGS}. SGS works efficiently with high-dimensional data, performs bi-level selection, and simultaneously controls the variable and group FDRs under orthogonal designs; the last of which is not a property shared by SGL. SGS achieves FDR-control by applying more stringent penalisation. This is imperative when dealing with datasets with sparse representations, such as those found in genetics. An efficient algorithm is proposed for fitting SGS through an adaptive three operator splitting approach. In $\S$\ref{section:sgs_FDR} we present new penalty sequences which enable SGS to control the bi-level FDR. Through the analysis of both simulated and real data, it is illustrated how SGS outperforms existing competitive lasso- and SLOPE-based approaches. SGS was found to achieve more accurate variable and group selection than such methods under various scenarios, including random signals, large groups, and under the null model ($\S$\ref{section:sim_studies}). The problem of model selection with regards to SGS is explored in $\S$\ref{section:model_selection}, with a new noise estimation procedure proposed. Finally, SGS achieved higher classification accuracy than existing high-dimensional approaches when applied to predicting colitis and breast cancer cases using real genetic data in $\S$\ref{section:real_data}.
	
\section{Sparse-group SLOPE (SGS)}\label{section:SGS}
To incorporate the SLOPE concept into a sparse-group framework, we define \textit{sparse-group SLOPE} (\textit{SGS}) as the solution to the convex optimisation problem given by
\begin{equation}\label{eqn:sgs}
	\hat{\beta}_\text{SGS} := \argmin_{b\in \mathbb{R}^p}\left\{ \frac{1}{2n}\ell (b ; y, \mathbf{X}) + \lambda \alpha \sum_{i=1}^{p}v_i |b|_{(i)} + \lambda (1-\alpha)\sum_{g=1}^{m}w_g \sqrt{p_g} \|b^{(g)}\|_2 \right\},
\end{equation}
where $\ell(\cdot)$ is the loss function (choices of loss function are described in $\S$\ref{section:algorithm}). SGS can be seen to be a convex combination of SLOPE and gSLOPE (Figure \ref{fig:sgs_3d}), balanced through $\alpha \in [0,1]$, such that it reduces to SLOPE for $\alpha = 0$ and to gSLOPE for $\alpha = 1$. The tuning parameter $\lambda>0$ defines the degree of sparsity, as in the lasso, and can also be used to define a pathwise solution (discussed in $\S$\ref{section:model_selection_path}). SGS uses adaptive penalty weights, with variable weights $v = [v_{1} \;\; \dots \;\; v_{p}]^\top $, where $v_1 \geq \dotsc \geq v_{p}$ are matched with $|b|_{(1)} \geq \dots \geq |b|_{(p)}$, and group weights $w = [w_{1} \;\; \dots \;\; w_{G}]^\top$, where $w_1 \geq \dotsc \geq w_G$ are matched with $\sqrt{p_1}\|b^{(1)}\|_2 \geq \dotsc \geq \sqrt{p_m}\|b^{(m)}\|_2$; the choice of these weights are discussed in $\S$\ref{section:sgs_penalty}. SGS is a generalisation of many existing high-dimensional approaches, including the lasso, gLasso, SGL, SLOPE, and gSLOPE, using certain hyperparameter choices.
\begin{figure}[H]
	\centering
	\begin{subfigure}[b]{0.3\textwidth}
		\centering
		\includegraphics[width=\textwidth]{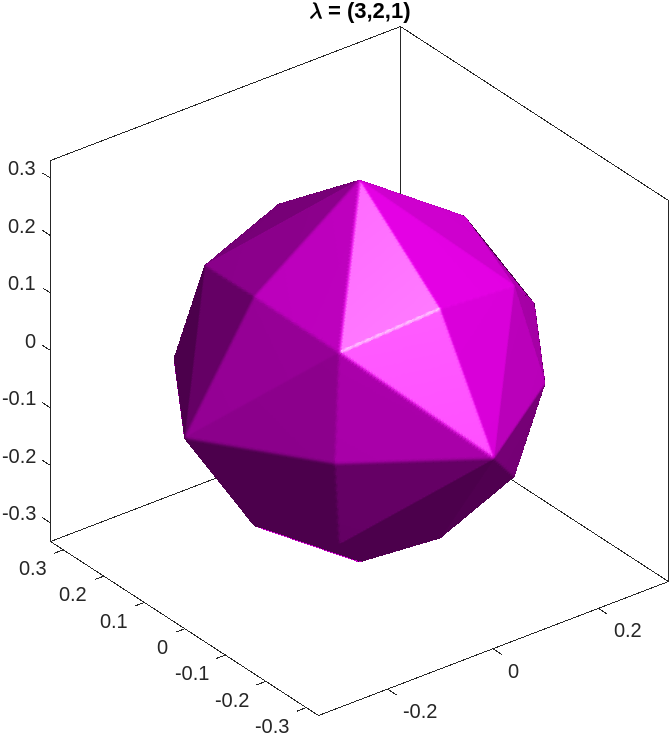}
		\caption[width=0.8\textwidth]{SLOPE}
		\label{fig:y equals x}
	\end{subfigure}
	\hfill
	\begin{subfigure}[b]{0.3\textwidth}
		\centering
		\includegraphics[width=\textwidth]{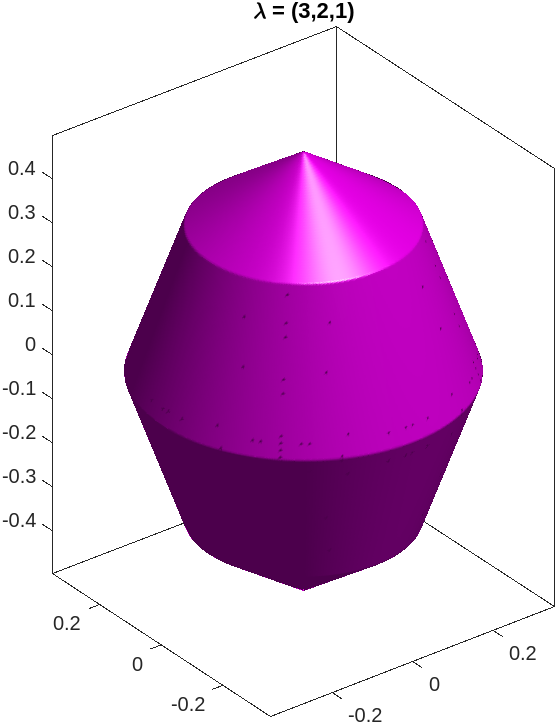}
		\caption[width=0.8\textwidth]{gSLOPE}
		\label{fig:three sin x}
	\end{subfigure}
	\hfill
	\begin{subfigure}[b]{0.3\textwidth}
		\centering
		\includegraphics[width=\textwidth]{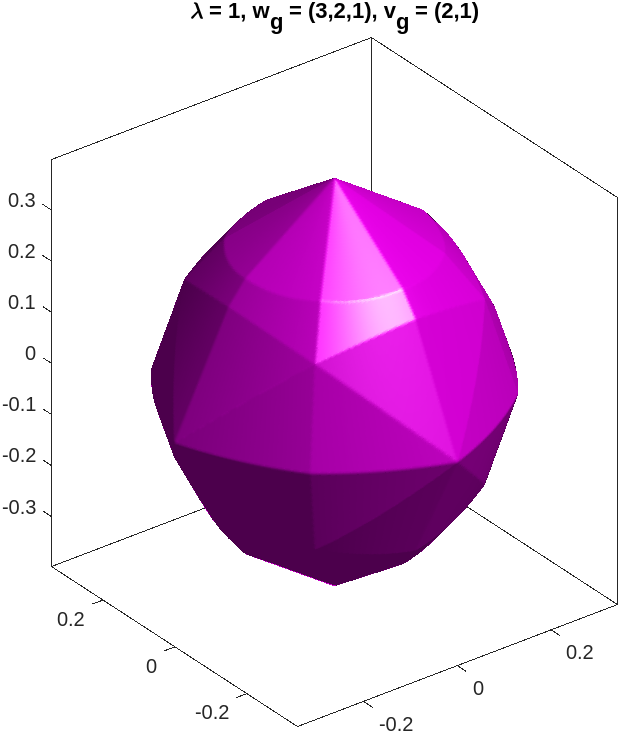}
		\caption[width=0.8\textwidth]{SGS with $\alpha = 0.5$}
		\label{fig:five over x}
	\end{subfigure}
	\caption[width=0.8\textwidth]{Units balls in $\mathbb{R}^3$ for the penalty functions of SLOPE (a), gSLOPE (b), and SGS (c). SGS can be seen to be a convex combination of SLOPE and gSLOPE.}
	\label{fig:sgs_3d}
\end{figure}
	
\subsection{Fitting algorithm.} \label{section:algorithm}
The penalty proposed in Equation (\ref{eqn:sgs}) is convex (the proof is given in $\S$\ref{appendix:fitting_algo}). Both the SLOPE and gSLOPE penalties are non-separable and data-dependent because of the sorting operation used \citep{Bu2019,Zhang2021}. Subsequently, the SGS penalty is also non-separable at both the variable and group level---that is,  $\text{Pen}_v(b) \neq \sum_{i=1}^{p} \text{pen}_v(b_i)$ and $\text{Pen}_g(b) \neq \sum_{g=1}^{m} \text{pen}_g(b^{(g)})$, where $\text{Pen}_v$ and $\text{Pen}_g$ are the SLOPE and gSLOPE penalties, respectively. As a result, blockwise gradient descent, which is used to fit SGL, is not guaranteed to converge to the global optimum \citep{Simon2013}. Instead, to fit SLOPE and gSLOPE, proximal algorithms are used, which do not require any separability assumptions. In proximal algorithms the coordinates are updated simultaneously, in contrast to the cyclic updates used in gradient descent. An upside is that non-separable penalties share information across variables and can detect grouping structures. This makes a non-separable penalty preferable for a group regression setting \citep{Rockova2016}.

Proximal gradient algorithms solve optimisation problems of the form $\min_x g(x) + h(x)$, where $g, h$ are convex functions and $g$ is differentiable. The SGS optimisation problem falls under such a scenario. SLOPE and gSLOPE are fitted using the proximal fast-iterative shrinkage-thresholding algorithm (FISTA) \citep{Beck2009}. Using a proximal algorithm requires being able to evaluate the \textit{proximal mapping}, given by
\begin{equation}\label{eqn:prox_mapping}
	\operatorname{prox}(x) := \argmin_{z}\left\{\frac{1}{2t} \left\|z-x \right\|_2^2 + h(z)\right\}.
\end{equation}
For non-separable penalties, such as the lasso, where it is given by the soft-thresholding operator, the mapping is usually derived using simple calculus. However, for non-separable penalties, finding the mapping is not trivial. Indeed, for both SLOPE and gSLOPE, a separate algorithm is required to compute the mappings, on top of the proximal algorithm \citep{Bogdan2015,Gossmann2015}. Instead of attempting to find the proximal mapping of SGS directly, we can exploit the fact that the mappings of SLOPE and gSLOPE are already known (given by Algorithm 3 in \cite{Bogdan2015} and Algorithm 2 in \cite{Gossmann2015}). To do this, SGS is reconsidered as a problem of the form $\min_{x} f(x) + g(x) + h(x)$, where $f$ is convex and $L_f$-smooth (differentiable with Lipschitz gradient), and both $g$ and $h$ are convex and proximal---that is, we have access to their proximal operator. The function $f$ corresponds to the loss function in Equation (\ref{eqn:sgs}) and the smoothness required is satisfied by the linear and logistic regression loss functions. The former is given by the least squares function, $\ell(b; y, \mathbf{X}) = \left\|y-\mathbf{X}b \right\|_2^2$, and the latter is described in $\S$\ref{appendix:binomial_loss_fcn}. The functions $g$ and $h$ are given by the SLOPE and gSLOPE penalties. To solve such a problem, adaptive three operator splitting (ATOS) \citep{Pedregosa2018} can be used. ATOS requires only evaluation of the gradient of $f$ and the proximal mappings of $g$ and $h$; all of which are already known for SGS.

\subsubsection{Adaptive three operator splitting (ATOS).}\label{section:proximal_algo}
First, the non-adaptive version of ATOS, three operator splitting (TOS), is described. The idea behind TOS is to introduce two auxiliary variables $y$ and $z$ and instead solve $\min_{x,y,z} f(x)+g(y)+h(z)$ subject to the constraint $x=y=z$. This allows the problem to broken down into three smaller (and simpler) sub-problems, keeping the solutions to the sub-problems as close as possible to each other. Formally, the update for step $t$ is given by
\begin{align*}
	&z_{[t]} = 	\operatorname{prox}_{h}(x_{[t]}), \\
	&y_{[t+1]} = 	\operatorname{prox}_{g}(2z_{[t]} - x_{[t]} - \gamma \nabla f(z_{[t]})), \\
	&x_{[t+1]} = x_{[t]} - z_{[t]} + y_{[t+1]},
\end{align*} 
where the subscript $[t]$ indicates the value of a variable at the $t$th iteration and $\gamma>0$ is the step-size \citep{Davis2017}. This is a generalisation of two popular splitting approaches: when $h=0$, we obtain the alternating direction method of multipliers approach, and for $g=0$, the forward-backward proximal splitting \citep{Parikh2014}. Both of these can be used to solve SLOPE and gSLOPE. 

ATOS is a modification to TOS in two ways. First, it applies an adaptive search to the step-size ($\S$\ref{section:backtracking}). Second, it reformulates the optimisation task as a saddle point problem. If we denote $h^*$ as the convex conjugate of $h$, then the optimisation problem can be written as \citep{Pedregosa2018}
\begin{align*}
	\min_x f(x) + g(x) + h(x) &= \min_x \left[ f(x) + g(x) + \max_u\{\langle x,u\rangle  - h^*(u)\} \right] \\
	&= \min_x \max_u \left[ \underbrace{f(x) + g(x) + \langle x,u \rangle - h^*(u)}_{:=L(x,u)} \right].
\end{align*} 
The problem reduces to finding the saddle point, $(x^*, u^*)$, of $L(x,u)$, where $x^*$ is the global minimum of the original optimisation problem. ATOS is given in full detail in Algorithm \ref{alg:cap}. From this, it is clear that ATOS recovers TOS by applying the transformation $x_{[t]} = b_{[t]} + \gamma_{[t]} u_{[t-1]}$ and using a constant step-size.
\begin{algorithm}[H]
	\caption{Adaptive three operator splitting for SGS}\label{alg:cap}
	\begin{algorithmic}
		\State \textbf{input:} $z_{[0]} \in \mathbb{R}^p, u_{[0]} \in \mathbb{R}^p, \gamma_{[0]}>0, \eta \in (0,1), v \in \mathbb{R}^p, w \in \mathbb{R}^m$ 
		\Repeat
		\For{$t=0,1,2,\dots$}
		\While{$f(b_{[t+1]}) > Q_t(b_{[t+1]}, \gamma_{[t]})$}\textcolor{blue}{\algorithmiccomment{Adaptive step-size search ($\S$\ref{section:backtracking})}}
		\State $b_{[t+1]} = \operatorname{prox}_{\text{SLOPE}}(z_{[t]} - \gamma_{[t]} u_{[t]} - \gamma_{[t]} \nabla f(z_{[t]}); \gamma_{[t]} v)$\textcolor{blue}{\algorithmiccomment{Proximal mapping for SLOPE}}
		\State $\gamma_{[t]} = \eta \gamma_t$\textcolor{blue}{\algorithmiccomment{Decrease step-size}}
		\EndWhile
		\State $\mathbf{m}_{[t+1]} =\mathbf{D}  b_{[t]}  + \mathbf{D}^{-1}\gamma_{[t]} u_{[t]}$
		\State $z_{[t+1]} = \operatorname{prox}_{\text{gSLOPE}}(m_{[t+1]} ; \gamma_{[t]} w)$\textcolor{blue}{\algorithmiccomment{Proximal mapping for gSLOPE}}
		\State $z_{[t+1]} = \mathbf{D}^{-1} z_{[t+1]}$
		\State $u_{[t+1]} = u_{[t]} + (b_{[t+1]} - z_{[t+1]})/\gamma_{[t]}$
		\EndFor
		\Until{$\left\|b_{[t+1]}- z_{[t]} \ \right\|_2\leq \epsilon$ or $t>t_\text{max}$ \textcolor{blue}{\algorithmiccomment{Stopping criteria}}}
		\State \textbf{output:} saddle point $(b_{[t+1]}, u_{[t+1]})$, where $b_{[t+1]} \in \mathbb{R}^p$ is the solution to SGS (Equation (\ref{eqn:sgs}).
	\end{algorithmic}
\end{algorithm}
\noindent Algorithm \ref{alg:cap} has the following parameters that can be tuned (stated with their default values):
\begin{itemize}
	\item Initial step-size, $\gamma_0$. Step-sizes are often set to 1 by default, although \cite{Pedregosa2018} recommend the following scheme instead: 1. Set $\epsilon=10^{-3}, \tilde{z}=z_{[0]}-\epsilon \nabla f(z_{[0]})$. 2. Calculate $\epsilon = 0.1\epsilon$ until $f(\tilde{z}) \leq f(z_{[0]})$. 3. Calculate $\gamma_0 = 4(f(z_{[0]}) - f(\tilde{z}_{[0]}))\left\| \nabla f(z_{[0]})\right\|^{-2}$.
	\item Backtracking parameter, $\eta$. \cite{Hastie2015} recommend $0.8$ for proximal algorithms, whilst \cite{Pedregosa2018} recommend $0.7$. The latter is used in this manuscript. 
	\item Relative accuracy (also known as tolerance), $\epsilon = 10^{-4}$.
	\item Maximum number of iterations, $t_\text{max} = 1000$.
	\item Initial values: $z_{[0]}, u_{[0]} = 0$. 
\end{itemize}

\subsubsection{Adaptive step-size search.}\label{section:backtracking}
As part of the update step, a step-size, $\gamma$, is used. A constant step-size may cause the algorithm to converge to a nonstationary point \citep{Hastie2015} and partially motivated the development of ATOS. ATOS uses an adaptive search for calculating the step-size. It works in a similar way to a backtracking line search, which is guaranteed to converge to the global optimum for convex functions \citep{Hastie2015}. To perform the search an initial step-size, $\gamma_{[0]}$, is set and a backtracking parameter, $\eta \in (0,1)$, fixed \citep{Pedregosa2018}. Then, the step-size is decreased using $\gamma_{[t]}=\eta \gamma_{[t]}$, until $f(b_{[t+1]}) > Q_t(b_{[t+1]}, \gamma_{[t]})$, where
\begin{equation}\label{eqn:backtracking}
	Q_{[t]}(x, \gamma) = f(z_{[t]}) + \langle \nabla f(z_{[t]}), x - z_{[t]} \rangle + \frac{1}{2\gamma}\left\|x-z_{[t]} \right\|_2^2.
\end{equation}
\subsubsection{gSLOPE proximal weight adjustment.}\label{section:gslope_transformation}
In the derivation of the proximal operator for gSLOPE, a transformation is applied, so that the gSLOPE proximal operator can not be used directly in the ATOS algorithm. In particular, gSLOPE is defined by the solution to the convex minimisation problem
\begin{equation} \label{eqn:gslope}
	\min_{b \in \mathbb{R}^p} \biggl\{ \frac{1}{2}\|y-  \mathbf{X}b \| _2^2 +\sum_{g=1}^{G} \lambda_g \sqrt{p_g} \| b^{(g)} \|_2 \biggr\},
\end{equation}
where $\lambda_1 \geq \dotsc \geq \lambda_G$, $\sqrt{p_1}\| b^{(1)} \|_2 \geq \dotsc \geq \sqrt{p_G}\|b^{(G)} \|_2$. In \cite{Gossmann2015} this is reformulated using the transformation $c_i = \sqrt{p_i}b_i$. In particular, let $\mathbf{D}$ be the diagonal matrix with entries $\sqrt{p_i}$, so that $c= \mathbf{D} b$. Then, Equation (\ref{eqn:gslope}) can equivalently be written as
\begin{align}
	\min_{c \in \mathbb{R}^p} \biggl\{\frac{1}{2}\| y - \mathbf{X}\mathbf{D}^{-1}c\|_2^2  + \underbrace{\sum_{g=1}^{G}\lambda_g\| c^{(g)}\|_2}_{:=f_g(c)}\biggr\}.
\end{align}
In the fitting algorithm for gSLOPE, the update step for $c$ is given by
\begin{equation}
	c_{[t+1]} = \operatorname{prox}_{f_g}(c_{[t]} - \gamma_{[t]} (\mathbf{X}\mathbf{D}^{-1})^\top (\mathbf{X}b_{[t]}-y)),
\end{equation}
where $\nabla f_g(b_{[t]}) = (\mathbf{X}\mathbf{D}^{-1})^\top (\mathbf{X}b_{[t]}-y)$ is the gradient of $f_g$ and $ \operatorname{prox}_{f_g}$ is the proximal mapping of Equation (\ref{eqn:gslope}) \citep{Gossmann2015}. The proximal mapping returns the vector $c$, instead of the desired vector $b$, and takes $c$ as input to the proximal mapping. As such, we need to apply a transformation onto the input and then undo the transformation after applying the proximal mapping. The gSLOPE update steps are altered as follows:
\begin{align}
	&z_{[t]} = 	\operatorname{prox}_\text{gSLOPE}(b_{[t]} + \gamma_{[t]} u_{[t]}; w)\longrightarrow z_{[t]} = \operatorname{prox}_\text{gSLOPE}( \mathbf{D} b_{[t]}  + \mathbf{D}^{-1}\gamma_{[t]} u_{[t]}; w),\\
	&u_{[t+1]} = u_{[t]} + (b_{[t]} - z_{[t]})/\gamma_{[t]}\;\,\:\;\;\;\;\;\longrightarrow u_{[t+1]} = u_{[t]} + (b_{[t]} -  \mathbf{D}^{-1} z_{[t]})/\gamma_{[t]}.
\end{align}
The transformation $\mathbf{D} b_{[t]}$ ensures $c_{[t]}$ is the input into the operator and the transformation $\mathbf{D}^{-1} z_{[t]}$ recovers $b_{[t]}$. Additionally, the gSLOPE transformation alters the gradient, $\nabla f_g$, to include an additional $\mathbf{D}^{-1}$ term, so the transformation $\mathbf{D}^{-1}\gamma_{[t]}$ accounts for this difference.

\section{FDR-control}\label{section:sgs_FDR}
Applying SGS to Problem \ref{problem:orthogonal_model}, guarantees of the variable and group FDR of the computed estimates $\hat{\beta}_\text{SGS}$ are sought. To do this, new penalty sequences were derived and shown to control bi-level FDR. Theorem \ref{thm:sgs_var_fdr_proof} introduces a new variable penalty sequence which is shown to control the variable FDR. Theorem \ref{thm:sgs_grp_fdr_proof} proposes a group penalty sequence which controls the group FDR (the proofs of both are given in $\S$\ref{appendix:fdr_proof}). Combined, these two penalty sequences guarantee bi-level FDR-control for SGS under orthogonal designs. The theorems are verified through simulations in $\S$\ref{section:ortho_results}.
\begin{problem}\label{problem:orthogonal_model}
	Suppose we have a linear model, $y = \mathbf{X}\beta + \epsilon$, where $\beta \in \mathbb{R}^p$ and $\epsilon \sim \mathcal{N}(0,\sigma^2 >0)$, with orthogonal input $\mathbf{X} \in \mathbb{R}^{n\times p}$. Under orthogonality, we consider the simplified model $\tilde{y} := \mathbf{X}^\top y = \beta + \epsilon$. For ease of notation, we refer to $\tilde{y}$ simply as $y$ in $\S$\ref{section:sgs_FDR} and $\S$\ref{appendix:fdr_proof}. So, $y \sim \mathcal{N}(\beta, \mathbf{I}_p)$ and $\epsilon \sim \mathcal{N}(0,\mathbf{I}_p)$. Further, suppose there exists some $m$-partition of the input space $\mathcal{G} = \{ G_1, \dots, G_m\}$ of the set $\{1,\dots, p\}$, such that $G_i \cap G_j = \emptyset$ for $i\neq j$ and $\bigcup_{i=1}^m G_i = \{1,\dots,p\}$. Find the oracle set $S = \{j: \beta_j \neq 0\}$.
\end{problem}
\subsection{FDR Theorems.}\label{section:fdr_thms}
\begin{theorem}\label{thm:sgs_var_fdr_proof}
	Suppose we apply SGS to Problem \ref{problem:orthogonal_model} using the variable weights given by
	\begin{equation}
		v_i^\text{max} = \max_{j=1,\dots,m} \left\{\frac{1}{\alpha} F_\mathcal{N}^{-1} \left(1-\frac{q_vi}{2p}\right) -   \frac{1}{3\alpha}(1-\alpha) a_j w_j\right\}, \; i=1,\dots,p,
	\end{equation}
    where $F_\mathcal{N}$ is the cumulative distribution function of a standard Gaussian distribution. We test the multiple variable hypotheses given by
	\begin{equation}
		H_i^v: \hat{\beta}_i =0,\; i=1,\dots,p,
	\end{equation}
	and define $V^v$ and $R^v$ to be the number of false and total variable rejections, given by
	\begin{align}
			V^v &= |\{i: \beta_i =0, \hat{\beta} \neq 0\}|,\\ 
			R^v &= |\{i: \hat{\beta}_i \neq 0\}|.
		\end{align}
		SGS has a variable FDR (vFDR) bounded by
		\begin{equation}
			\text{vFDR} := \mathbb{E}\left[\frac{V^v}{\max(R^v,1)}\right] \leq q_v\frac{p_0}{p},
		\end{equation}
		for a specified vFDR level $q_v \in (0,1)$, where $p_0 := |\{i:\beta_i = 0\}|$ is the number of true null hypotheses.
\end{theorem}
	\begin{theorem}\label{thm:sgs_grp_fdr_proof}
		Suppose we apply SGS to Problem \ref{problem:orthogonal_model} using the group weights given by
		\begin{equation}
			w_i^\text{max} =\max_{j=1,\dots,m}\left\{\frac{F_\text{FN}^{-1}(1-\frac{q_gi}{m})-\alpha \sum_{k \in G_j}v_k }{(1-\alpha) p_j}\right\}, \; i=1,\dots,m,
		\end{equation}
        where $F_\mathcal{N}$ is the cumulative distribution function of a folded Gaussian distribution. We test the multiple group hypotheses given by
		\begin{equation}
			H_i^g: \|\hat{\beta}^{(i)}\|_2=0, \; i= 1,\dots,m,
		\end{equation}
		and define $V^g$ and $R^g$ to be the number of false and total group rejections, given by
		\begin{align}
			V^g &= |\{i: \|\beta^{(g)}\|_2=0, \|\hat{\beta}^{(g)}\|_2 \neq 0\}|,\\ 
			R^g &= |\{i: \|\hat{\beta}^{(g)}\|_2 \neq 0\}|.
		\end{align}
		SGS has a group FDR (gFDR) bounded by
		\begin{equation}
			\text{gFDR} := \mathbb{E}\left[\frac{V^g}{\max (R^g, 1)}\right]\leq q_g \frac{m_0}{m},
		\end{equation}
		for a specified gFDR level $q_g \in (0,1)$, where $m_0 := |\{i:\|\beta^{(i)}\|_2 = 0\}|$ is the number of true null hypotheses.
	\end{theorem}
\subsection{Penalty sequences.}\label{section:sgs_penalty}
The penalty sequences for SLOPE and gSLOPE \citep{Bogdan2015,Brzyski2015} are respectively given by
\begin{align}
	&v_i = F_\mathcal{N}^{-1}(1-q_vi/2p), \; \text{for} \; i = 1,\dots,p,\label{eqn:slope_bh_seq}\\
	&w_i^\text{max} = \max_{j=1,\dots,m}\left\{\frac{1}{\sqrt{p_j}} F^{-1}_{\chi_{p_j}} (1-q_gi/m)\right\}, \; \text{for} \; i=1,\dots,m, \label{eqn:gslope_pen_max}
\end{align}
where $q_v, q_g\in(0,1)$ are the desired variable/group FDR levels, and $F_{\chi_{p_j}}$ is the cumulative distribution function of a $\chi$ distribution with $p_j$ degrees of freedom. These sequences are referred to as the SLOPE BH and gSLOPE max sequences. Both were derived under the orthogonal case to provide FDR-control. For more general settings, a modified sequence, termed the \textit{Gaussian sequence}, was derived for SLOPE in \cite{Bogdan2015}, but it reduces to the lasso when $p\gg n$ \citep{Larsson2020}, which is the primary focus of this manuscript and the sequence is not considered further. One could use these penalties for SGS (termed \textit{SGS Original}), but this would be a rather naive approach as they were not derived specifically for SGS. Indeed, applying SGS Original to orthogonal data (with the set-up from $\S$\ref{section:ortho_results}) does not achieve bi-level FDR-control (seen in Figures \ref{fig:even_sgs_org_double} and \ref{fig:uneven_sgs_org_double}). An alternative approach would be to set $\alpha = 0.5$ and $\lambda = 2$ in SGS Original, to apply both penalties in their original form (termed \textit{SGS Double}), however this was also found to be unsatisfactory, as too much penalisation is applied, so that the FDR-sensitivity trade-off is not optimised (seen in Figures \ref{fig:even_sgs_org_double} and \ref{fig:uneven_sgs_org_double}).

Considering Theorems \ref{thm:sgs_var_fdr_proof} and \ref{thm:sgs_grp_fdr_proof}, the penalty sequences which guarantee bi-level FDR-control for SGS are given by
\begin{align}
	&v_i^\text{max} = \max_{j=1,\dots,m} \left\{\frac{1}{\alpha} F_\mathcal{N}^{-1} \left(1-\frac{q_vi}{2p}\right) -   \frac{1}{3\alpha}(1-\alpha) a_j w_j\right\}, \; i=1,\dots,p, \label{eqn:sgs_var_pen_max}\\
	&w_i^\text{max} =\max_{j=1,\dots,m}\left\{\frac{F_\text{FN}^{-1}(1-\frac{q_gi}{m})-\alpha \sum_{k \in G_j}v_k }{(1-\alpha) p_j}\right\}, \; i=1,\dots,m,\label{eqn:sgs_grp_pen_max}
\end{align}
where $a_j$ is a quantity to be estimated (discussed in $\S$\ref{section:ortho_results}). A key aspect of these sequences is that they depend on each other, accommodating bi-level FDR-control. A relaxation of these penalty sequences is possible. For the gSLOPE sequence, \cite{Brzyski2015} applies a relaxation to obtain the gSLOPE mean sequence
\begin{align}
	&   w_i^\text{mean} = \overline{F}^{-1}_{\chi_{p_j}} (1-q_gi/m), \; \text{for} \; i=1,\dots,m,\label{eqn:gslope_pen_mean}\\
	&\text{where} \;\overline{F}_{\chi_{p_j}}(x):= \frac{1}{m}\sum_{j=1}^{m}F_{\chi_{p_j}}(\sqrt{p_j}x).    
\end{align}
To see how a similar relaxation for SGS is feasible, observe that in the proof for Theorem \ref{thm:sgs_var_fdr_proof} ($\S$\ref{appendix:var_proof}), Equation (\ref{eqn:strict_condition}) can be recast as
\begin{align}
	&\frac{1}{m}\sum_{j=1}^{m} \left(1-F_\mathcal{N}\left(\alpha v_i + \frac{1}{3}(1-\alpha)  a_j w_j\right)\right) \leq \frac{q_vi}{2p},\\
	&\implies \frac{1}{m}\sum_{j=1}^{m} F_\mathcal{N}\left(\alpha  v_i +\frac{1}{3}(1-\alpha) a_j w_j\right) \geq 1- \frac{q_vi}{2p}.
\end{align}
So we can pick
\begin{equation}\label{eqn:sgs_var_pen_mean}
	v_i^\text{mean} = \overline{F}_\mathcal{N}^{-1}\left(1-\frac{q_vi}{2p}\right), \; \text{where}\; \overline{F}_\mathcal{N}(x) := \frac{1}{m}\sum_{j=1}^{m} F_\mathcal{N}\left(\alpha x +  \frac{1}{3}(1-\alpha) a_j w_j\right), \; i\in \{1,\dots,p\}.
\end{equation}
Applying a similar relaxation to $w_i^\text{max}$ (Equation (\ref{eqn:sgs_grp_pen_max})) gives
\begin{equation}\label{eqn:sgs_grp_pen_mean}
	w_i^\text{mean} = \overline{F}_\text{FN}^{-1}\left(1-\frac{q_gi}{p}\right), \; \text{where}\; \overline{F}_\text{FN}(x) := \frac{1}{m}\sum_{j=1}^{m} F_\text{FN}\left((1-\alpha) p_j x + \alpha \sum_{k \in G_j} v_k\right), \; i\in \{1,\dots,m\}.
\end{equation}
The derived sequences, Equations (\ref{eqn:sgs_var_pen_max}), (\ref{eqn:sgs_grp_pen_max}), (\ref{eqn:sgs_var_pen_mean}), (\ref{eqn:sgs_grp_pen_mean}), will be referred to as the \textit{vMax}, \textit{gMax}, \textit{vMean}, and \textit{gMean} sequences, respectively. The relaxed penalty sequences are visualised in Figure \ref{fig:penalty_seq}.
\begin{figure}[H]
	\includegraphics[width=1\textwidth]{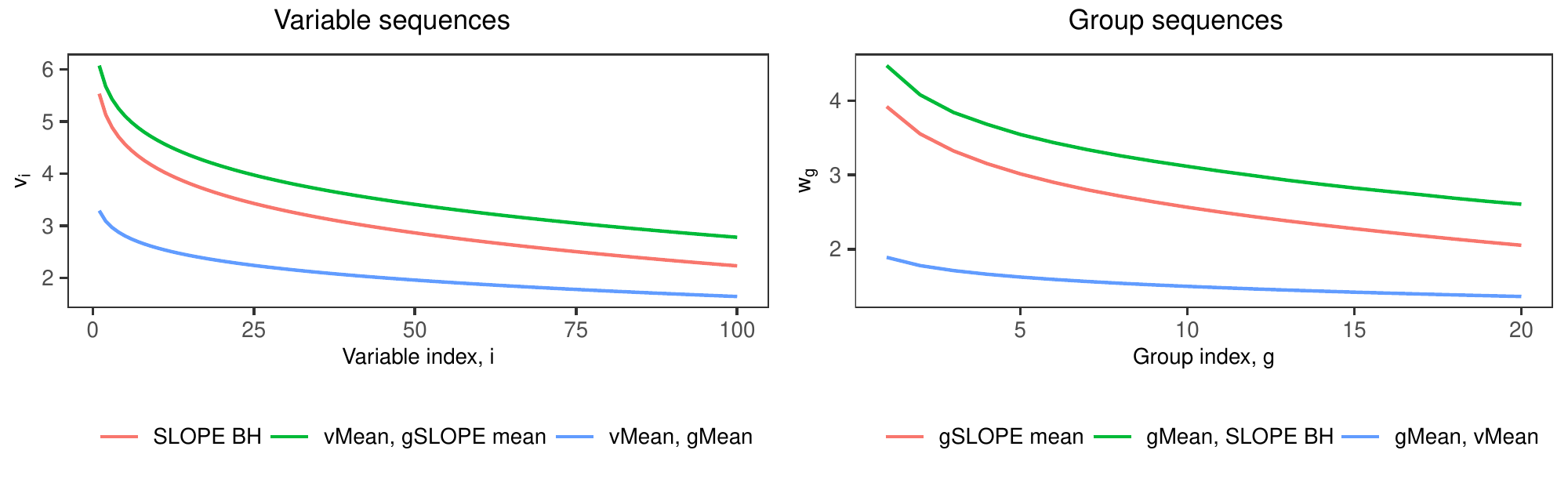}  
	\caption[width=0.8\textwidth]{Variable and group sequences shown for $m=20$ uneven groups of sizes $\{3,\dots,7\}$ with $p=100$ and $\alpha = 0.5$. The SLOPE BH, gSLOPE mean, vMean, and gMean sequences correspond to Equations (\ref{eqn:slope_bh_seq}), (\ref{eqn:gslope_pen_mean}), (\ref{eqn:sgs_var_pen_mean}) and (\ref{eqn:sgs_grp_pen_mean}), respectively.}
		\label{fig:penalty_seq}
\end{figure}

\subsection{Computational experiments.}\label{section:ortho_results}
To verify the bi-level FDR-control computationally, an orthogonal design matrix $\mathbf{X} = \boldsymbol{I}_{1000}$ was generated. Two cases were considered: even and uneven groups. In the even case, 200 groups were used, each of size 5, and for the uneven case, 40 groups of each size $\{3,\dots,7\}$ (so that there were also 200 groups in total). Within an active group, $60\%$ of the variables were randomly set to active (so, $\alpha$ was set to $0.6$ for SGS). For both cases, the variable and group sparsity proportions of the true signal varied from $1$ to $0.75$ and $0.59$, where sparsity proportion refers to the proportion of inactive variables/groups in the true model. The true effects were set to $\beta =5\delta \sqrt{2\log{p}}$, where $\delta \sim \mathcal{N}(0,1)$, because the expected value of the maximum of $p$ independent standard normal variables is approximately $\sqrt{2\log{p}}$ \citep{Cai2014}. The response was generated using the linear Gaussian model $y = \beta + \epsilon$, where $\epsilon \sim \mathcal{N}(0,1)$. In both cases, the hyperparameters $q_v, q_g =0.05,0.1,0.2$ and $\lambda = 1/n$ were used, with $1000$ Monte Carlo (MC) repetitions performed per sparsity proportion considered.
 
To apply SGS, the quantity $a_j$ in the variable sequences needs to be estimated (Equations (\ref{eqn:sgs_var_pen_max}) and (\ref{eqn:sgs_var_pen_mean})). The quantity represents the number of active variables within an active group (as shown in Theorem \ref{thm:sgs_var_fdr_proof}). A suitable estimator is given by $\hat{a}_j := \lfloor\alpha p_j\rfloor$, illustrated in Figure \ref{fig:var_max_ag_all}. For the even case, the highest sensitivity, whilst maintaining bi-level FDR-control, was achieved when $\hat{a}_j=3$ (which is $\alpha p_j = 0.6 \cdot 5$).

For the group sequences, complications arise from the quantity $\sum_{k \in G_j}v_k$ (Equations (\ref{eqn:sgs_grp_pen_max}) and (\ref{eqn:sgs_grp_pen_mean})). Whilst the variable sequence is known, we do not have prior information about the exact mappings of the penalties to the variables, and so to the groups to which the variables belong. As such, we have made an assumption that the highest ranking groups (those with the largest $\sqrt{p_j}\|\beta^{(j)}\|_2$ values) are those with the largest group size, so that they are assigned the largest variable penalties. So, for the highest ranking group (say, of size $p_j$), the variable penalty values $\{v_1,\dots,v_{p_j}\}$ are used. 
    
\subsubsection{Even groups.}
SGS achieves bi-level FDR-control using the vMax and gMax sequences (Figure \ref{fig:even_gmax_vmax}). Using the relaxed sequences, vMean and gMean, the bi-level FDR is kept close to the desired level, but FDR-control is not achieved (shown in Figure \ref{fig:sim_3_even_gmean_vmean}). The best balance between FDR and sensitivity was found using the vMean sequence with the gSLOPE mean sequence (Figure \ref{fig:even_gmax_vmax}), where it can be observed that bi-level FDR-control is obtained, even with the SGS variable relaxed sequence.

\begin{figure}[H]
\includegraphics[width=1\textwidth]{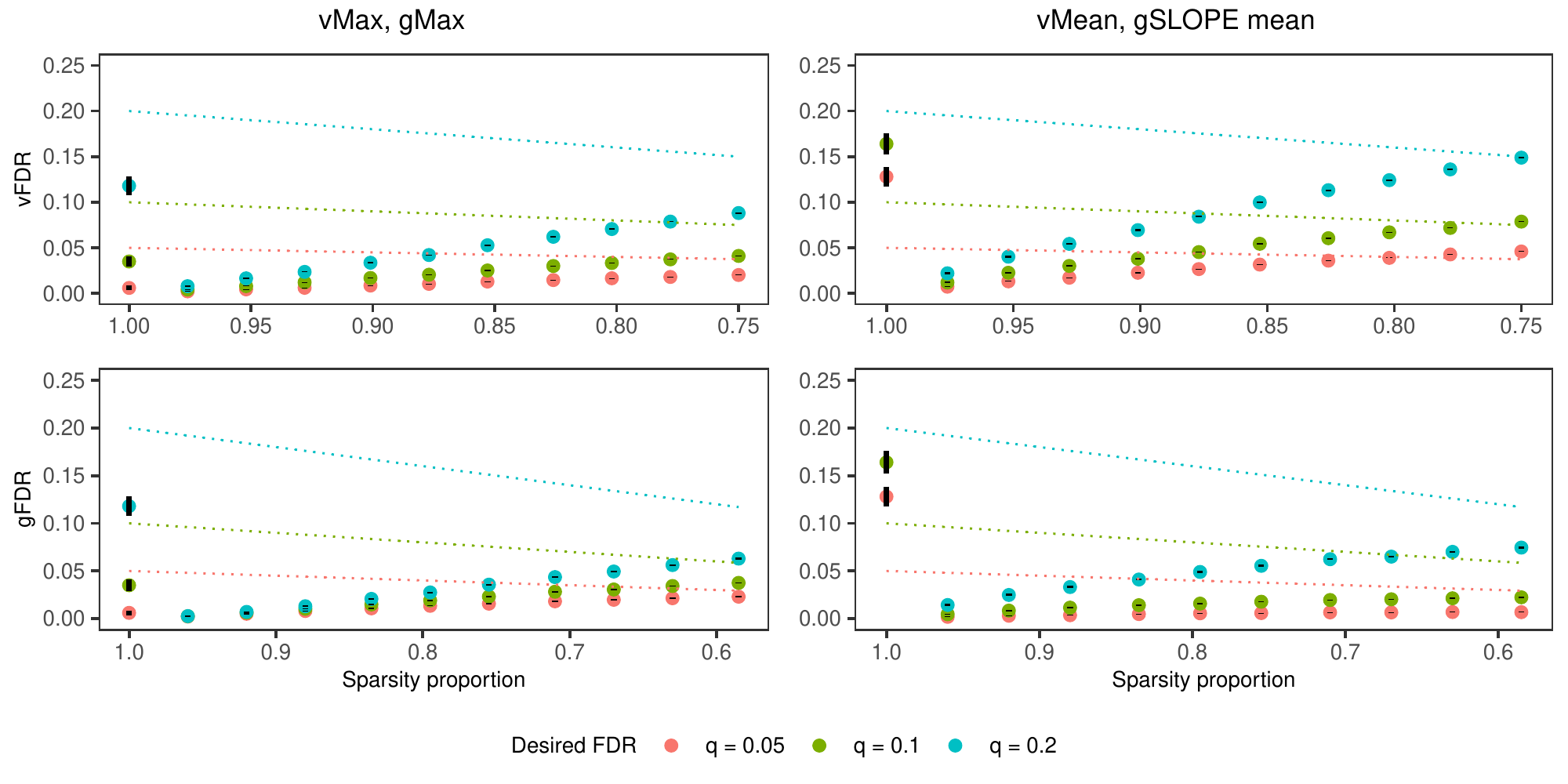}
	\caption[width=0.8\textwidth]{vFDR and gFDR shown for SGS with the vMax, gMax and the vMean, gSLOPE mean sequences under orthogonal design with even groups, as a function of decreasing sparsity proportion. 1000 MC repetitions performed per sparsity proportion. The sensitivity is given in Figure \ref{fig:sim_3_even_sgs_sens_final}.}
	\label{fig:even_gmax_vmax}
\end{figure}

\subsubsection{Uneven groups.}
The active groups were chosen to ensure the true model had a similar sparsity pattern as the even case. Under uneven groups, bi-level FDR-control is achieved using gMax and vMax penalty sequences (Figure \ref{fig:uneven_gmax_vmax}). However, using the relaxed sequences, FDR-control was again not obtained (shown in Figure \ref{fig:sim_3_uneven_gmean_vmean}). The best results came from using the vMean and gSLOPE mean sequences (Figure \ref{fig:uneven_gmax_vmax}), where bi-level FDR-control occurs. For the rest of the manuscript, SGS will use the vMean sequences for the variables and the gSLOPE mean sequences for the groups.
\begin{figure}[H]
\includegraphics[width=1\textwidth]{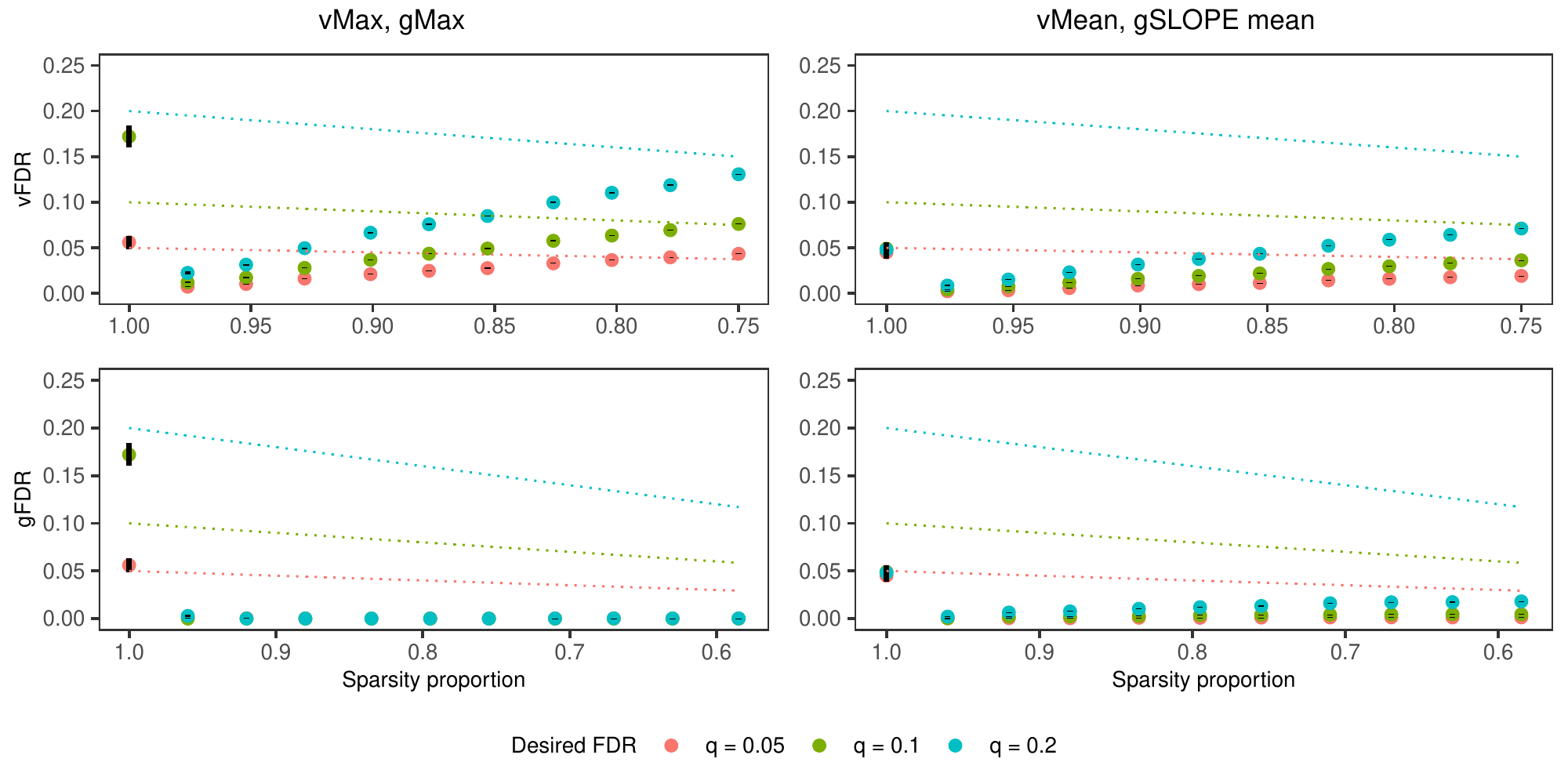}  
	\caption[width=0.8\textwidth]{vFDR and gFDR shown for SGS with the vMax, gMax and the vMean, gSLOPE mean sequences under orthogonal design with uneven groups, as a function of decreasing sparsity proportion. 1000 MC repetitions performed per sparsity proportion. The sensitivity is given in Figure \ref{fig:sim_3_uneven_sgs_sens_final}.}
		\label{fig:uneven_gmax_vmax}
\end{figure}
\section{Simulation study}\label{section:sim_studies}
An extensive simulation study was conducted to investigate the variable and group selection performance of SGS under non-orthogonal data. SGS is compared to the lasso, gLasso, SGL, SLOPE, and gSLOPE under various scenarios. First, we consider how the methods perform under a fixed signal strength, which represents an easier detection case ($\S$\ref{section:fixed_signal}). We then consider how the detection changes as the amount of sparsity in the true model decreases ($\S$\ref{section:increase_sparsity}). Further, the performance under a random signal is explored, as it is reflective of a real scenario ($\S$\ref{section:random_signal}). Of particular interest is how SGS adapts to detection under the presence of large groups, indicative of a genomics scenario, which is explored in $\S$\ref{section:large_groups}. Finally, the impact of changing the $p/n$ ratio is investigated ($\S$\ref{section:pn_ratio}) and estimates of the type I error are calculated ($\S$\ref{section:null_model}). 

\subsection{Synthetic data.}
The design matrix $\mathbf{X} \sim \mathcal{N}(0,\boldsymbol\Sigma) \in \mathbb{R}^{200\times 800}$ was used with correlation matrix $\boldsymbol\Sigma$. Three cases of within-group correlation are considered: no, medium, and high correlation, corresponding to $\rho=0,0.3,0.9$, for $\Sigma_{i,j} = \rho$, where $i\neq j$ and $i$ and $j$ belong to the same group. The response was generated using the linear model $y = \mathbf{X} \beta + \epsilon$, with Gaussian noise $\epsilon \sim \mathcal{N}(0,\sigma^2)$ and $\sigma$ chosen adaptively so that the signal-to-noise ratio was set at $6$. The variables were split into $160$ non-overlapping groups of sizes $\{3,\dots,7\}$, with variable and group sparsity proportions set to 0.95 and 0.92, and the proportion of active variables within an active group set to 0.6. For each correlation case $600$ MC repetitions were performed.

The $\text{F}_1$ score is used a primary comparison metric (defined formally in Definition \ref{defn:f1_score}), as it provides a balance between sensitivity and FDR, with a high $\text{F}_1$ score being preferable. SGS and SGL were both applied using $\alpha = 0.95$, and $q_v, q_g = 0.1$ for SGS. For each model, the data was $\ell_2$ standardised and an intercept fit, using $10$-fold cross-validation (CV) along a log-linear path of $20$ $\lambda$ values, and the 1se model was chosen\footnote{The \texttt{glmnet} \citep{Friedman2010b} \texttt{R} package was used to fit the lasso, \texttt{SLOPE} \citep{SLOPEpackage} package for SLOPE, and \texttt{SGL} \citep{SGLpackage} package for SGL. SGS and gSLOPE were fitted using the \texttt{sgs} GitHub repository.}.

\subsection{Results.}	
\subsubsection{Fixed signal.}\label{section:fixed_signal}
\vspace{-20pt}
The first scenario considered a strong signal, so that the detection was not particularly challenging. The signal strength was fixed at $\beta = 5$ for the active variables. SGS achieves a substantiality higher $\text{F}_1$ score and lower FDR score than both SLOPE and gSLOPE methods, across all correlation values (Figure \ref{fig:slope_models_non_orthog_bd_fixedsnr_summary}). Clearly, in this case, the grouping information is useful in selecting the relevant variables and groups. In general, we see improved performance as the correlation increases, especially for the group selection, as the grouping information becomes more important. As an illustration of the downside of selecting all variables within an active group when using gSLOPE, we note that the variable $\text{F}_1$ score and FDR of gSLOPE are $0.44$ and $0.69$, in comparison to $0.74$ and $0.37$ for SGS, averaged across all correlation cases.
\begin{figure}[H]
	\vspace{-5pt}
	\includegraphics[width=1\textwidth]{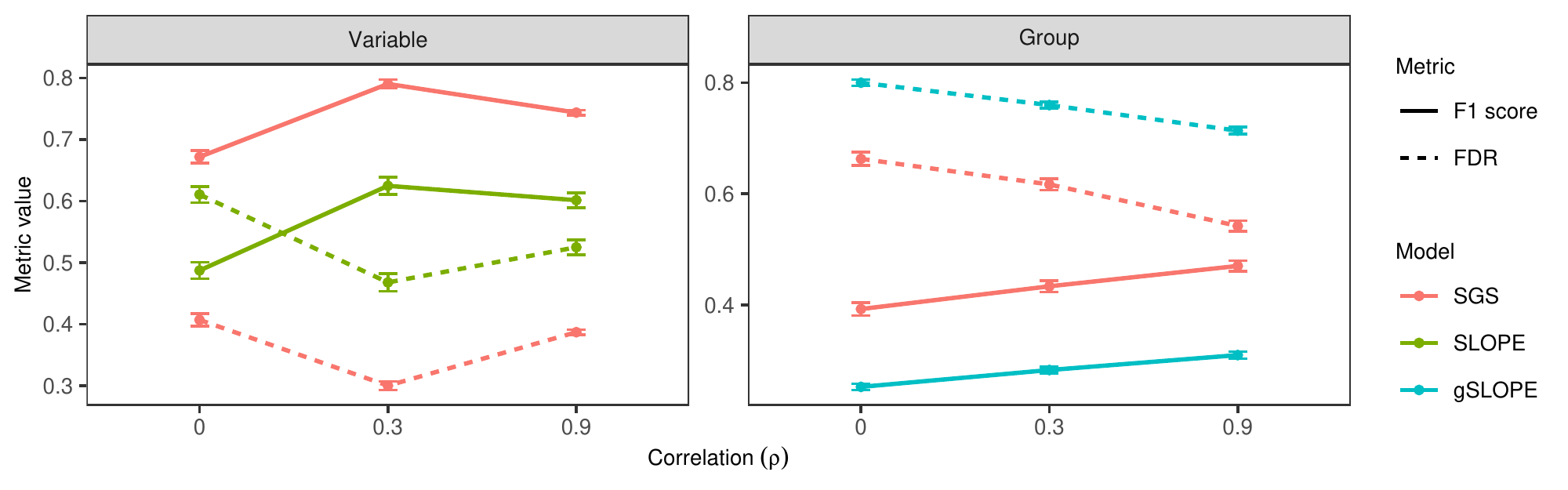}
	\vspace{-25pt}
	\caption[width=0.8\textwidth]{$\text{F}_1$ score and FDR for SLOPE-based models, shown for $\rho =0, 0.3, 0.9$, split by the type of selection, with standard errors shown. 600 MC repetitions performed per correlation case.}
	%\vspace{-10pt}
	\label{fig:slope_models_non_orthog_bd_fixedsnr_summary}
\end{figure}

SGS is further compared to the lasso and SGL, to determine whether the additional sparsity induced by SGS improves variable and group selection (Figure \ref{fig:lasso_models_non_orthog_bd_fixedsnr_summary}). SGS has an almost identical $\text{F}_1$ score to the lasso, although surprisingly higher FDR. This illustrates the downside of using CV for model selection when aiming to obtain FDR-control (discussed in Section $\S$\ref{section:model_selection}). 
\begin{figure}[H]
	\vspace{-5pt}
	\includegraphics[width=1\textwidth]{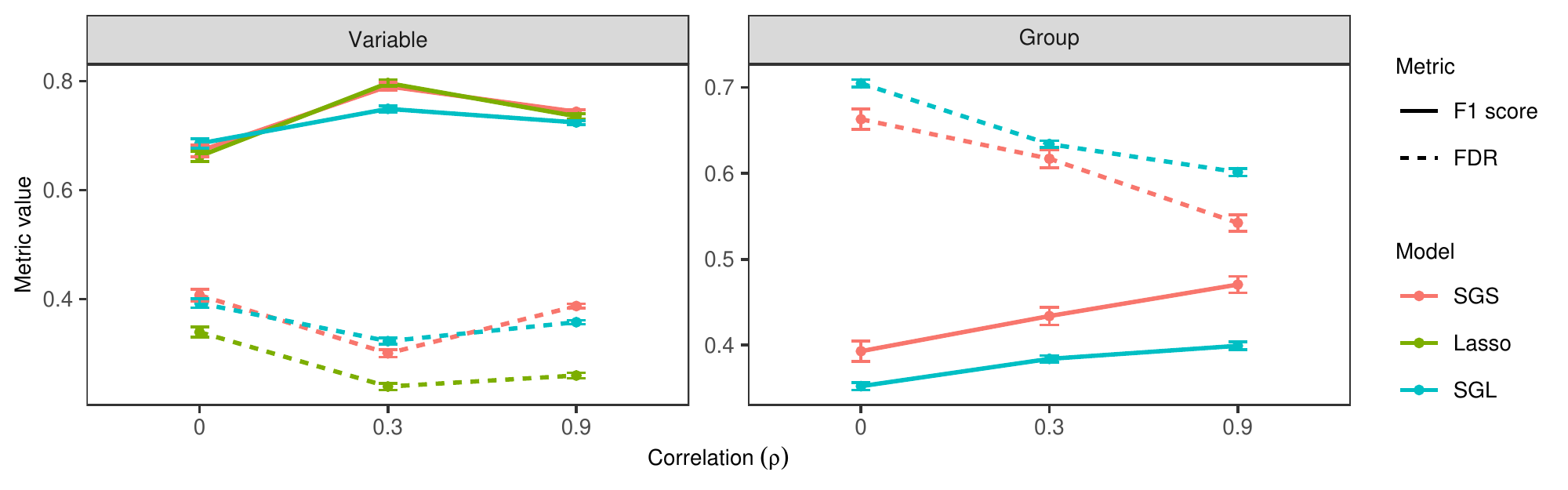}
	\vspace{-25pt}
	\caption[width=0.8\textwidth]{$\text{F}_1$ score and FDR for SGS, lasso, SGL,  shown for $\rho =0, 0.3, 0.9$, split by the type of selection, with standard errors shown. 600 MC repetitions performed per correlation case.}
	\vspace{-10pt}
\label{fig:lasso_models_non_orthog_bd_fixedsnr_summary}
\end{figure}
Comparing the two bi-level selection approaches, SGS clearly outperforms SGL, for both selection types. Interestingly, the difference in performance increases as the correlation increases for group selection, providing evidence that SLOPE-based models perform stronger under correlated designs, which is in agreement with findings presented in \cite{Zeng2014}. Averaging across all correlation cases, SGS obtains the highest mean variable $\text{F}_1$ score ($0.74\pm0.02$) and mean group $\text{F}_1$ score ($0.43\pm0.01$) of all the models considered. The full results averaged across the three correlation classes are presented in Table \ref{tbl:all_models_non_orthog_bd_fixedsnr_summary}.
\begin{table}[H]
		\centering
		
		\begin{tabular}{l|cc|ccc|ccc}
			\hline
			% & model & spar\_1 & spar\_2 & spar\_3 & spar\_4 & spar\_5 & spar\_6 & avg\_mse & avg\_mae \\ 
			\multirow{2}{*}{} & \multicolumn{2}{c}{Distance from $\beta$}&\multicolumn{3}{c}{Variable mean}&\multicolumn{3}{c}{Group mean}\\
	Model&MSE $\downarrow$&MAE $\downarrow$&$\text{F}_1$ $\uparrow$&FDR $\downarrow$&Sens. $\uparrow$&$\text{F}_1$ $\uparrow$&FDR $\downarrow$&Sens. $\uparrow$\\
			\hline
			SGS & 0.48\scriptsize\textcolor{gray}{$\pm 0.01$} & \textbf{0.16}\scriptsize\textcolor{gray}{$\pm 0.00$} & \textbf{0.74}\scriptsize\textcolor{gray}{$\pm 0.02$} & 0.37\scriptsize\textcolor{gray}{$\pm 0.01$} & 0.95\scriptsize\textcolor{gray}{$\pm 0.02$}& \textbf{0.43}\scriptsize\textcolor{gray}{$\pm 0.01$} & \textbf{0.61}\scriptsize\textcolor{gray}{$\pm 0.01$} &\textbf{0.57}\scriptsize\textcolor{gray}{$\pm 0.01$} \\ 
			SLOPE & 0.43\scriptsize\textcolor{gray}{$\pm 0.01$} & \textbf{0.16}\scriptsize\textcolor{gray}{$\pm 0.01$} & 0.57\scriptsize\textcolor{gray}{$\pm 0.01$} & 0.54\scriptsize\textcolor{gray}{$\pm 0.01$} & \textbf{0.96}\scriptsize\textcolor{gray}{$\pm 0.01$} &-&-&-\\ 
			gSLOPE & \textbf{0.35}\scriptsize\textcolor{gray}{$\pm 0.01$} & \textbf{0.16}\scriptsize\textcolor{gray}{$\pm 0.00$} &-&-&-& 0.28\scriptsize\textcolor{gray}{$\pm 0.01$} & 0.76\scriptsize\textcolor{gray}{$\pm 0.02$} & 0.54\scriptsize\textcolor{gray}{$\pm 0.01$} \\ 
			Lasso & 0.65\scriptsize\textcolor{gray}{$\pm 0.02$} & 0.17\scriptsize\textcolor{gray}{$\pm 0.00$} & 0.73\scriptsize\textcolor{gray}{$\pm 0.02$} & \textbf{0.28}\scriptsize\textcolor{gray}{$\pm 0.01$} & 0.83\scriptsize\textcolor{gray}{$\pm 0.02$} &-&-&- \\ 
			SGL & 91.0\scriptsize\textcolor{gray}{$\pm 2.14$} & 1.90\scriptsize\textcolor{gray}{$\pm 0.04$} & 0.72\scriptsize\textcolor{gray}{$\pm 0.02$} & 0.36\scriptsize\textcolor{gray}{$\pm 0.01$} & 0.91\scriptsize\textcolor{gray}{$\pm 0.02$}& 0.38\scriptsize\textcolor{gray}{$\pm 0.01$} & 0.65\scriptsize\textcolor{gray}{$\pm 0.02$} & 0.52\scriptsize\textcolor{gray}{$\pm 0.01$}  \\ 
			\hline
		\end{tabular}
		\caption{Mean squared error (MSE), mean absolute error (MAE), $\text{F}_1$ score, FDR, and sensitivity, averaged over all correlation cases, for SLOPE and Lasso-based models, with standard errors shown in grey. 1800 MC repetitions performed.} \label{tbl:all_models_non_orthog_bd_fixedsnr_summary}
	\end{table}	
\subsubsection{Decreasing sparsity.} \label{section:increase_sparsity}
Varying the sparsity proportion from the null model to a model with variable and group sparsity proportions of $0.90$ and $0.84$, grants investigation into how the performance of SGS changes as a function of the sparsity proportion. Under such a scenario, the $\text{F}_1$ score drops as the sparsity proportion decreases, with the FDR in turn increasing (Figure \ref{fig:slope_models_non_orthog_bd_fixedsnr}). This decrease is slowed as correlation increases, but is still present. This pattern is present for all of the models. The results are unsurprising, as decreasing the sparsity in the underlying model means there are more true signals for the models to detect, which generally means that obtaining a higher $\text{F}_1$ score is more challenging.

A limitation of the lasso is that it can select at most $n$ predictors \citep{Zou2005}, which is not a limitation for the SLOPE-based models. However, as this case illustrates, once the underlying true model is no longer strongly sparse, the performance of the SLOPE-based models drops. Therefore, these methods are probably not suitable for such cases, rendering this limitation of the lasso as relatively insignificant in comparison to the SLOPE-based models.

\begin{figure}[H]
	\vspace{-5pt}
	\includegraphics[width=1\textwidth]{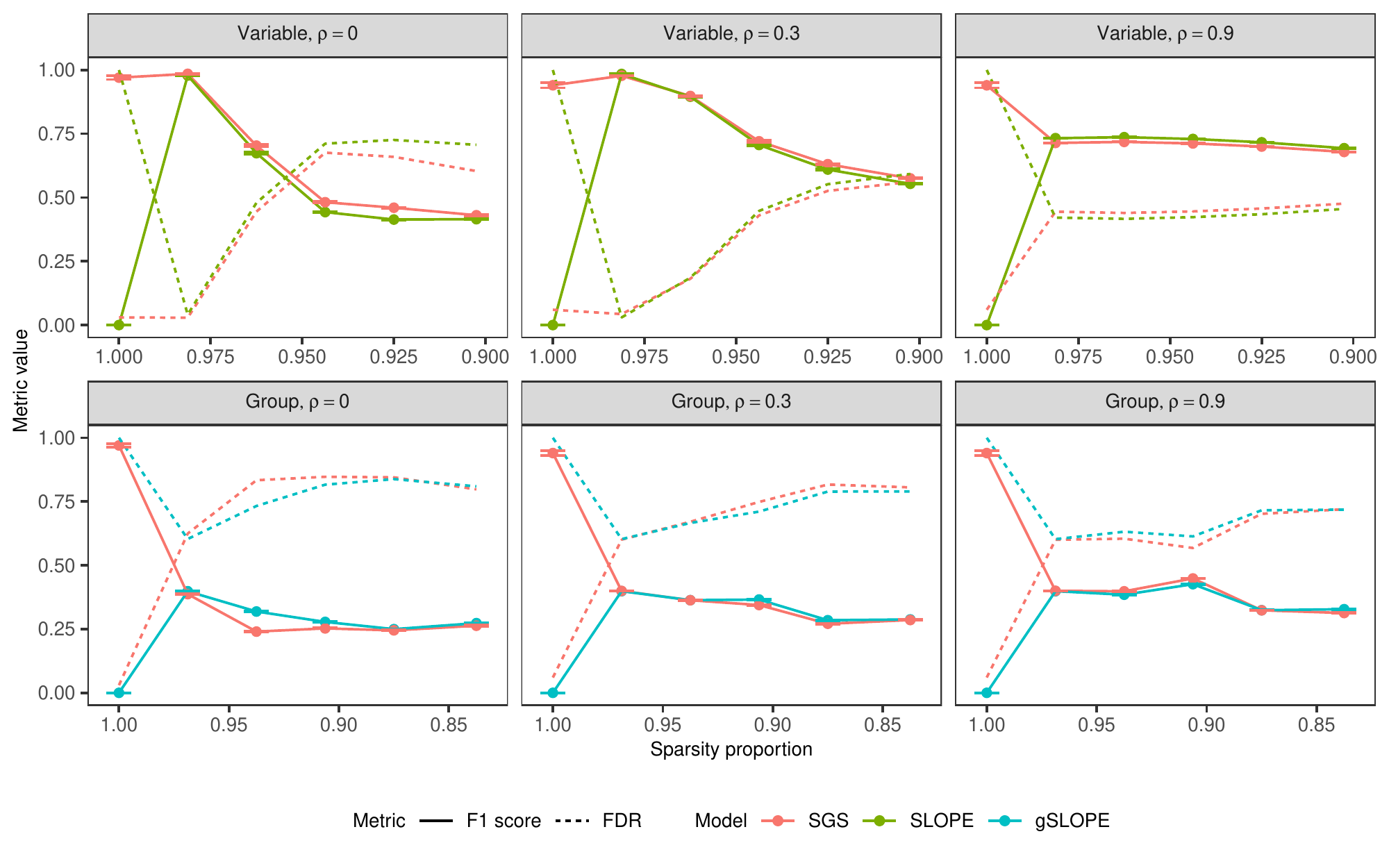}
	\vspace{-25pt}
	\caption[width=0.8\textwidth]{$\text{F}_1$ score and FDR shown as a function of decreasing sparsity proportion, for the SLOPE-based models. This is shown for the different correlation cases and split by the type of selection, with standard errors shown. 100 MC repetitions performed per sparsity proportion and correlation case. The sensitivity is shown in Figure \ref{fig:slope_models_non_orthog_bd_fixedsnr_sens}.}
	\vspace{-10pt}
	\label{fig:slope_models_non_orthog_bd_fixedsnr}
\end{figure}

\subsubsection{Random signal.}\label{section:random_signal}
So far, a fixed signal $\beta = 5$ and relatively high sparsity in the underlying model were used, and we have observed strong performance for SGS under such conditions. Here, a more realistic case of a random signal, $\beta \sim \mathcal{N}(0,5^2)$, and a lower average variable and group sparsity of $0.88$ and $0.80$ are explored. This case was designed to be more challenging, as the signal is weaker. A clear drop-off in performance in terms of the $\text{F}_1$ score can be observed, for all models, as is to be expected (Table \ref{tbl:all_models_non_orthog_gg_fixedsnr_summary}). SGS again has the highest $\text{F}_1$ score amongst all models, for both types of selections. The lasso has the lowest FDR and sensitivity. Interestingly, the SLOPE-based models have lower MSE than the lasso-based ones, in particular compared to SGL, which suffers from inflated $\hat{\beta}$ estimates. As the SLOPE-based models apply stronger penalisation, and therefore introduce additional bias, one would expect this trend to go the other way. 

We also take this case to illustrate the difference between SGS Original and SGS using the sequences derived in $\S$\ref{section:sgs_penalty}. We find that SGS with the derived sequences obtains far superior performance than the naive application of SGS Original, highlighting the importance of using theory to extract the full performance of SGS.

\begin{table}[H]
\centering
\begin{tabular}{l|cc|ccc|ccc}
	\hline
	% & model & spar\_1 & spar\_2 & spar\_3 & spar\_4 & spar\_5 & spar\_6 & avg\_mse & avg\_mae \\ 
	\multirow{2}{*}{} & \multicolumn{2}{c}{Distance from $\beta$}&\multicolumn{3}{c}{Variable mean}&\multicolumn{3}{c}{Group mean}\\
	Model&MSE $\downarrow$&MAE $\downarrow$&$\text{F}_1$ $\uparrow$&FDR $\downarrow$&Sens. $\uparrow$&$\text{F}_1$ $\uparrow$&FDR $\downarrow$&Sens. $\uparrow$\\
	\hline
	SGS & 1.93\scriptsize\textcolor{gray}{$\pm 0.04$} & \textbf{0.43}\scriptsize\textcolor{gray}{$\pm 0.01$} & \textbf{0.55}\scriptsize\textcolor{gray}{$\pm 0.00$} & 0.46\scriptsize\textcolor{gray}{$\pm 0.00$} & 0.58\scriptsize\textcolor{gray}{$\pm 0.01$} & \textbf{0.46}\scriptsize\textcolor{gray}{$\pm 0.01$} & \textbf{0.58}\scriptsize\textcolor{gray}{$\pm 0.01$} & \textbf{0.56}\scriptsize\textcolor{gray}{$\pm 0.01$}\\ 
	SLOPE & \textbf{1.77}\scriptsize\textcolor{gray}{$\pm 0.04$} & \textbf{0.43}\scriptsize\textcolor{gray}{$\pm 0.01$} & 0.39\scriptsize\textcolor{gray}{$\pm 0.00$} & 0.64\scriptsize\textcolor{gray}{$\pm 0.01$} & \textbf{0.61}\scriptsize\textcolor{gray}{$\pm 0.00$}&-&-&- \\ 
	gSLOPE & 2.03\scriptsize\textcolor{gray}{$\pm 0.04$}& 0.48\scriptsize\textcolor{gray}{$\pm 0.01$}   &-&-&-  & 0.30\scriptsize\textcolor{gray}{$\pm 0.00$} & 0.74\scriptsize\textcolor{gray}{$\pm 0.00$} & 0.55\scriptsize\textcolor{gray}{$\pm 0.01$} \\ 
	Lasso & 2.03\scriptsize\textcolor{gray}{$\pm 0.04$} & 0.41\scriptsize\textcolor{gray}{$\pm 0.01$} & 0.45\scriptsize\textcolor{gray}{$\pm 0.00$} & \textbf{0.36}\scriptsize\textcolor{gray}{$\pm 0.00$} & 0.43\scriptsize\textcolor{gray}{$\pm 0.01$} &-&-&-\\ 
	SGL & 145\scriptsize\textcolor{gray}{$\pm 2.16$} & 2.60\scriptsize\textcolor{gray}{$\pm 0.04$} & 0.43\scriptsize\textcolor{gray}{$\pm 0.00$} & 0.51\scriptsize\textcolor{gray}{$\pm 0.00$} & 0.47\scriptsize\textcolor{gray}{$\pm 0.01$} & 0.38\scriptsize\textcolor{gray}{$\pm 0.00$} & 0.63\scriptsize\textcolor{gray}{$\pm 0.00$} & 0.49\scriptsize\textcolor{gray}{$\pm 0.01$} \\ 
	SGS Original & 3.28\scriptsize\textcolor{gray}{$\pm 0.05$} & 0.55\scriptsize\textcolor{gray}{$\pm 0.01$} & 0.25\scriptsize\textcolor{gray}{$\pm 0.01$} & 0.70\scriptsize\textcolor{gray}{$\pm 0.01$} & 0.23\scriptsize\textcolor{gray}{$\pm 0.01$} & 0.36\scriptsize\textcolor{gray}{$\pm 0.01$} &\textbf{0.58}\scriptsize\textcolor{gray}{$\pm 0.01$}& 0.35\scriptsize\textcolor{gray}{$\pm 0.01$} \\ 
	\hline
\end{tabular}
\caption{Mean squared error (MSE), mean absolute error (MAE), $\text{F}_1$ score, FDR, and sensitivity, averaged over all correlation cases, for SLOPE and Lasso-based models, with standard errors shown in grey. 1800 MC repetitions performed.}
\label{tbl:all_models_non_orthog_gg_fixedsnr_summary}
\end{table}	

\subsubsection{Larger groups.}\label{section:large_groups}
To gain an indication of how SGS performs for larger groups, $25$ groups of sizes $\{5,\dots,75\}$ were generated. The number of active groups was varied from the null model to $4$, and the proportion of active variables within an active group was randomly sampled from $\mathcal{U}[0.2,0.6]$; otherwise the set-up remained as described earlier (Section $\S$\ref{section:fixed_signal}). The results are shown in Table \ref{tbl:all_models_non_orthog_lg_fixedsnr_summary}. In comparison to Figure \ref{fig:slope_models_non_orthog_bd_fixedsnr_summary}, we observe a drop in the variable $\text{F}_1$ score for SGS, but a large increase in the group score. As there were fewer groups present, a false group discovery was less likely, leading to lower group FDR. Here, SGS obtains a higher $\text{F}_1$ score than the other two models, as well as lower FDRs, giving us confidence that SGS is well suited to work with datasets with large group sizes. %This is explored further with real genetics data in $\S$\ref{section:real_data}.
	\begin{table}[H]
		\centering
		\begin{tabular}{l|cc|ccc|ccc}
			\hline
			% & model & spar\_1 & spar\_2 & spar\_3 & spar\_4 & spar\_5 & spar\_6 & avg\_mse & avg\_mae \\ 
			\multirow{2}{*}{} & \multicolumn{2}{c}{Distance from $\beta$}&\multicolumn{3}{c}{Variable mean}&\multicolumn{3}{c}{Group mean}\\
		Model&MSE $\downarrow$&MAE $\downarrow$&$\text{F}_1$ $\uparrow$&FDR $\downarrow$&Sens. $\uparrow$&$\text{F}_1$ $\uparrow$&FDR $\downarrow$&Sens. $\uparrow$\\
			\hline
			SGS & 0.27\scriptsize\textcolor{gray}{$\pm 0.01$} & 0.11\scriptsize\textcolor{gray}{$\pm 0.00$} & \textbf{0.59}\scriptsize\textcolor{gray}{$\pm 0.01$} & \textbf{0.51}\scriptsize\textcolor{gray}{$\pm 0.01$} & \textbf{0.99}\scriptsize\textcolor{gray}{$\pm 0.00$}& \textbf{0.71}\scriptsize\textcolor{gray}{$\pm 0.01$} & \textbf{0.30}\scriptsize\textcolor{gray}{$\pm 0.01$} & 0.98\scriptsize\textcolor{gray}{$\pm 0.00$} \\ 
			SLOPE & 0.26\scriptsize\textcolor{gray}{$\pm 0.01$} & \textbf{0.10}\scriptsize\textcolor{gray}{$\pm 0.00$} & 0.56\scriptsize\textcolor{gray}{$\pm 0.01$} & 0.55\scriptsize\textcolor{gray}{$\pm 0.01$} & \textbf{0.99}\scriptsize\textcolor{gray}{$\pm 0.00$}&-&-&- \\ 
			gSLOPE & \textbf{0.19}\scriptsize\textcolor{gray}{$\pm 0.01$} & \textbf{0.10}\scriptsize\textcolor{gray}{$\pm 0.00$}&-&-&-& 0.69\scriptsize\textcolor{gray}{$\pm 0.01$} & 0.35\scriptsize\textcolor{gray}{$\pm 0.01$} & \textbf{1.00}\scriptsize\textcolor{gray}{$\pm 0.00$} \\ 
			\hline
		\end{tabular}
		\caption{Mean squared error (MSE), mean absolute error (MAE), $\text{F}_1$ score, FDR, and sensitivity, averaged over all correlation cases, for SLOPE-based models, with standard errors shown in grey. 1800 MC repetitions performed.}
		\label{tbl:all_models_non_orthog_lg_fixedsnr_summary}
	\end{table}	
 
\subsubsection{Decreasing the p/n ratio.}\label{section:pn_ratio}
In the simulation studies considered so far, $p$ and $n$ have both been set to give a $p/n$ ratio of $4$. Here, the performance of SGS is explored as this ratio decreases to $1$, which reflects the scenario of obtaining more observations. The variable/group sparsity proportions were set to $0.94$ and $0.9$ respectively. Figure \ref{fig:slope_models_non_orthog_bd_in_fixedsnr} shows how for no correlation, the $\text{F}_1$ score increases linearly as the ratio decreases, whilst the FDR quickly decreases. The increase is apparent, but less dramatic, for $\rho = 0.3$. Under high correlation ($\rho = 0.9$), the $\text{F}_1$ score stagnates as the ratio decreases. The stagnation of the score under high correlation is similar to the trend seen in Figure \ref{fig:slope_models_non_orthog_bd_fixedsnr}, where the $\text{F}_1$ score stays the same under decreasing sparsity proportion. In terms of model performance, SGS tends to have stronger performance at a higher $p/n$ ratio, in comparison to SLOPE and gSLOPE, but the gap decreases with the ratio, showing that SGS provides a clear advantage when there are less observations available.

\begin{figure}[H]
	\vspace{-5pt}
	\includegraphics[width=1\textwidth]{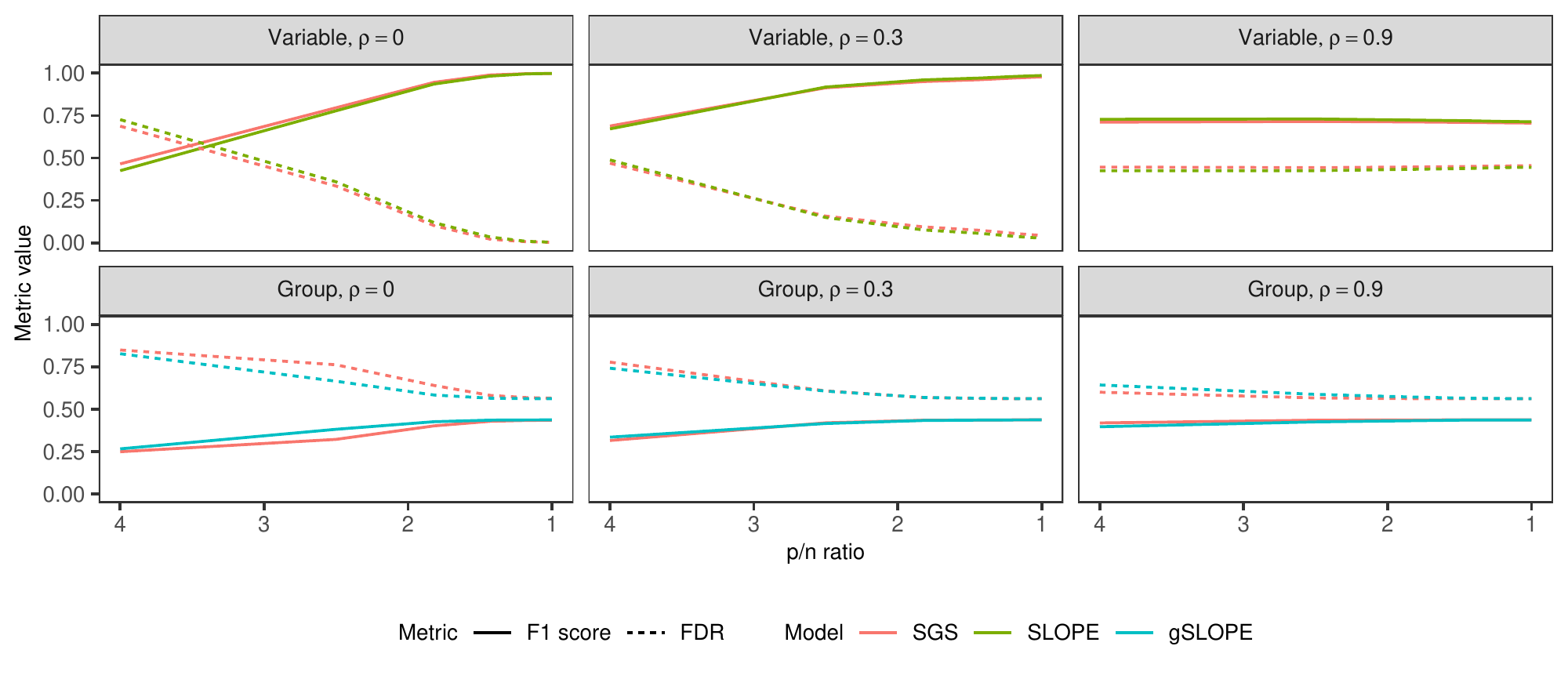}
	\vspace{-25pt}
	\caption[width=0.8\textwidth]{$\text{F}_1$ score and FDR shown as a function of decreasing $p/n$ ratio, for the SLOPE-based models. This is shown for the different correlation cases and split by the type of selection. The sensitivity is shown in Figure \ref{fig:slope_models_non_orthog_bd_in_fixedsnr_sens}. 100 MC repetitions performed per $p/n$ ratio and correlation case.}
	\vspace{-10pt}
	\label{fig:slope_models_non_orthog_bd_in_fixedsnr}
\end{figure}
\newpage
\subsubsection{Detection under the null model.}\label{section:null_model}
\begin{wraptable}{r}{7.3cm}
		\vspace{0pt}
		\begin{tabular}{l|cc}
			\hline
			% & model & spar\_1 & spar\_2 & spar\_3 & spar\_4 & spar\_5 & spar\_6 & avg\_mse & avg\_mae \\ 
			%\multirow{2}{*}{Model \ $\rho$}
			
	    & \thead{Type I\\ error rate $\downarrow$} &\thead{Mean number\\ selected $\downarrow$} \\
			\hline
			SGS &$9\times 10^{-4}$\scriptsize\textcolor{gray}{$\pm  2 \times 10^{-6}$}&$0.51$\scriptsize\textcolor{gray}{$\pm  0.20$} \\ 
			SLOPE & $0.03$\scriptsize\textcolor{gray}{$\pm 1 \times 10^{-5}$} & $23.2$\scriptsize\textcolor{gray}{$\pm 1.54$}\\ 
			gSLOPE &$0.20$\scriptsize\textcolor{gray}{$\pm 2 \times 10^{-5}$}&$208$\scriptsize\textcolor{gray}{$\pm 3.60$}\\ 
			Lasso &$\mathbf{5\times 10^{-4}}$\scriptsize\textcolor{gray}{$\pm 1 \times 10^{-6}$} & \textbf{0.40}\scriptsize\textcolor{gray}{$\pm 0.06$}\\
			SGL & $0.01$\scriptsize\textcolor{gray}{$\pm 1\times 10^{-5}$} & $8.33$\scriptsize\textcolor{gray}{$\pm 1.01 $}  \\
			\hline
		\end{tabular}
		\caption{Type I error rate and the mean number of selected variables, with standard errors shown in grey. 300 MC repetitions performed.}
		\label{tbl:models_under_null}
	\end{wraptable}
 In $\S$\ref{section:increase_sparsity}, the SLOPE-based models were applied under the null model. This gives insight into whether the approaches detect signal when none is present, allowing for calculation of the type I error rate. The lasso-based models were further applied to the null case and the results were averaged over the three correlation cases (Table \ref{tbl:models_under_null}). SGS and the lasso have the lowest type I errror rate. This case illustrates the downside of applying only group sparsity in gSLOPE, as the method had the highest rate, selecting all variables in a group as active, leading to a large number of inactive variables being selected as false positives. 
\section{Model selection}\label{section:model_selection}
In most regularisation approaches, including SGS, the tuning parameter $\lambda$ controls the level of sparsity in the fitted model and can also be seen to be proportional to the noise level of the underlying data-generating process \citep{Sun2012}. In most situations this is an unknown quantity. As shown in $\S$\ref{section:ortho_results} for orthogonal designs, the choice of $\lambda=1$ gives bi-level FDR-control. For non-orthogonal designs this quantity needs to be estimated. This section presents two common approaches for estimating the tuning parameter. The first set of approaches describe how models can be generated by fitting across a path of $\lambda$ values. The second set of approaches presented illustrate how to simultaneously estimate the noise and the coefficients. The section ends with a comparison of the performance of the different approaches. 
	
\subsection{Model selection on a path.}\label{section:model_selection_path}
By fitting models for a path of $\lambda$ values, a pathwise solution is created. This raises the question of which model to pick along the path, as two objectives can be model discovery and predictive performance. These two objectives are known to be in conflict with one another and may not lead to the same choice of tuning parameter \citep{Leng2006, Yang2005}. There is no clear consensus on which approach to use to discriminate between models on a path. In genetics, the type I error is often desirable to use as a discrimination tool, but there are no finite sample guarantees for type I errors with current model selection strategies \citep{Bogdan2015}. In general, CV is the most widely used \citep{Freijeiro-Gonzalez2022}, with the optimum model chosen as the model with the largest value of $\lambda$ such that the mean-squared error is within one standard error of the minimum error (also known as the 1se model).
	
However, whilst CV may pick the best predictive model, it does not aim for FDR-control and can potentially introduce bias \citep{Moscovich2019}. The lack of FDR-control is confirmed through our experimental results (Figure \ref{fig:fdr_control_slope_vs_sgs}). SLOPE and SGS were applied using CV to a simulated dataset, varying the FDR parameter $q$ (for SGS, $q_v = q_g = q$). Along the path for SGS, the models are able to achieve vFDR levels close to the desired level for most choices of $q$, but the chosen CV models tend to have amongst the highest vFDR levels.

This raises two questions: 1. does the true model exist on the path (also known as path consistency \citep{Hastie2015})? and 2. how is the tuning parameter picked to give the desired FDR? To achieve FDR-control, we require use of an additional method to work in conjunction SGS for non-orthogonal designs. A number of such methods have been proposed in the literature including post-inference, model selection, and variable selection approaches. One variable selection approach proposed in the literature is Knockoff. \cite{Barber2015} introduced \textit{Knockoff} as a FDR-controlling variable selection approach that can be used alongside high-dimensional regression approaches. Knockoff introduces pseudovariables, called knockoff variables, into the fitting process. The number of knockoff variables selected provides an estimate for the number of false positives. The initial version of Knockoff was shown to attain exact FDR-control but works only when $n>p$. \cite{Barber2019} extended Knockoff for the high-dimensional setting.

\begin{figure}[H]
	\vspace{-5pt}
	\includegraphics[width=.9\textwidth]{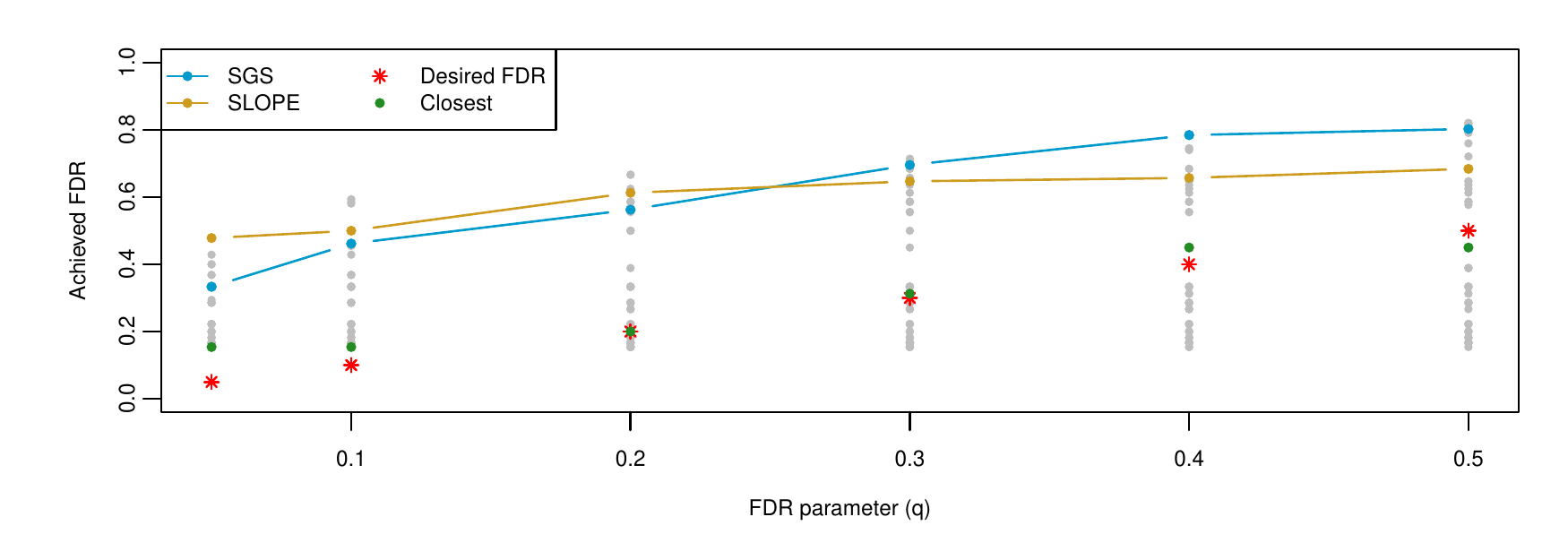}  
	\vspace{-15pt}
	\caption[width=0.7\textwidth]{vFDR levels achieved for SLOPE and SGS. The lines show the vFDR levels for the 1se CV models. Each grey dot represents the vFDR for a $\lambda$ value along the path of SGS models, with the green dots representing the value closest to the desired level. The design matrix used was i.i.d $\mathcal{N}(0,1)$ with $n=100$, $p=250$, $m=50$, and within-group correlation of $\rho = 0.3$.}
	\vspace{-10pt}
	\label{fig:fdr_control_slope_vs_sgs}
\end{figure}

\subsection{Estimating the noise.}
Alternative model selection approaches involve estimating the tuning parameter directly. One such approach comes from using scaled sparse regression, which jointly estimate the coefficients and noise \citep{Sun2012}. When $n>p$, this is easily done using unbiased estimators. However, when $p\geq n$, iterative procedures are required. An example of such a procedure for the lasso is the \textit{scaled lasso} \citep{Sun2012}, which iteratively estimates the noise using the mean residual square and scales the tuning parameter in proportion to the estimated noise. This procedure was adapted to SLOPE in Algorithm 5 in \cite{Bogdan2015}. We further adapt it here for our proposed SGS method by calculating $\hat{\beta}$ using SGS instead of SLOPE. The approach is named \textit{scaled SGS}.

%A related idea comes from noticing that $\lambda$ is assumed to be 1 without loss of generality in Theorems \ref{thm:sgs_var_fdr_proof} and \ref{thm:sgs_grp_fdr_proof} when deriving the penalty sequences. If we instead make no such assumption, 

Without loss of generality for Theorems \ref{thm:sgs_var_fdr_proof} and \ref{thm:sgs_grp_fdr_proof}, the assumption of $\lambda=1$ was made. If no such assumption is made, the penalty sequences are derived as:
\begin{align}
	&v_i^\text{max}(\lambda) = \max_{j=1,\dots,m} \left\{\frac{F^{-1}_\mathcal{N} \left(1-\frac{q_vi}{2p}\right) -   \frac{1}{3}(1-\alpha) \lambda a_j w_j}{\alpha \lambda}\right\}, \; i= 1,\dots,p, \label{eqn:sgs_var_pen_max_2}\\
	&w_i^\text{max}(\lambda) =\max_{j=1,\dots,m}\left\{\frac{F^{-1}_\text{FN}(1-\frac{q_gi}{m})-\alpha \lambda \sum_{k \in G_j}v_k }{(1-\alpha)\lambda p_j}\right\}, \; i = 1,\dots,m.\label{eqn:sgs_grp_pen_max_2}
\end{align}
Hence, an update of $\lambda$ would result in an adaptive update of the penalty sequences. This can be incorporated into an iterative procedure as described in Algorithm \ref{alg:noise_est} and is named \textit{adaptively scaled SGS} (\textit{AS-SGS}). An interesting consequence of this proposed noise estimation approach is that it is agnostic to the choice of $\alpha$. By applying AS-SGS to the simulation set-up from $\S$\ref{section:fixed_signal}, we observed that the solutions produced along a path of $\alpha$ values were all identical to each other. This property does not hold for scaled SGS/SLOPE.

\begin{algorithm}[H]
	\caption{Adaptively scaled SGS (AS-SGS)}\label{alg:noise_est}
	\begin{algorithmic}
		\State \textbf{input:} $y, \mathbf{X}.$ 
		\State Set $\hat{S}_+ = \emptyset.$
		\Repeat
		\State Set $\hat{S} = \hat{S}_+.$
		\State \multiline{Set $\hat{\lambda} = \text{RSS}/(n-|\hat{S}|-1)$, where RSS (residual sum of squares) is calculated using a linear model with $y$ and $\mathbf{X}$ restricted to the variables in $\hat{S}$.}
		\State Generate $v^\text{max}(\hat{\lambda})$ and $w^\text{max}(\hat{\lambda})$ with $\hat{\lambda}$ using Equations (\ref{eqn:sgs_var_pen_max_2}) and (\ref{eqn:sgs_grp_pen_max_2}).
		\State Compute $\hat{\beta}$ using SGS (Equation (\ref{eqn:sgs})) with $\hat{\lambda}$, $v^\text{max}(\hat{\lambda})$, and $w^\text{max}(\hat{\lambda}).$
		\State Set $\hat{S}_+=\{j: \hat{\beta}_j \neq 0\}$.
		\Until{$\hat{S}_+ = \hat{S}$.}
		\State \textbf{output:}  $\hat{S}, \lambda_\text{est}, v^\text{max}, w^\text{max}$.
	\end{algorithmic}
\end{algorithm}

\subsection{Comparing model selection approaches.}\label{section:model_selection_results}
Four different approaches for tackling the model selection task have been discussed: CV, Knockoff, scaled SGS, and AS-SGS. The performance of these approaches is investigated using synthetic data generated by the set-up described in $\S$\ref{section:fixed_signal}. The approaches all worked in conjunction with SGS and the FDR-control parameters were set to $q_v = q_g = 0.1$. The $\text{F}_1$ score is used as the primary comparison metric.

 Whilst AS-SGS is a definite improvement over scaled SGS in terms of selection, CV still produces higher $\text{F}_1$ scores for both types of selection, so is best for general selection (Table \ref{tbl:models_selection_results}). However, AS-SGS produces estimates closer to the true $\beta$ values than the other approaches, whilst scaled SGS achieves the best FDR-control. The results provide useful information for a practitioner wishing to apply such methods, as different methods perform better for different metrics, but also illustrate the general need for further development of model selection approaches, as none of the approaches considered were able to obtain an FDR level below the set threshold of 0.1. %Some promising alternative approaches include stability selection (\citep{Meinshausen2010,Ahmed2011}), ET-Lasso (\citep{Yang2018}), and complementary pairs stability selection (\citep{Shah2013}). 
	\begin{table}[H]
	\centering
	\begin{tabular}{l|cc|ccc|ccc}
		\hline
		% & model & spar\_1 & spar\_2 & spar\_3 & spar\_4 & spar\_5 & spar\_6 & avg\_mse & avg\_mae \\ 
		\multirow{2}{*}{} & \multicolumn{2}{c}{Distance from $\beta$}&\multicolumn{3}{c}{Variable mean}&\multicolumn{3}{c}{Group mean}\\
		Model& MSE $\downarrow$& MAE $\downarrow$ & $\text{F}_1$ $\uparrow$ & FDR $\downarrow$& Sens. $\uparrow$& $\text{F}_1$ $\uparrow$ & FDR $\downarrow$& Sens. $\uparrow$\\
		\hline
		CV &$0.48$\scriptsize\textcolor{gray}{$\pm 0.01$} & $0.16$\scriptsize\textcolor{gray}{$\pm 0.00$}&$\textbf{0.74}$\scriptsize\textcolor{gray}{$\pm 0.00$}&$0.37$\scriptsize\textcolor{gray}{$\pm 0.01$}&$0.95$\scriptsize\textcolor{gray}{$\pm 0.00$}& $\textbf{0.43}$\scriptsize\textcolor{gray}{$\pm 0.01$}&$0.61$\scriptsize\textcolor{gray}{$\pm 0.01$}&$0.57$\scriptsize\textcolor{gray}{$\pm 0.01$} \\
		AS-SGS &$\textbf{0.28}$\scriptsize\textcolor{gray}{$\pm 0.01$}&$\textbf{0.13}$\scriptsize\textcolor{gray}{$\pm 0.00$} &$0.58$\scriptsize\textcolor{gray}{$\pm 0.01$} & $0.53$\scriptsize\textcolor{gray}{$\pm 0.01$}& $\textbf{0.98}$\scriptsize\textcolor{gray}{$\pm 0.00$}&$0.41$\scriptsize\textcolor{gray}{$\pm 0.01$}&$0.65$\scriptsize\textcolor{gray}{$\pm 0.01$}&$\textbf{0.67}$\scriptsize\textcolor{gray}{$\pm 0.01$} \\
		Scaled SGS &$0.75$\scriptsize\textcolor{gray}{$\pm 0.02$}&$0.18$\scriptsize\textcolor{gray}{$\pm 0.00$}&$0.42$\scriptsize\textcolor{gray}{$\pm 0.01$}&$\textbf{0.21}$\scriptsize\textcolor{gray}{$\pm 0.01$}&$0.57$\scriptsize\textcolor{gray}{$\pm 0.01$}& $0.30$\scriptsize\textcolor{gray}{$\pm 0.01$}&$\textbf{0.30}$\scriptsize\textcolor{gray}{$\pm 0.01$}&$0.35$\scriptsize\textcolor{gray}{$\pm 0.01$}\\  
		Knockoff & - & - & $0.53$\scriptsize\textcolor{gray}{$\pm 0.01$}&$0.50$\scriptsize\textcolor{gray}{$\pm 0.01$}&$0.76$\scriptsize\textcolor{gray}{$\pm 0.01$}&$0.07$\scriptsize\textcolor{gray}{$\pm 0.00$}&$0.93$\scriptsize\textcolor{gray}{$\pm 0.00$}&$0.23$\scriptsize\textcolor{gray}{$\pm 0.01$}\\
		\hline
	\end{tabular}
	\caption{Mean squared error (MSE), mean absolute error (MAE), $\text{F}_1$ score, FDR, and sensitivity, averaged over all correlation cases, for various model selection approaches applied using SGS, with standard errors shown in grey. 1800 MC repetitions performed. \textit{Note:} Knockoff does not produce $\hat{\beta}$ estimates.}
	\label{tbl:models_selection_results}
\end{table}	
	
\section{Real data}\label{section:real_data}
In this section, the use of SGS as a prediction tool is explored through its application to two real datasets. The classification performance of SGS is compared to both lasso- and SLOPE- based models. 

The first dataset includes $127$ individuals with $85$ colitis patients and $42$ controls \citep{Burczynski2006}. The expression of 22283 genes were microarrayed across the individuals. The second dataset contains data from $60$ patients who had suffered from early-stage estrogen receptor-positive breast cancer and had been treated with tamoxifen \citep{Ma2004}. The patients were classified on whether the cancer had recurred. The initial dataset contained over $22575$ genes, but had a high level of missingness. Genes with over $50\%$ missingness were removed, resulting in $12071$ remaining genes and mean imputation on those genes was applied. Both datasets were accessed using the \texttt{GEOquery} R function\footnote{Accessed on 08/03/2023.}. The two datasets have previously been analysed in \cite{Simon2013}, where the authors applied lasso-based models to them. However, as the dataset sources have been updated since this publication, the analysis is repeated here, using the same cleaning steps as in \cite{Simon2013}.

The 9 major collections of gene-sets, C1-C8 and H, of the Human Molecular Signatures Database (MSigDB)\footnote{\href{https://www.gsea-msigdb.org/gsea/msigdb/human/collections.jsp}{gsea-msigdb.org/gsea/msigdb/human/collections.jsp}. Accessed on 08/03/2023.} were downloaded for grouping the genes of the two datasets into pathways. Table \ref{tbl:real_data_full_results} presents the number of pathways and their allocated genes for each dataset and each collection. As the pathways contain overlapping genes, we opted to duplicate the overlapping genes into the different pathways that they belong \citep{Jacob2009,Tang2018}. 

For both datasets, the samples were split into training and test sets, and the classification rate of the test set was computed using the trained models. All 9 collections were analysed for both datasets (see Table \ref{tbl:real_data_full_results}). Below, the results from the gene-set collection that achieved the highest peak classification are presented for the two datasets. 

The 127 samples of the colitis dataset were split into $50/77$ train/test set observations following the work of \cite{Simon2013}. The C3 pathway collection shared $12031$ genes with the dataset. Each model was applied to a log-linear path of $100$ $\lambda$ values, starting at a value of $\lambda_\text{max}$ which generates a null model and terminating at $\lambda_\text{min} = 0.1 \lambda_\text{max}$. 

SGS achieved the highest peak classification of the six models considered, at $97.4\%$, and was applied using $\alpha = 0.99$, showing that inducing only a small amount of group sparsity is enough to improve upon the peak of $94.8\%$ for SLOPE (Figure \ref{fig:real_data_slope_models}). The much lower peak of $84.4\%$ for gSLOPE highlights the downside of selecting all variables within a group, as often noise variables will enter the prediction. Interestingly, the lasso was found to have a higher classification peak than SGL (which also used $\alpha = 0.99$), with $93.5\%$ compared to $92.2\%$, but with both being lower than SGS (Table \ref{tbl:real_data_summary} and Figure \ref{fig:real_data_lasso_models}). This further illustrates the benefit of inducing stringent bi-level sparsity. Further fitting information, including the correct classification rate as a function of the number of predictors and the decision boundaries, is shown in Figure \ref{fig:colitis_solution}.

At the peak index of $37$, the SGS model selected $9$ genes from $7$ pathways (Table \ref{tbl:colitis_solution}). The gene \textit{NCK2} was found to be most strongly associated with a change in risk of colitis, which is in agreement to the findings of \cite{Burczynski2006}. Amongst the other genes found by SGS, \textit{TMEM158} and \textit{BASP1} were found to be up-regulated for the development of colitis by \cite{Xu2020}.

\begin{figure}[H]
		\vspace{-5pt}
		\includegraphics[width=1\textwidth]{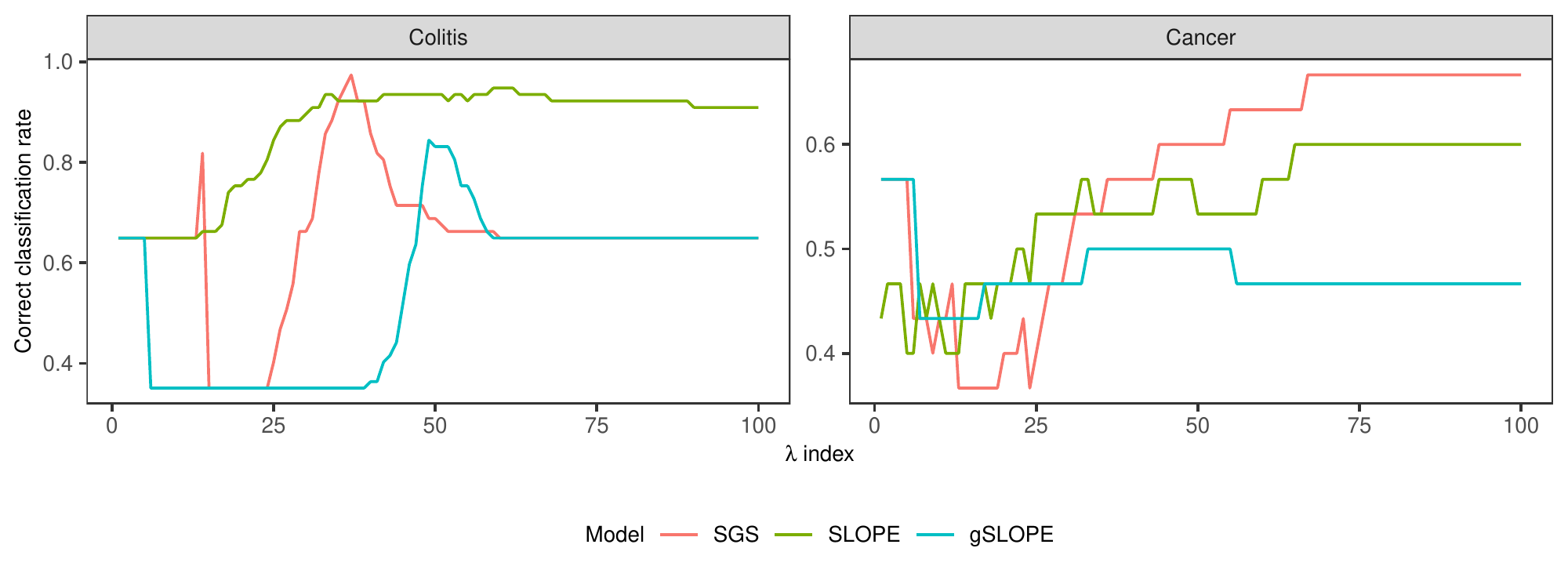}
		\vspace{-25pt}
		\caption[width=0.8\textwidth]{Correct classification rate (\%) ($\uparrow$) shown for SGS, SLOPE, and gSLOPE applied to the colitis and cancer datasets, along a $100$-$\lambda$ regularisation path.}
		\vspace{-10pt}
		\label{fig:real_data_slope_models}
\end{figure}

The $60$ patients of the breast cancer dataset were split evenly to train/test sets. From our conducted analysis, the C8 cell type signature gene sets collection gave the best classification results. The final dataset contains $6375$ genes that are grouped into $550$ pathways with sizes in the range $[1,533]$. For this dataset, the path was extended further to $\lambda_\text{max} = 0.01 \lambda_\text{min}$ to allow for denser models. The development of breast cancer follows a complex genetic landscape \citep{Skol2016}, so more genes are required for better predictive accuracy.

SGS is again found to outperform both SLOPE and gSLOPE, obtaining a peak accuracy of $66.7\%$, in comparison to $60.0\%$ and $50.0\%$ for SLOPE and gSLOPE (Figure \ref{fig:real_data_slope_models}). The optimal SGS model was found at the peak index of $67$ and selected $32$ genes from $20$ pathways (Table \ref{tbl:cancer_solution}). SGS is found to have the highest peak amongst the models considered (Table \ref{tbl:real_data_summary}). From the most associated genes found by SGS, \textit{COX6A1} and \textit{SUSD3} have also been shown previously to have an association with breast cancer \citep{Iacopetta2010, Yu2015}.

Table \ref{tbl:real_data_summary} presents the peak classification rates for each method considered. For the lasso-based models, we observe that the lasso outperforms SGL for both datasets, showing that for SGL, the grouping information provided no useful information for classification, but instead just increased the model variance. In constrast, by inducing more sparsity SGS is able to extract relevant grouping information, whilst discarding noisy variables, to improve predictive performance over SLOPE.
\begin{table}[H]
		\centering
		\begin{tabular}{l|lll|lll}
			\hline
			% & model & spar\_1 & spar\_2 & spar\_3 & spar\_4 & spar\_5 & spar\_6 & avg\_mse & avg\_mae \\ 
			\multirow{2}{*}{} & \multicolumn{3}{c}{SLOPE-based models} & \multicolumn{3}{c}{Lasso-based model}\\
			Dataset& SGS &SLOPE & gSLOPE & SGL & Lasso & gLasso\\
			\hline
			Colitis & \textbf{97.4}&94.8&84.4&92.2&93.5&89.6\\
			Cancer & \textbf{66.7}&60.0&56.7&50.0&56.7&36.7\\
			\hline
		\end{tabular}
		\caption{Correct classification rate (\%) ($\uparrow$) for the SLOPE- and lasso-based models applied to the colitis and cancer datasets.}
		\label{tbl:real_data_summary}
	\end{table}
 
These two data examples highlight the challenges and rewards of applying SGS to real data. In comparison to SGL, SGS has adaptive penalty weights, which require two additional hyperparameters to specify ($q_v$ and $q_g$), the choice of which influence the ultimate performance of SGS. Indeed, when optimising the performance of SGS for the colitis data, it obtained the highest peak ($97.4\%$), but beyond the peak it had a lower classification rate than SLOPE. Setting both FDR-control parameters to $0.01$ for SGS, we achieve consistently higher accuracy along the path than for the values $q_v =  10^{-4}$ and $q_g= 10^{-10} $ used in Figure \ref{fig:real_data_slope_models}, but with a peak slightly lower at $96.1\%$ (shown in Figure \ref{fig:real_data_slope_models_q_example}). As such, care is required in specifying the hyperparameters for SGS.

\begin{figure}[H]
		\vspace{-5pt}
		\includegraphics[width=1\textwidth]{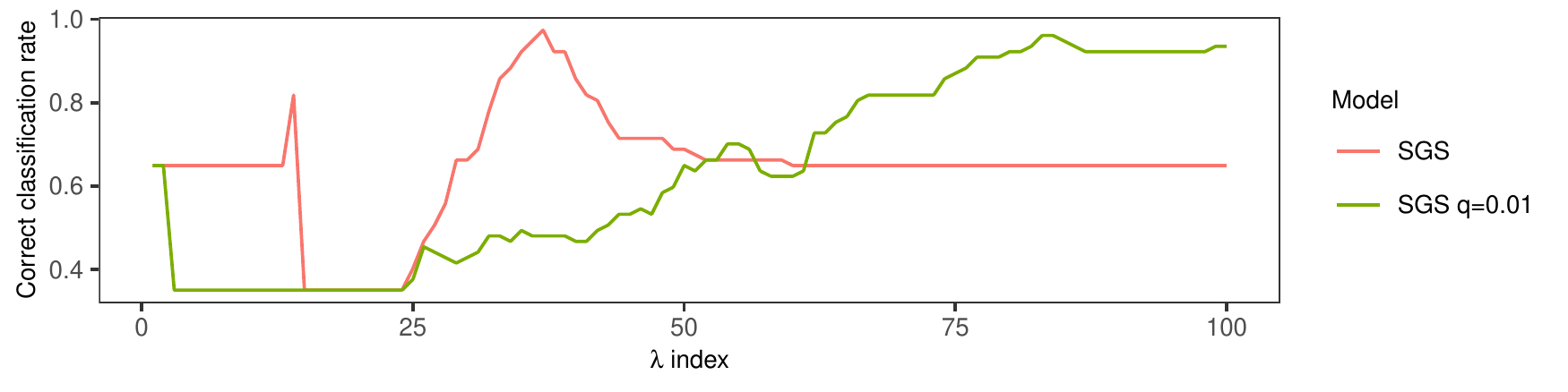}
		\vspace{-25pt}
		\caption[width=0.8\textwidth]{Correct classification rate (\%) ($\uparrow$) shown for SGS with $q_v = 10^{-4}, q_g = 10^{-10}$ and $q_v=q_g=0.1$, applied to the colitis dataset along a $100$-$\lambda$ path.}
		\vspace{-10pt}
		\label{fig:real_data_slope_models_q_example}
	\end{figure}
    
\section{Discussion}
This manuscript presents SGS, a new approach for bi-level selection based on incorporating SLOPE into a sparse-group framework. SGS aims to make use of the advantages of SLOPE with regards to FDR-control, whilst also integrating grouping information. SGS was shown to control bi-level FDR under orthogonal designs, using new penalty sequences derived specifically for SGS. The proposal has a convex and non-separable penalty. Due to the non-separability of the penalty, a proximal algorithm, ATOS, was applied to fitting SGS, which exploits knowledge of the proximal operators of SLOPE and gSLOPE. %ATOS requires the loss function to be differentiable with Lipschitz gradient, which is satisfied for both linear and logistic regression. %Future work should consider how to incorporate additional loss functions into such a framework, to allow for modelling of a wider class of cases, including Poisson and Cox regression. 

Through an extensive simulation study with grouped data, the performance of SGS was explored and compared to both lasso- and SLOPE-based methods. 
%The results presented in $\S$\ref{section:sim_studies} 
The conducted study showed that SGS achieves stronger bi-level selection performance than other lasso- and SLOPE-based models. SGS achieves higher performance by using grouping information and applying more stringent penalisation to discard noise variables. In particular, SGS was found to maintain strong performance under highly correlated designs, in comparison to the lasso and SGL, highlighting the benefit of adaptive penalisation. SGS was also found to perform very well under the null model; selecting very few variables as being significant. In comparison, gSLOPE was found to select many false variables, illustrating the downside of applying only group-wise sparsity. 

SGS was further applied to two real datasets and was assessed as a prediction tool. For both datasets, SGS achieved the highest peak classification accuracy, showing the benefit of applying both bi-level sparsity and more penalisation. In particular, SGL struggled in comparison to the lasso, showing that, unlike SGS, it was not able to utilise the grouping information. From the conducted analyses, genes linked with both colitis and breast cancer were identified. %of $97.4\%$ and $66.7\%$ for the colitis and cancer datasets, respectively. These peak scores were higher than the other methods considered, again highlighting the benefit of applying both bi-level sparsity and more penalisation. In particular, SGL struggled in comparison to the lasso, showing that, unlike SGS, it was not able to utilise the grouping information.

One of the challenges when working with regularised regression models is the selection of the tuning parameter, $\lambda$. One of the most widely used approaches for selecting the tuning parameter is through cross-validation, where the chosen value is the one that minimises the prediction error. The problem of model selection under a pathwise solution is a topic that has been extensively studied \citep{Giraud2012,Lee2016,Homrighausen2018} and in this manuscript was explored for SGS with a focus on finding an approach that encourages FDR-control under non-orthogonal designs. Knockoff was considered as an approach for FDR-control in conjunction with SGS, but failed to achieve the desired control, nor strong selection performance. A new algorithm for estimating jointly estimating the coefficients and $\lambda$, AS-SGS, was proposed and was shown to obtain the least biased estimates of the approaches considered. However, as with Knockoff, this approach failed to achieve exact FDR-control and was outperformed by cross-validation in terms of selection performance. Future research is required to develop model selection approaches to achieve exact FDR-control for SLOPE-based models under non-orthogonal designs. Such approaches will be able to extract the full potential of these models. %Some promising alternative approaches include stability selection (\citep{Meinshausen2010,Ahmed2011}), ET-Lasso (\citep{Yang2018}), and complementary pairs stability selection (\citep{Shah2013}). %In particular, the difficulty of choosing a parameter $\lambda$ which allows for FDR-control under non-orthogonal designs was explored. 

Similarly to both elastic-net and SGL, alongside $\lambda$, the hyperparameter $\alpha$ needs to be defined for SGS. In most cases, $\alpha$ tends to be set subjectively or found using a grid-search approach with cross-validation. Exploring more advanced approaches for the joint optimisation of both $\alpha$ and $\lambda$, possibly through the use of bi-level optimisation with FDR-control in mind, would be an exciting next step. An alternative future direction of work would be the implementation of screening rules, as the ones developed for SLOPE in \cite{Larsson2020}, for optimising the computational time of fitting an SGS model.

%The proposed SGS is found both theoretically and computationally to control FDR at both variable and group levels, and to outperform other regularisation approaches. 

%In general, SGS has the potential to outperform the other models considered in this manuscript for both prediction and selection, as these approaches can be seen as special cases of SGS. However, we also saw that SGS can often be challenging to implement, due to the increased number of hyperparameters. The work presented in this manuscript is just the first step in the use of SGS. Future work should investigate methods for extracting the full potential of SGS by exploring approaches for optimising the hyperparameters. A similar issue is faced in the application of elastic-net and SGL, which both require the selection of the $\alpha$ and $\lambda$ parameters, the former of which tends to be set subjectively or found using a grid-search approach with cross-validation. Exploring more advanced approaches for the joint optimisation of $\alpha$ and $\lambda$, perhaps through the use of bi-level optimisation, would be an exciting next step, using FDR-control as a primary goal of the optimisation. Further, in the timing tests (Table \ref{tbl:timing_results_real_data}), SGS was found to scale poorly for larger datasets. This highlights the need to develop screening rules specific to SGS, to reduce the dimension of the parameter space before fitting. The screening rule developed for SLOPE in \citep{Larsson2020} can be taken as a starting point.
\subsection*{Funding.}
FF gratefully acknowledges funding provided by the Engineering and Physical Sciences Research Council's Modern Statistics and Statistical Machine Learning Centre for Doctoral Training.
\subsection*{Conflict of interest.}
The authors declare they have no competing interests.
%%%%%%%%%%%%%%%%%%%%%%%%%%%%%%%%%%%%%%%%%%%%
\bibliographystyle{plainnat}
\bibliography{sgs_preprint}

\begin{thebibliography}{43}
\providecommand{\natexlab}[1]{#1}
\providecommand{\url}[1]{\texttt{#1}}
\expandafter\ifx\csname urlstyle\endcsname\relax
  \providecommand{\doi}[1]{doi: #1}\else
  \providecommand{\doi}{doi: \begingroup \urlstyle{rm}\Url}\fi

\bibitem[Baker et~al.(2020)Baker, Tang, and Allen]{Baker2020}
Yulia Baker, Tiffany~M. Tang, and Genevera~I. Allen.
\newblock {Feature selection for data integration with mixed multiview data}.
\newblock \emph{The Annals of Applied Statistics}, 14\penalty0 (4):\penalty0
  1676 -- 1698, 2020.
\newblock \doi{10.1214/20-AOAS1389}.

\bibitem[Barber and Candès(2015)]{Barber2015}
Rina~Foygel Barber and Emmanuel~J. Candès.
\newblock Controlling the false discovery rate via knockoffs.
\newblock \emph{The Annals of Statistics}, 43, 10 2015.
\newblock ISSN 0090-5364.
\newblock \doi{10.1214/15-AOS1337}.

\bibitem[Barber and Candès(2019)]{Barber2019}
Rina~Foygel Barber and Emmanuel~J. Candès.
\newblock A knockoff filter for high-dimensional selective inference.
\newblock \emph{The Annals of Statistics}, 47, 10 2019.
\newblock ISSN 0090-5364.
\newblock \doi{10.1214/18-AOS1755}.

\bibitem[Beck and Teboulle(2009)]{Beck2009}
Amir Beck and Marc Teboulle.
\newblock A fast iterative shrinkage-thresholding algorithm for linear inverse
  problems.
\newblock \emph{SIAM Journal on Imaging Sciences}, 2:\penalty0 183--202, 1
  2009.
\newblock ISSN 1936-4954.
\newblock \doi{10.1137/080716542}.

\bibitem[Bogdan et~al.(2015)Bogdan, van~den Berg, Sabatti, Su, and
  Candès]{Bogdan2015}
Małgorzata Bogdan, Ewout van~den Berg, Chiara Sabatti, Weijie Su, and
  Emmanuel~J. Candès.
\newblock Slope—adaptive variable selection via convex optimization.
\newblock \emph{The Annals of Applied Statistics}, 9, 9 2015.
\newblock ISSN 1932-6157.
\newblock \doi{10.1214/15-AOAS842}.

\bibitem[Brzyski et~al.(2015)Brzyski, Su, and Bogdan]{Brzyski2015}
Damian Brzyski, Weijie Su, and Małgorzata Bogdan.
\newblock Group slope - adaptive selection of groups of predictors, 2015.
\newblock arXiv:1511.09078.

\bibitem[Bu et~al.(2021)Bu, Klusowski, Rush, and Su]{Bu2019}
Zhiqi Bu, Jason~M. Klusowski, Cynthia Rush, and Weijie~J. Su.
\newblock Algorithmic analysis and statistical estimation of slope via
  approximate message passing.
\newblock \emph{IEEE Transactions on Information Theory}, 67:\penalty0
  506--537, 1 2021.
\newblock ISSN 0018-9448.
\newblock \doi{10.1109/TIT.2020.3025272}.

\bibitem[Burczynski et~al.(2006)Burczynski, Peterson, Twine, Zuberek, Brodeur,
  Casciotti, Maganti, Reddy, Strahs, Immermann, Spinelli, Schwertschlag,
  Slager, Cotreau, and Dorner]{Burczynski2006}
Michael~E. Burczynski, Ron~L. Peterson, Natalie~C. Twine, Krystyna~A. Zuberek,
  Brendan~J. Brodeur, Lori Casciotti, Vasu Maganti, Padma~S. Reddy, Andrew
  Strahs, Fred Immermann, Walter Spinelli, Ulrich Schwertschlag, Anna~M.
  Slager, Monette~M. Cotreau, and Andrew~J. Dorner.
\newblock Molecular classification of crohn's disease and ulcerative colitis
  patients using transcriptional profiles in peripheral blood mononuclear
  cells.
\newblock \emph{The Journal of Molecular Diagnostics}, 8:\penalty0 51--61, 2
  2006.
\newblock ISSN 15251578.
\newblock \doi{10.2353/jmoldx.2006.050079}.

\bibitem[Cai and Wu(2014)]{Cai2014}
Tony~T. Cai and Yihong Wu.
\newblock Optimal detection of sparse mixtures against a given null
  distribution.
\newblock \emph{IEEE Transactions on Information Theory}, 60:\penalty0
  2217--2232, 4 2014.
\newblock ISSN 0018-9448.
\newblock \doi{10.1109/TIT.2014.2304295}.

\bibitem[Davis and Yin(2017)]{Davis2017}
Damek Davis and Wotao Yin.
\newblock A three-operator splitting scheme and its optimization applications.
\newblock \emph{Set-Valued and Variational Analysis}, 25:\penalty0 829--858, 12
  2017.
\newblock ISSN 1877-0533.
\newblock \doi{10.1007/s11228-017-0421-z}.

\bibitem[Evangelou et~al.(2014)Evangelou, Smyth, Fortune, Burren, Walker, Guo,
  Onengut-Gumuscu, Chen, Concannon, Rich, Todd, and Wallace]{Evangelou2014}
Marina Evangelou, Deborah~J. Smyth, Mary~D. Fortune, Oliver~S. Burren, Neil~M.
  Walker, Hui Guo, Suna Onengut-Gumuscu, Wei-Min Chen, Patrick Concannon,
  Stephen~S. Rich, John~A. Todd, and Chris Wallace.
\newblock A method for gene-based pathway analysis using genomewide association
  study summary statistics reveals nine new type 1 diabetes associations.
\newblock \emph{Genetic Epidemiology}, 38\penalty0 (8):\penalty0 661--670,
  2014.
\newblock \doi{https://doi.org/10.1002/gepi.21853}.

\bibitem[Freijeiro‐González et~al.(2022)Freijeiro‐González,
  Febrero‐Bande, and González‐Manteiga]{Freijeiro-Gonzalez2022}
Laura Freijeiro‐González, Manuel Febrero‐Bande, and Wenceslao
  González‐Manteiga.
\newblock A critical review of lasso and its derivatives for variable selection
  under dependence among covariates.
\newblock \emph{International Statistical Review}, 90:\penalty0 118--145, 4
  2022.
\newblock ISSN 0306-7734.
\newblock \doi{10.1111/insr.12469}.

\bibitem[Friedman et~al.(2010)Friedman, Hastie, and Tibshirani]{Friedman2010b}
Jerome Friedman, Trevor Hastie, and Robert Tibshirani.
\newblock Regularization paths for generalized linear models via coordinate
  descent.
\newblock \emph{Journal of Statistical Software}, 33, 2010.
\newblock ISSN 1548-7660.
\newblock \doi{10.18637/jss.v033.i01}.

\bibitem[Giraud et~al.(2012)Giraud, Huet, and Verzelen]{Giraud2012}
Christophe Giraud, Sylvie Huet, and Nicolas Verzelen.
\newblock {High-Dimensional Regression with Unknown Variance}.
\newblock \emph{Statistical Science}, 27\penalty0 (4):\penalty0 500 -- 518,
  2012.
\newblock \doi{10.1214/12-STS398}.

\bibitem[Gossmann et~al.(2015)Gossmann, Cao, and Wang]{Gossmann2015}
Alexej Gossmann, Shaolong Cao, and Yu-Ping Wang.
\newblock Identification of significant genetic variants via slope, and its
  extension to group slope.
\newblock pages 232--240. ACM, 9 2015.
\newblock ISBN 9781450338530.
\newblock \doi{10.1145/2808719.2808743}.

\bibitem[Hastie et~al.(2015)Hastie, Tibshirani, and Wainwright]{Hastie2015}
Trevor Hastie, Robert Tibshirani, and Martin Wainwright.
\newblock \emph{Statistical Learning with Sparsity}.
\newblock Chapman and Hall/CRC, 5 2015.
\newblock ISBN 9780429171581.
\newblock \doi{10.1201/b18401}.

\bibitem[Homrighausen and McDonald(2018)]{Homrighausen2018}
Darren Homrighausen and Daniel~J. McDonald.
\newblock A study on tuning parameter selection for the high-dimensional lasso.
\newblock \emph{Journal of Statistical Computation and Simulation}, 88\penalty0
  (15):\penalty0 2865--2892, jun 2018.
\newblock \doi{10.1080/00949655.2018.1491575}.

\bibitem[Iacopetta et~al.(2010)Iacopetta, Lappano, Cappello, Madeo, Francesco,
  Santoro, Curcio, Capobianco, Pezzi, Maggiolini, and Dolce]{Iacopetta2010}
Domenico Iacopetta, Rosamaria Lappano, Anna~Rita Cappello, Marianna Madeo,
  Ernestina Marianna~De Francesco, Antonella Santoro, Rosita Curcio, Loredana
  Capobianco, Vincenzo Pezzi, Marcello Maggiolini, and Vincenza Dolce.
\newblock Slc37a1 gene expression is up-regulated by epidermal growth factor in
  breast cancer cells.
\newblock \emph{Breast Cancer Research and Treatment}, 122:\penalty0 755--764,
  8 2010.
\newblock ISSN 0167-6806.
\newblock \doi{10.1007/s10549-009-0620-x}.

\bibitem[Jacob et~al.(2009)Jacob, Obozinski, and Vert]{Jacob2009}
Laurent Jacob, Guillaume Obozinski, and Jean-Philippe Vert.
\newblock Group lasso with overlap and graph lasso.
\newblock pages 433--440. ACM, 6 2009.
\newblock ISBN 9781605585161.
\newblock \doi{10.1145/1553374.1553431}.

\bibitem[Larsson et~al.(2020)Larsson, Bogdan, and Wallin]{Larsson2020}
Johan Larsson, Ma\l{}gorzata Bogdan, and Jonas Wallin.
\newblock The strong screening rule for slope.
\newblock In \emph{Proceedings of the 34th International Conference on Neural
  Information Processing Systems}, NIPS'20, Red Hook, NY, USA, 2020. Curran
  Associates Inc.
\newblock ISBN 9781713829546.

\bibitem[Larsson et~al.(2022)Larsson, Wallin, Bogdan, {van den Berg}, Sabatti,
  Candes, Patterson, Su, Kała, Grzesiak, and Burdukiewicz]{SLOPEpackage}
Johan Larsson, Jonas Wallin, Malgorzata Bogdan, Ewout {van den Berg}, Chiara
  Sabatti, Emmanuel Candes, Evan Patterson, Weijie Su, Jakub Kała, Krystyna
  Grzesiak, and Michal Burdukiewicz.
\newblock \emph{{SLOPE}: Sorted L1 Penalized Estimation}, 2022.
\newblock URL \url{https://CRAN.R-project.org/package=SLOPE}.
\newblock R package version 0.5.0.

\bibitem[Lee et~al.(2016)Lee, Sun, Sun, and Taylor]{Lee2016}
Jason~D. Lee, Dennis~L. Sun, Yuekai Sun, and Jonathan~E. Taylor.
\newblock {Exact post-selection inference, with application to the lasso}.
\newblock \emph{The Annals of Statistics}, 44\penalty0 (3):\penalty0 907 --
  927, 2016.
\newblock \doi{10.1214/15-AOS1371}.

\bibitem[Leng et~al.(2006)Leng, Lin, and Wahba]{Leng2006}
Chenlei Leng, Yi~Lin, and Grace Wahba.
\newblock A note on the lasso and related procedures in model selection.
\newblock \emph{Statistica Sinica}, 16:\penalty0 1273--1284, 2006.

\bibitem[Ma et~al.(2004)Ma, Wang, Ryan, Isakoff, Barmettler, Fuller, Muir,
  Mohapatra, Salunga, Tuggle, Tran, Tran, Tassin, Amon, Wang, Wang, Enright,
  Stecker, Estepa-Sabal, Smith, Younger, Balis, Michaelson, Bhan, Habin, Baer,
  Brugge, Haber, Erlander, and Sgroi]{Ma2004}
Xiao-Jun Ma, Zuncai Wang, Paula~D Ryan, Steven~J Isakoff, Anne Barmettler,
  Andrew Fuller, Beth Muir, Gayatry Mohapatra, Ranelle Salunga, J.Todd Tuggle,
  Yen Tran, Diem Tran, Ana Tassin, Paul Amon, Wilson Wang, Wei Wang, Edward
  Enright, Kimberly Stecker, Eden Estepa-Sabal, Barbara Smith, Jerry Younger,
  Ulysses Balis, James Michaelson, Atul Bhan, Karleen Habin, Thomas~M Baer,
  Joan Brugge, Daniel~A Haber, Mark~G Erlander, and Dennis~C Sgroi.
\newblock A two-gene expression ratio predicts clinical outcome in breast
  cancer patients treated with tamoxifen.
\newblock \emph{Cancer Cell}, 5:\penalty0 607--616, 6 2004.
\newblock ISSN 15356108.
\newblock \doi{10.1016/j.ccr.2004.05.015}.

\bibitem[Moscovich and Rosset(2019)]{Moscovich2019}
Amit Moscovich and Saharon Rosset.
\newblock On the cross-validation bias due to unsupervised pre-processing.
\newblock 1 2019.
\newblock \doi{10.1111/rssb.12537}.

\bibitem[Parikh(2014)]{Parikh2014}
Neal Parikh.
\newblock Proximal algorithms.
\newblock \emph{Foundations and Trends in Optimization}, 1:\penalty0 127--239,
  2014.
\newblock ISSN 2167-3888.
\newblock \doi{10.1561/2400000003}.

\bibitem[Pedregosa and Gidel(2018)]{Pedregosa2018}
Fabian Pedregosa and Gauthier Gidel.
\newblock Adaptive three operator splitting.
\newblock In Jennifer Dy and Andreas Krause, editors, \emph{Proceedings of the
  35th International Conference on Machine Learning}, volume~80 of
  \emph{Proceedings of Machine Learning Research}, pages 4085--4094. PMLR,
  10--15 Jul 2018.

\bibitem[Ro{\v{c}}kov{\'a} and George(2016)]{Rockova2016}
Veronika Ro{\v{c}}kov{\'a} and Edward~I. George.
\newblock Bayesian penalty mixing: The case of a non-separable penalty.
\newblock In Arnoldo Frigessi, Peter B{\"u}hlmann, Ingrid~K. Glad, Mette
  Langaas, Sylvia Richardson, and Marina Vannucci, editors, \emph{Statistical
  Analysis for High-Dimensional Data}, pages 233--254, Cham, 2016. Springer
  International Publishing.
\newblock ISBN 978-3-319-27099-9.

\bibitem[Simon et~al.(2013)Simon, Friedman, Hastie, and Tibshirani]{Simon2013}
Noah Simon, Jerome Friedman, Trevor Hastie, and Robert Tibshirani.
\newblock A sparse-group lasso.
\newblock \emph{Journal of Computational and Graphical Statistics},
  22:\penalty0 231--245, 4 2013.
\newblock ISSN 1061-8600.
\newblock \doi{10.1080/10618600.2012.681250}.

\bibitem[Simon et~al.(2019)Simon, Friedman, Hastie, and Tibshirani]{SGLpackage}
Noah Simon, Jerome Friedman, Trevor Hastie, and Rob Tibshirani.
\newblock \emph{SGL: Fit a GLM (or Cox Model) with a Combination of Lasso and
  Group Lasso Regularization}, 2019.
\newblock URL \url{https://CRAN.R-project.org/package=SGL}.
\newblock R package version 1.3.

\bibitem[Skol et~al.(2016)Skol, Sasaki, and Onel]{Skol2016}
Andrew~D. Skol, Mark~M. Sasaki, and Kenan Onel.
\newblock The genetics of breast cancer risk in the post-genome era: thoughts
  on study design to move past brca and towards clinical relevance.
\newblock \emph{Breast Cancer Research}, 18:\penalty0 99, 12 2016.
\newblock ISSN 1465-542X.
\newblock \doi{10.1186/s13058-016-0759-4}.

\bibitem[Su and Candès(2016)]{Su2016}
Weijie Su and Emmanuel Candès.
\newblock Slope is adaptive to unknown sparsity and asymptotically minimax.
\newblock \emph{The Annals of Statistics}, 44, 6 2016.
\newblock ISSN 0090-5364.
\newblock \doi{10.1214/15-AOS1397}.

\bibitem[Sun and Zhang(2012)]{Sun2012}
Tingni Sun and Cun-Hui Zhang.
\newblock Scaled sparse linear regression.
\newblock \emph{Biometrika}, 99:\penalty0 879--898, 4 2012.

\bibitem[Tang et~al.(2018)Tang, Shen, Li, Zhang, Wen, Qian, Zhuang, Shi, and
  Yi]{Tang2018}
Zaixiang Tang, Yueping Shen, Yan Li, Xinyan Zhang, Jia Wen, Chen’ao Qian,
  Wenzhuo Zhuang, Xinghua Shi, and Nengjun Yi.
\newblock Group spike-and-slab lasso generalized linear models for disease
  prediction and associated genes detection by incorporating pathway
  information.
\newblock \emph{Bioinformatics}, 34:\penalty0 901--910, 3 2018.
\newblock ISSN 1367-4803.
\newblock \doi{10.1093/bioinformatics/btx684}.

\bibitem[Tibshirani(1996)]{Tibshirani1996}
Robert Tibshirani.
\newblock Regression shrinkage and selection via the lasso.
\newblock \emph{Journal of the Royal Statistical Society: Series B
  (Methodological)}, 58:\penalty0 267--288, 1 1996.
\newblock ISSN 00359246.
\newblock \doi{10.1111/j.2517-6161.1996.tb02080.x}.

\bibitem[Xu et~al.(2020)Xu, Yan, Chen, Guo, Guo, Tang, and Shi]{Xu2020}
Guangya Xu, Xueling Yan, Jie Chen, Xiaoheng Guo, Xiaolan Guo, Yong Tang, and
  Zheng Shi.
\newblock Bioinformatics analysis of key candidate genes and pathways in
  ulcerative colitis.
\newblock \emph{Biological and Pharmaceutical Bulletin}, 43\penalty0
  (11):\penalty0 1760--1766, 2020.
\newblock \doi{10.1248/bpb.b20-00488}.

\bibitem[Yang(2005)]{Yang2005}
Yuhong Yang.
\newblock Can the strengths of aic and bic be shared? a conflict between model
  indentification and regression estimation.
\newblock \emph{Biometrika}, 92:\penalty0 937--950, 12 2005.
\newblock ISSN 1464-3510.
\newblock \doi{10.1093/biomet/92.4.937}.

\bibitem[Yu et~al.(2015)Yu, Jiang, Wang, Shi, Shangguan, Zhang, and Li]{Yu2015}
Zhenghong Yu, Enze Jiang, Xinxing Wang, Yaqin Shi, Anna~Junjie Shangguan, Luo
  Zhang, and Jie Li.
\newblock Sushi domain-containing protein 3: A potential target for breast
  cancer.
\newblock \emph{Cell Biochemistry and Biophysics}, 72:\penalty0 321--324, 6
  2015.
\newblock ISSN 1085-9195.
\newblock \doi{10.1007/s12013-014-0480-9}.

\bibitem[Yuan and Lin(2006)]{Yuan2006}
Ming Yuan and Yi~Lin.
\newblock Model selection and estimation in regression with grouped variables.
\newblock \emph{Journal of the Royal Statistical Society: Series B (Statistical
  Methodology)}, 68:\penalty0 49--67, 2 2006.
\newblock ISSN 1369-7412.
\newblock \doi{10.1111/j.1467-9868.2005.00532.x}.

\bibitem[Zeng and Figueiredo(2015)]{Zeng2014}
Xiangrong Zeng and Mário A.~T. Figueiredo.
\newblock The ordered weighted $\ell_1$ norm: Atomic formulation, projections,
  and algorithms, 2015.
\newblock arXiv:1409.4271.

\bibitem[Zhang and Bu(2021)]{Zhang2021}
Yiliang Zhang and Zhiqi Bu.
\newblock Efficient designs of slope penalty sequences in finite dimension.
\newblock \emph{Proceedings of the 24th International Conference on Artificial
  Intelligence and Statistics (AISTATS)}, 130, 2 2021.

\bibitem[Zou(2006)]{Zou2006}
Hui Zou.
\newblock The adaptive lasso and its oracle properties.
\newblock \emph{Journal of the American Statistical Association}, 101:\penalty0
  1418--1429, 12 2006.
\newblock ISSN 0162-1459.
\newblock \doi{10.1198/016214506000000735}.

\bibitem[Zou and Hastie(2005)]{Zou2005}
Hui Zou and Trevor Hastie.
\newblock Regularization and variable selection via the elastic net.
\newblock \emph{Journal of the Royal Statistical Society Series B: Statistical
  Methodology}, 67:\penalty0 301--320, 4 2005.
\newblock ISSN 1369-7412.
\newblock \doi{10.1111/j.1467-9868.2005.00503.x}.

\end{thebibliography}
\newpage
\appendix
\counterwithin{figure}{section}
\counterwithin{table}{section}
\renewcommand\thefigure{\thesection\arabic{figure}}
\renewcommand\thetable{\thesection\arabic{table}}

\section{Definitions}\label{appendix:definitions}
    \subsection{Definitions.}
    Let $TP$, $TN$, $FP$, and $FN$ define the number of true positives, true negatives, false positives, and false negatives respectively.
    \begin{definition}[Type I error]
        A type I error in hypothesis testing is the mistaken rejection of an actually true null hypothesis.
    \end{definition}
    \begin{definition}[Sensitivity]
    The sensitivity of a variable selection event is defined as the probability of rejecting the null of no effect, given that the variable is a true signal. Formally, it is given by 
    \begin{equation}
        Sensitivity = \frac{TP}{TP + FN}.
    \end{equation}
    \end{definition}
    \begin{definition}[False discovery rate (FDR)]
    The false discovery rate (FDR) defines the rate of type I errors when conducting multiple testing. Formally, it is defined as 
    \begin{equation}
        FDR = \frac{FP}{FP+TP}.
    \end{equation}
    \end{definition}
    \begin{definition}[$F_1$ score]\label{defn:f1_score}
    The $F_1$ score is a measure of a test's accuracy, with it being the harmonic mean of precision and sensitivity. Formally, it is given by
    \begin{equation}
        F_1 = \frac{TP+TN}{TP+TN+FP+FN}.
    \end{equation}
    \end{definition}
\section{Sparse-group SLOPE (SGS)}\label{appendix:sgs}
\subsection{Binomial loss function.} 
\label{appendix:binomial_loss_fcn}
To apply ATOS to a Binomial response, the loss function needs to be convex and $L_f$-smooth. The loss function for logistic regression satisfies these constraints, given by $\ell(b;y,\mathbf{X})=-1/n \log(\mathcal{L}(b; y, \mathbf{X}))$, where $\mathcal{L}$ is the log-likelihood of a binomial distribution, given by
\begin{equation}\label{eqn:sgs_log}
	\mathcal{L}(b; y, \mathbf{X}) = \sum_{i=1}^{n}\left\{y_i b^\intercal x_i - \log(1+\exp(b^\intercal x_i)) \right\}.
\end{equation}
The negative of the log-likelihood is used as this is equivalent to maximising the likelihood.
\subsection{Fitting algorithm.}\label{appendix:fitting_algo}
\begin{theorem}\label{appendix:convexity_proof}
The SGS penalty (Equation \ref{eqn:sgs}) is convex.
\end{theorem}
	\begin{proof}
		The SLOPE penalty is convex \citep{Bogdan2015}. Similarly, the group SLOPE penalty is also convex \citep{Brzyski2015}. Finally, the sum of convex functions is convex. Hence, the penalty function for SGS is convex.
	\end{proof}
\section{FDR-control}\label{appendix:fdr_control}
\begin{figure}[H]		
 \includegraphics[width=1\textwidth]{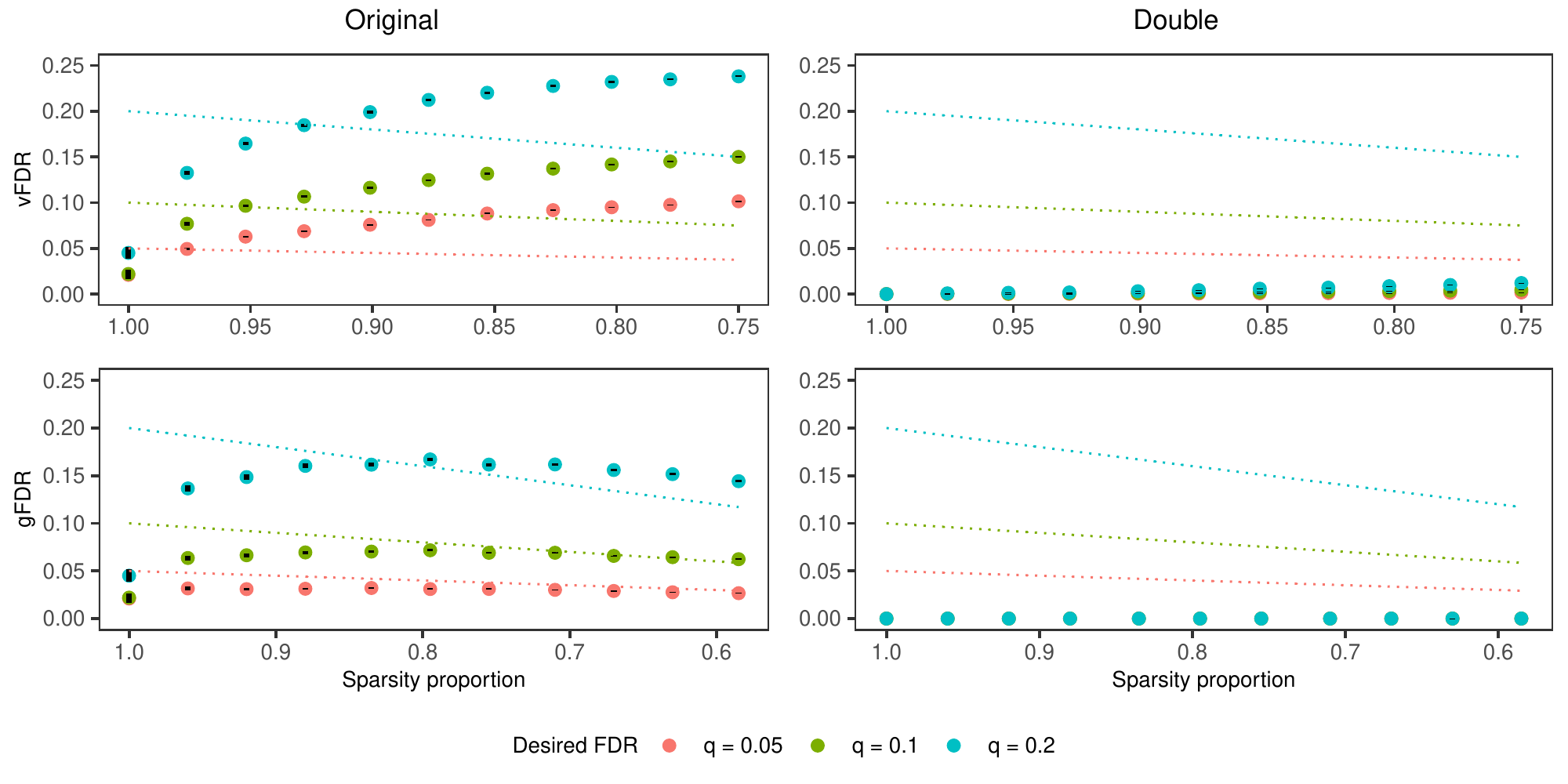}  
	\caption[width=0.8\textwidth]{vFDR and gFDR shown for SGS Original and SGS Double (both Max) under orthogonal design with even groups, as a function of decreasing sparsity proportion. $100$ MC repetitions performed per sparsity proportion. The sensitivity is given in Figure \ref{fig:even_sgs_org_double_sens}.}
	\label{fig:even_sgs_org_double}
\end{figure}
\begin{figure}[H]
\includegraphics[width=1\textwidth]{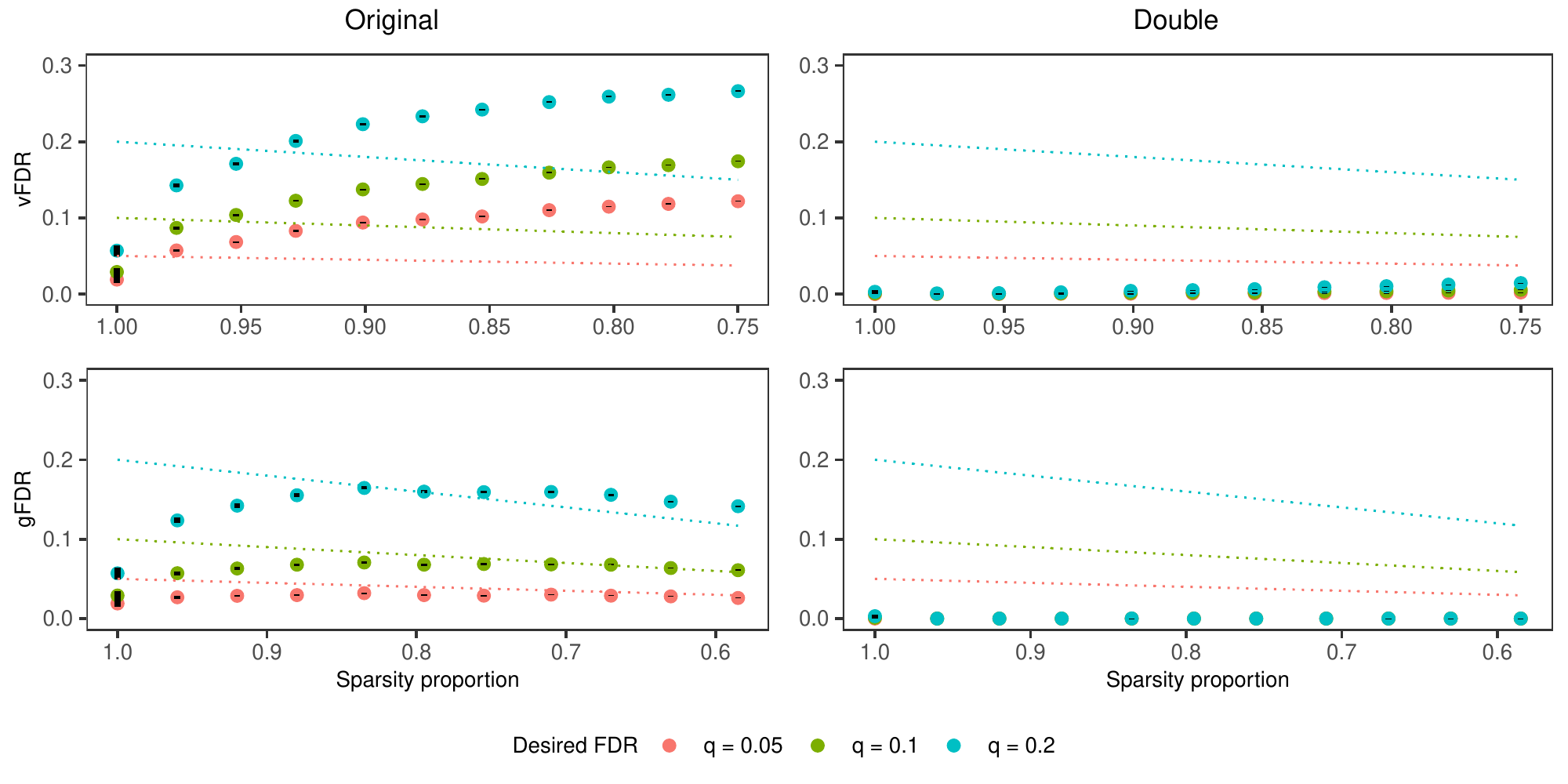}    
	\caption[width=0.8\textwidth]{vFDR and gFDR shown for SGS Original and SGS Double (both Max) under orthogonal design with uneven groups, as a function of decreasing sparsity proportion. $100$ MC repetitions performed per sparsity proportion. The sensitivity is given in Figure \ref{fig:uneven_sgs_org_double_sens}.}
\label{fig:uneven_sgs_org_double}
\end{figure}
\begin{figure}[H]
\includegraphics[width=1\textwidth]{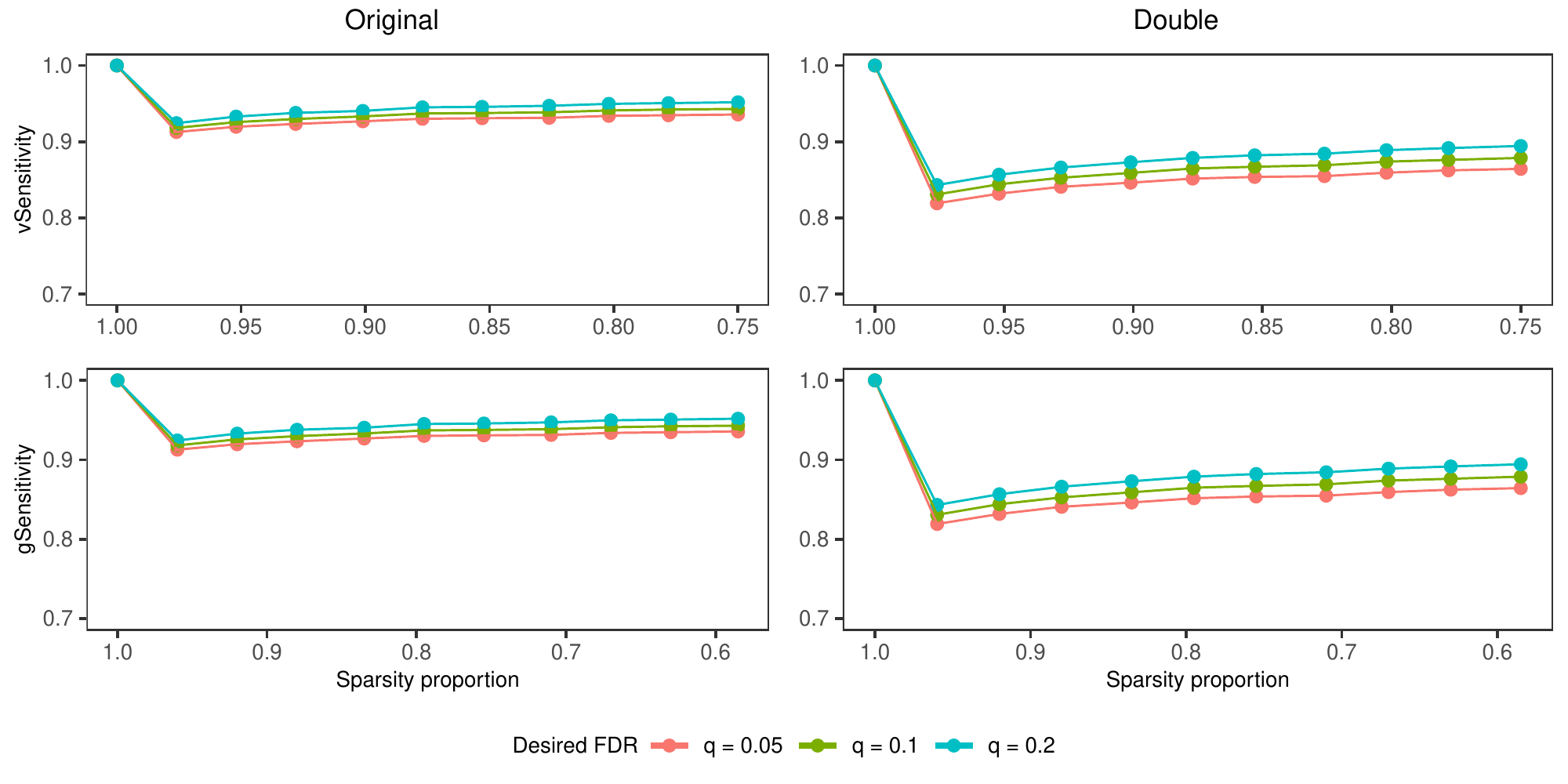} 
		\caption[width=0.8\textwidth]{vSensitivity and gSensitivity shown for SGS Original and SGS Double (both Max) under orthogonal design with even groups, as a function of decreasing sparsity proportion. $100$ MC repetitions performed per sparsity proportion.}
		\label{fig:even_sgs_org_double_sens}
	\end{figure}
 	\begin{figure}[H]
		\includegraphics[width=1\textwidth]{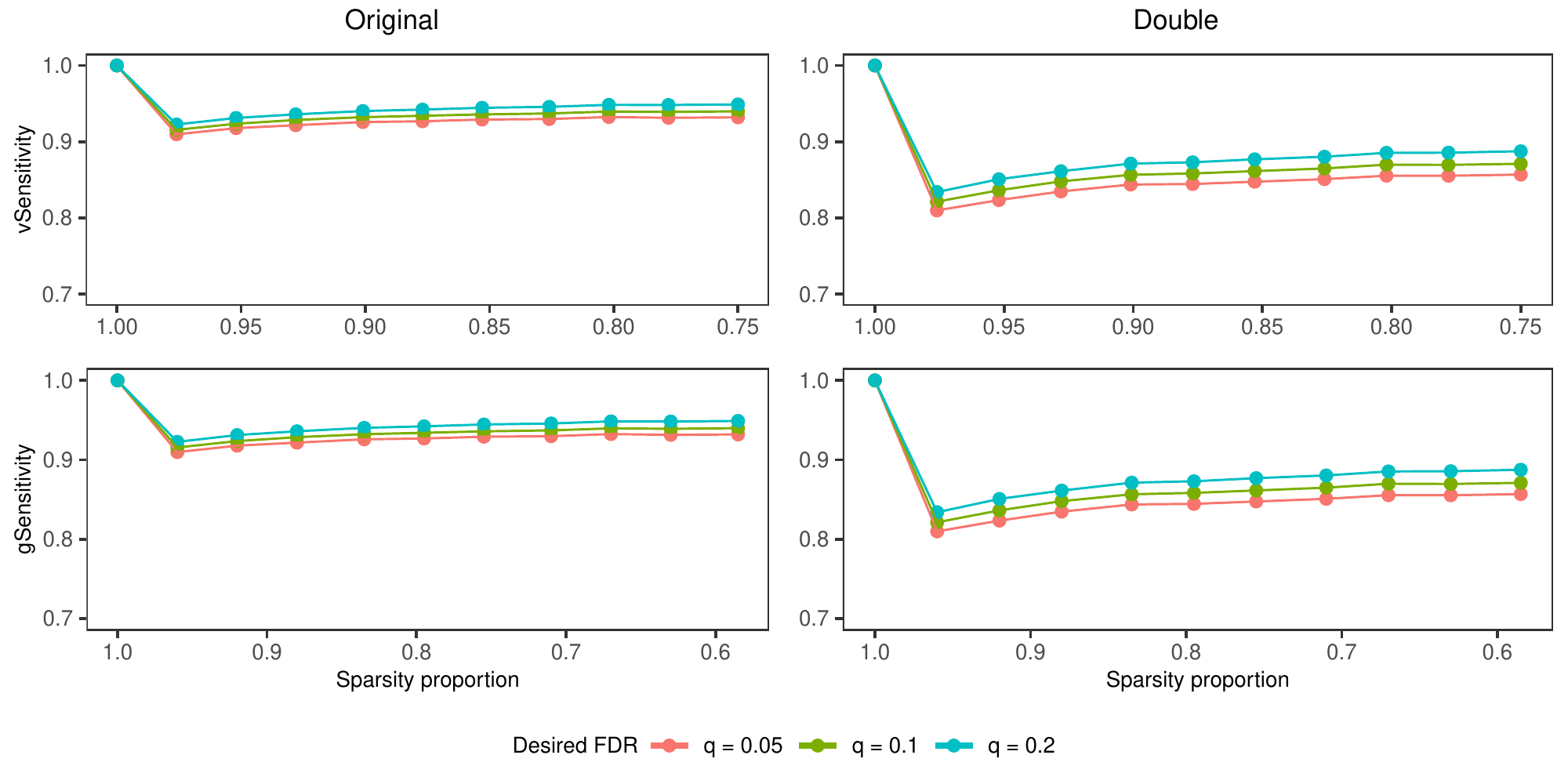}   
		\caption[width=0.8\textwidth]{vSensitivity and gSensitivity shown for SGS Original and SGS Double (both Max) under orthogonal design with uneven groups, as a function of decreasing sparsity proportion. $100$ MC repetitions performed per sparsity proportion.}
		\label{fig:uneven_sgs_org_double_sens}
	\end{figure}
    	\begin{figure}[H]
		\includegraphics[width=1\textwidth]{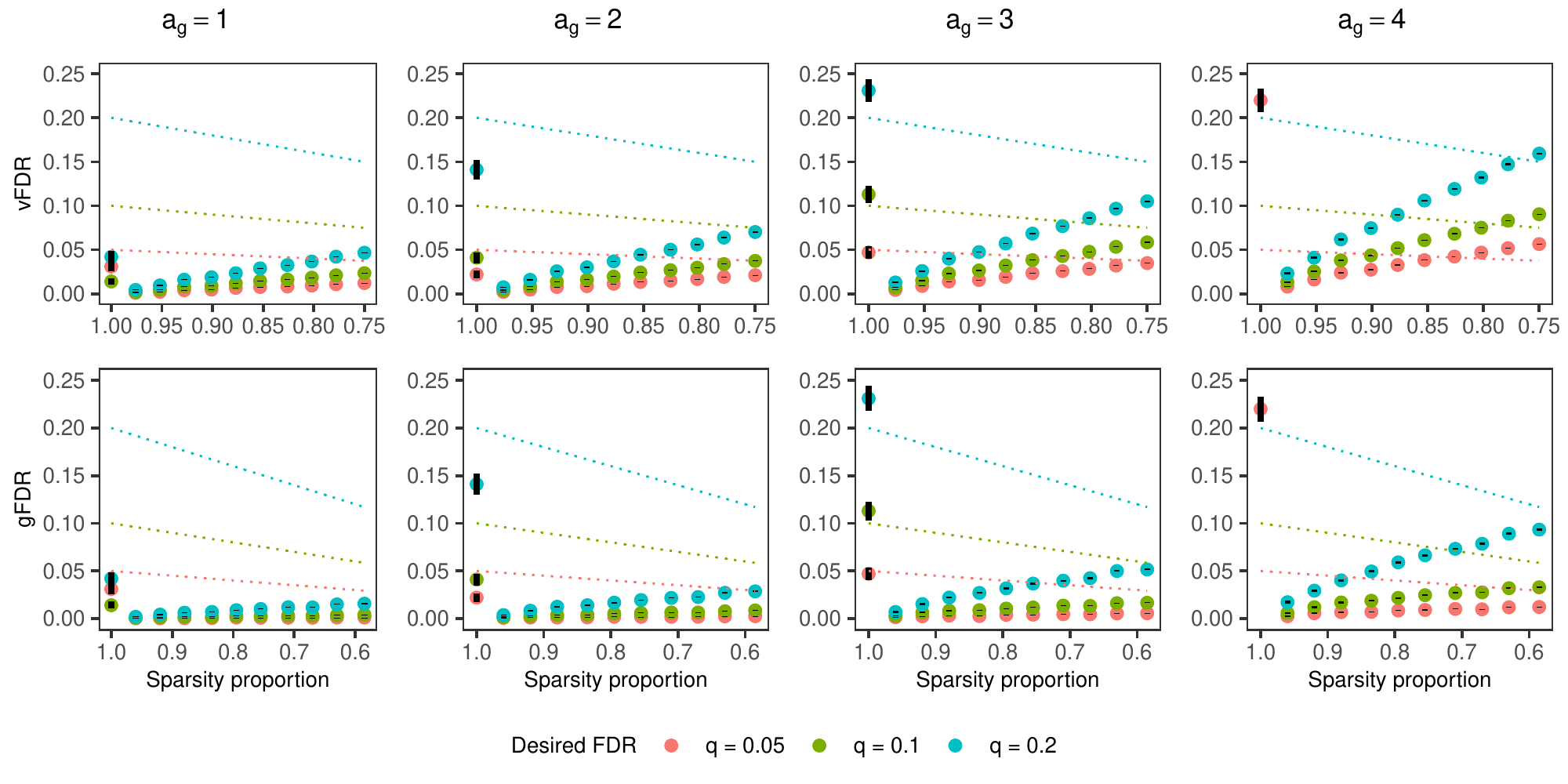}   
		\caption[width=0.8\textwidth]{vFDR and gFDR shown for SGS with vMax and gSLOPE mean sequences under orthogonal design with even groups, for different values of $a_g$, as a function of decreasing sparsity proportion. $100$ MC repetitions performed per sparsity proportion. The sensitivity is given in Figure \ref{fig:var_max_ag_all_sens}.}
		\label{fig:var_max_ag_all}
	\end{figure}
	\begin{figure}[H]
		\includegraphics[width=1\textwidth]{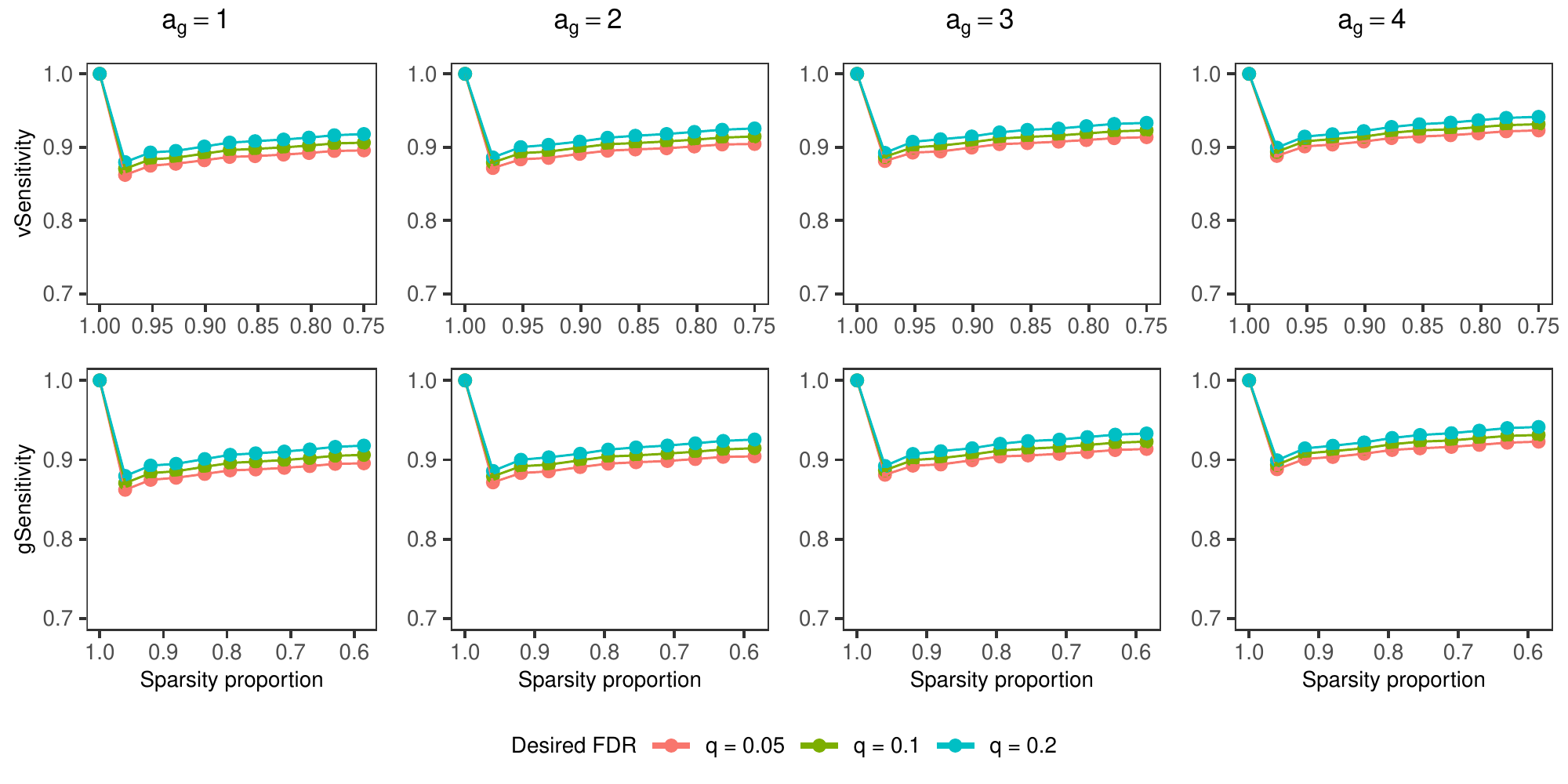}  
		\caption[width=0.8\textwidth]{vSensitivity and gSensitivity shown for SGS with vMax and gSLOPE mean sequences under orthogonal design with even groups, for different values of $a_g$, as a function of decreasing sparsity proportion. $100$ MC repetitions performed per sparsity proportion.}
		\label{fig:var_max_ag_all_sens}
	\end{figure}
   \begin{figure}[H]
		\includegraphics[width=1\textwidth]{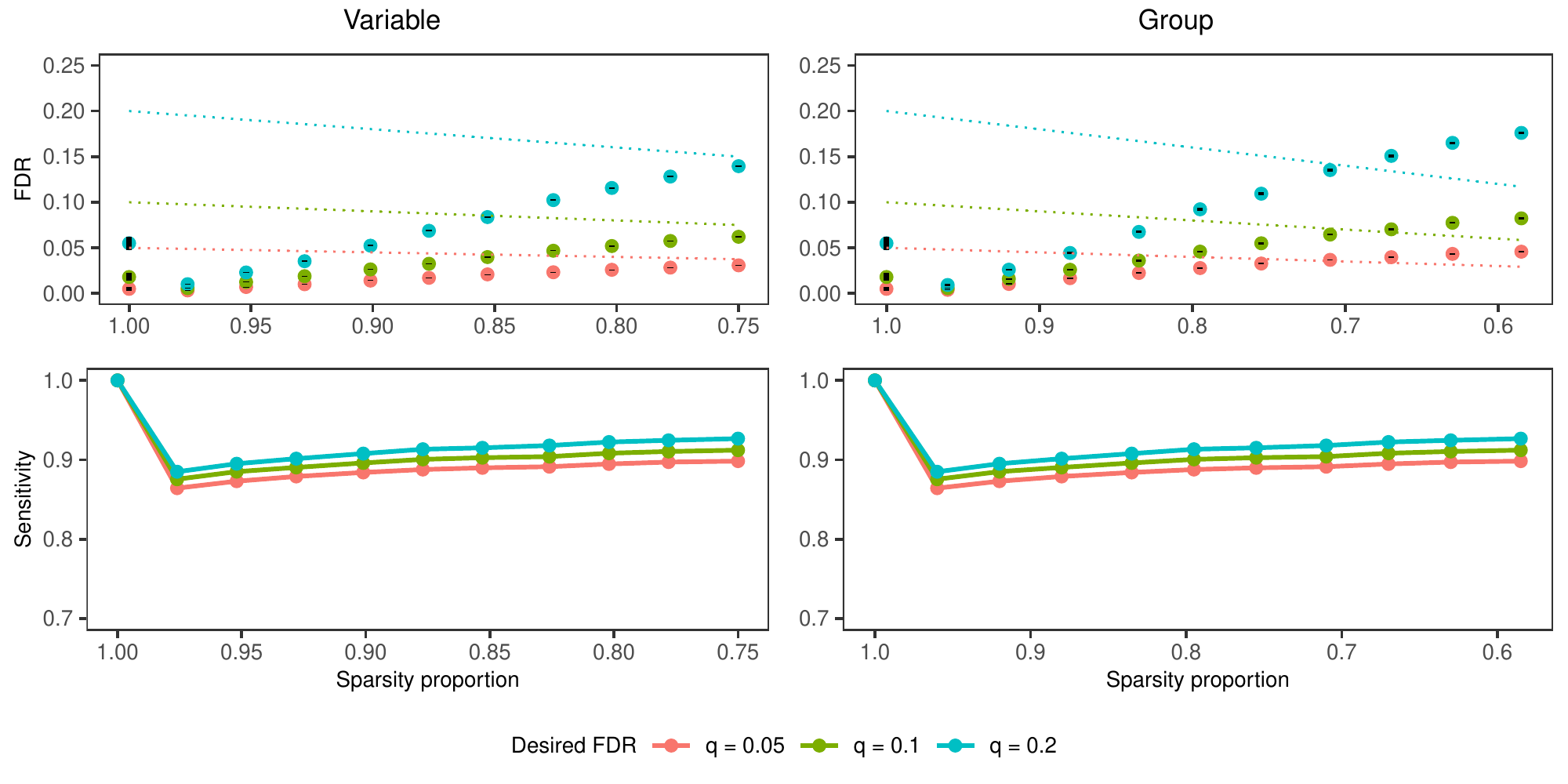} 
		\caption[width=0.8\textwidth]{Variable and group metrics shown for SGS with vMean and gMean sequences under orthogonal design with even groups, as a function of decreasing sparsity proportion. $100$ MC repetitions performed per sparsity proportion.}
		\label{fig:sim_3_even_gmean_vmean}
	\end{figure}
 \begin{figure}[H]
		\includegraphics[width=1\textwidth]{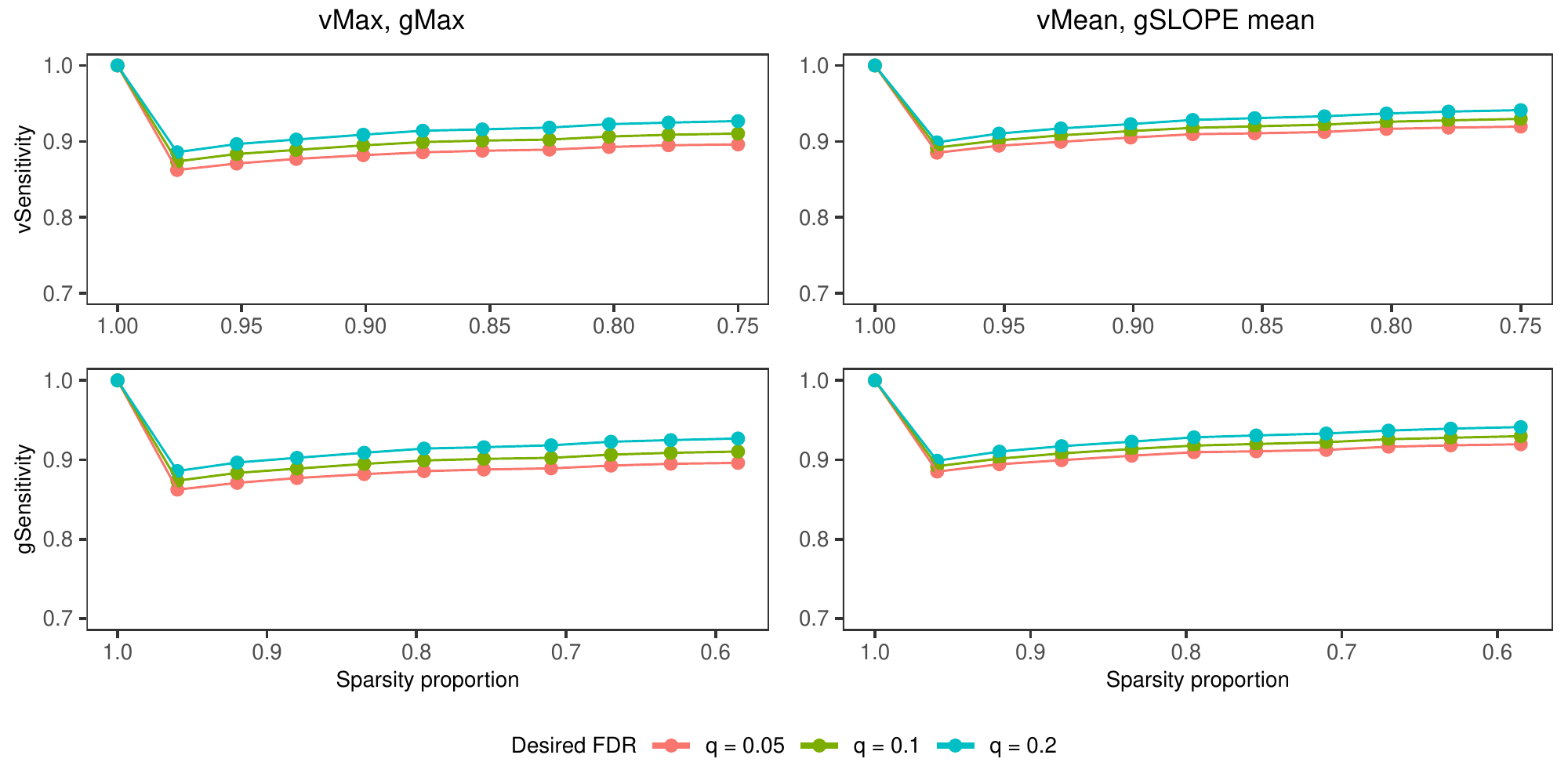}   
		\caption[width=0.8\textwidth]{vSensitivity and gSensitivity shown for SGS with the vMax, gMax and the vMean, gSLOPE mean sequences under orthogonal design with even groups, as a function of decreasing sparsity proportion. $100$ MC repetitions performed per sparsity proportion.}
		\label{fig:sim_3_even_sgs_sens_final}
\end{figure}
   \begin{figure}[H]
		\includegraphics[width=1\textwidth]{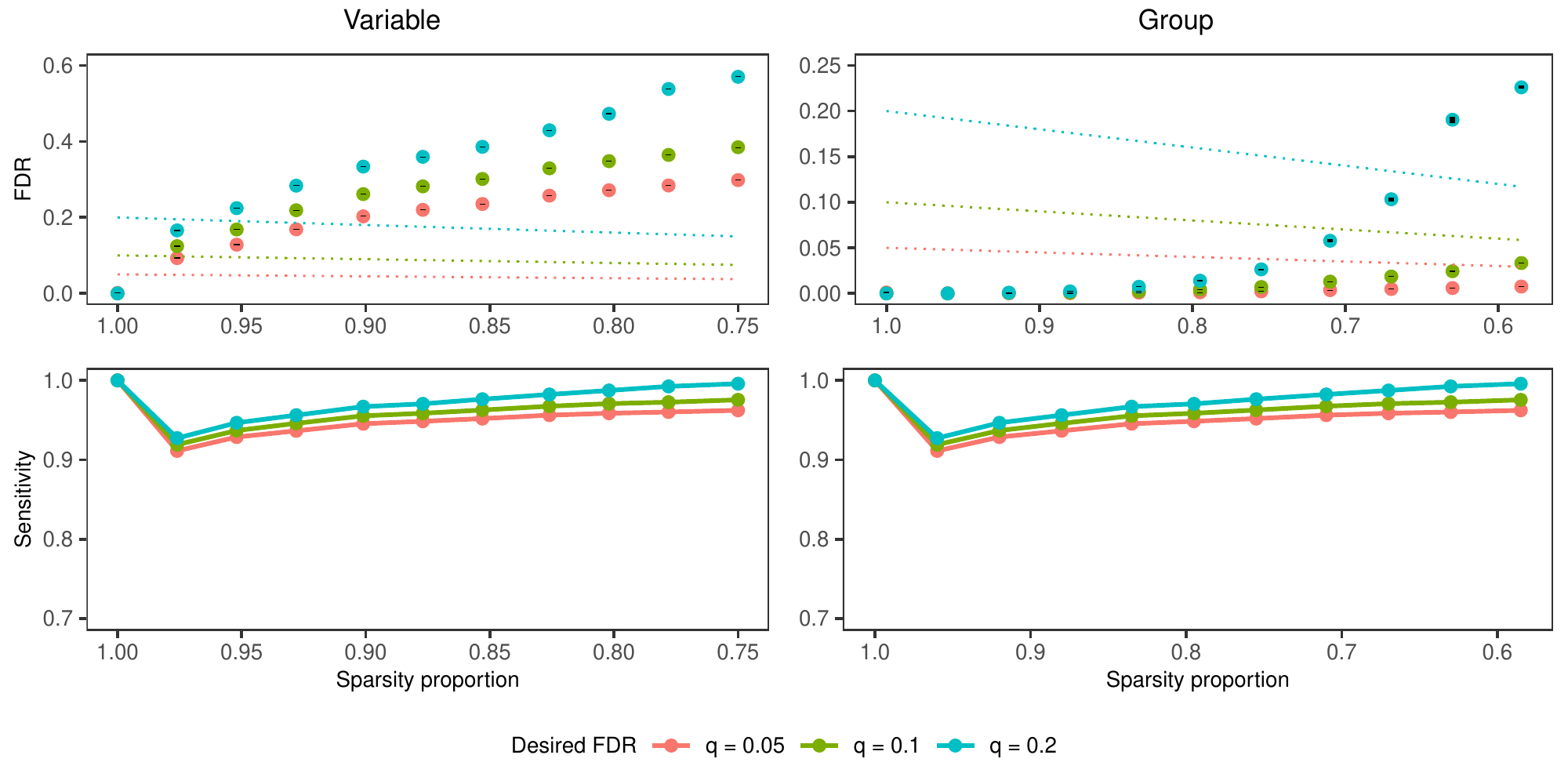} 
		\caption[width=0.8\textwidth]{Variable and group metrics shown for SGS with vMean and gMean sequences under orthogonal design with uneven groups, as a function of decreasing sparsity proportion. $100$ MC repetitions performed per sparsity proportion.}
		\label{fig:sim_3_uneven_gmean_vmean}
	\end{figure}
 \begin{figure}[H]
		\includegraphics[width=1\textwidth]{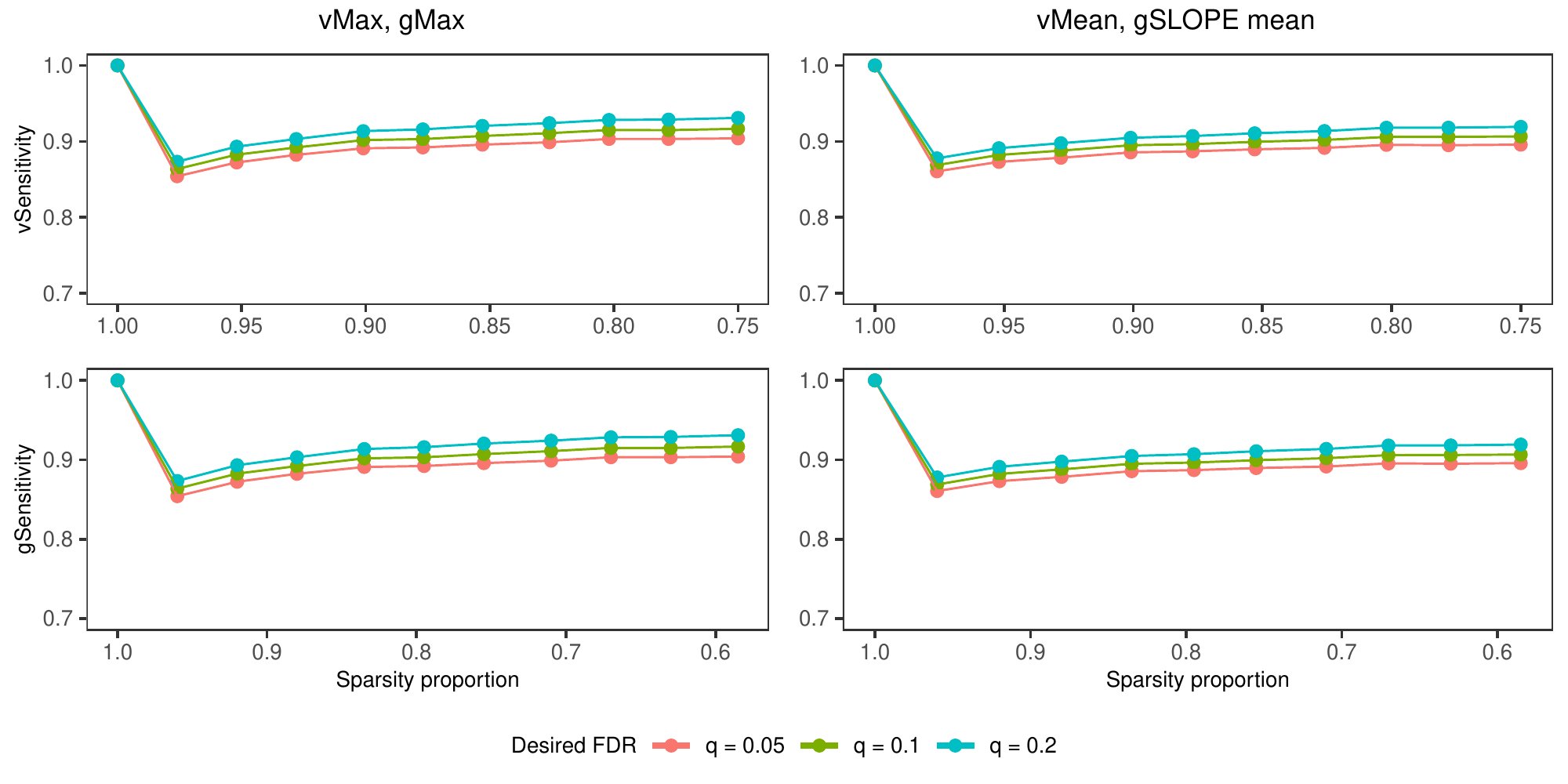}   
		\caption[width=0.8\textwidth]{vSensitivity and gSensitivity shown for SGS with the vMax, gMax and the vMean, gSLOPE mean sequences under orthogonal design with even groups, as a function of decreasing sparsity proportion. $100$ MC repetitions performed per sparsity proportion.}
		\label{fig:sim_3_uneven_sgs_sens_final}
\end{figure}

\subsection{FDR-control proofs.}\label{appendix:fdr_proof}
In both proofs, we assume without loss of generality that $\lambda = 1$. Additionally, the $1/n$ factor in Equation (\ref{eqn:sgs}) can be absorbed into $\lambda$ and is omitted from the proofs for simplicity. Hence, in the simulation studies, $\lambda$ is set to $1/n$ rather than $1$ as in the proofs. 
\newpage
\subsection{Variable FDR proof}\label{appendix:var_proof}
	\begin{proof}[Proof of Theorem \ref{thm:sgs_var_fdr_proof}]
		Under orthogonality, we can rewrite the response as $\tilde{y} := \mathbf{X}^\top y = \beta + \epsilon$. This has distribution $\tilde{y} \sim N(\beta, \mathbf{I}_p)$. Hence, SGS can be reduced to 
		\begin{equation} \label{eqn:sgs_convex_proof}
			\hat{\beta} = \argmin_{b \in \mathbb{R}^{p}}\left\{ \frac{1}{2}\left\|y- b\right\|_2^2 + \alpha \sum_{i=1}^{p} v_i \left| b \right|_{(i)} +(1-\alpha)\sum_{g=1}^{m} w_g \sqrt{p_g} \|b^{(g)}\|_2 \right\},
		\end{equation}
		Hence, from here it suffices to consider the scenario where $p=n$ and $y \sim \mathcal{N}(\beta, \mathbf{I}_p)$. 
For the hypotheses, we have that
		\begin{equation}
			\mathbb{P}(H_i^v \;  \text{rejected}) = \mathbb{P}(\hat{\beta}_i \neq 0 \; | \;\|\hat{\beta}^{(g)} \|_2 \neq 0, i\in G_g)
		\end{equation}
		Without loss of generality, we set the problem up so that the first $p_0$ hypotheses are null, i.e, $\beta_i = 0$ for $i\leq p_0$. The variable FDR is given as
		\begin{equation}\label{eqn:var_fdr}
			\text{vFDR} = \mathbb{E}\left[\frac{V^v}{\max(R^v,1)}\right] = \sum_{i=1}^{p}\mathbb{E}\left[\frac{V^v}{r}\mathbbm{1}_{\{R^v=r\}} \right] = \sum_{r=1}^{p}\frac{1}{r} \sum_{i=1}^{p_0} \mathbb{P}(H_i^v \; \text{rejected and} \; R^v=r).
		\end{equation}
		To bound the key quantity, $\mathbb{P}(H_i^v \; \text{rejected and}\; R^v=r)$, we use the following two lemmas (the proofs of which are given in $\S$\ref{appendix:lemma_proofs}).
		\begin{lemma} \label{lemma:sgs_1}
			Let $H_i^v$ be a null hypothesis, so that $i\leq p_0$ and $i \in G_g$, and let $r\geq 1$. Then,
			\begin{equation}
				\left\{y: H_i^v \; \text{rejected and} \; R^v=r\right\} =\left\{y: |y_i|>\alpha  v_r +\frac{1}{3}(1-\alpha)  a_g w_g  \; \text{and} \; R^v=r \right\}.
			\end{equation}
		\end{lemma}
		\begin{lemma}\label{lemma:sgs_split}
			Consider applying SGS to $\tilde{y} = (y_1,\dots, y_{i-1},y_{i+1},\dots,y_p)$ with weights $\tilde{v} = (v_2, \dots, v_p)$ and $\tilde{w}$ (which is $w$ if variable $i$ is not a singleton group and $w$ with its corresponding group penalty removed if it is), and let $\tilde{R}^v$ be the number of rejections generated. Then, for $r\geq 1$ and $i \in G_g$,
			\begin{equation}
				\left\{y: |y_i|>\alpha  v_r +\frac{1}{3}(1-\alpha)  a_g w_g \; \text{and} \; R^v=r \right\} = \left\{y: |y_i|>\alpha  v_r + \frac{1}{3}(1-\alpha)  a_g w_g \; \text{and} \; \tilde{R}^v=r-1 \right\}.
			\end{equation}
	\end{lemma}
 Hence, using these lemmas, we calculate 
		\begin{align}
			\mathbb{P}(H_i^v \; \text{rejected and}\; R^v=r) &\leq  \mathbb{P}\left(|y_i|>\alpha  v_r +\frac{1}{3}(1-\alpha)  a_g w_g \; \text{and} \; \tilde{R}^v=r-1\right) \\
			&=  \mathbb{P}\left(|y_i|>\alpha  v_r +\frac{1}{3}(1-\alpha)  a_g w_g\right)     \mathbb{P}\left(\tilde{R}^v=r-1\right), 
		\end{align}
		where the second step from independence of $y$ and $\tilde{y}$. Now, from the definition of $v_r$, we have, for $i \in G_g$
		
		\begin{equation}
			v_r \geq \frac{F_\mathcal{N}^{-1} \left(1-\frac{q_vr}{2p}\right) -\frac{1}{3}(1-\alpha)  a_g w_g}{\alpha } \implies 1- F_\mathcal{N}\left(\alpha v_r +\frac{1}{3}(1-\alpha)  a_g w_g\right) \leq \frac{q_vr}{2p}.
		\end{equation}
		Hence,
		\begin{align}
			\mathbb{P}\left(|y_i|>\alpha  v_r +  \frac{1}{3}(1-\alpha)  a_g w_g\right) &= \mathbb{P}\left(y_i>\alpha  v_r + \frac{1}{3}(1-\alpha)  a_g w_g\right)+    \mathbb{P}\left(y_i<-\alpha  v_r - \frac{1}{3}(1-\alpha)  a_g w_g\right) \\
			&=1-  F_\mathcal{N}\left(\alpha v_r + (1-\alpha) w_g\right)+ F_\mathcal{N}\left(-\alpha v_r -  \frac{1}{3}(1-\alpha)  a_g w_g\right) \\
			&\leq \frac{q_vr}{2p}+\frac{q_vr}{2p} =\frac{q_vr}{p},
		\end{align}
		where $F_\mathcal{N}\left(-\alpha v_r -\frac{1}{3}(1-\alpha)  a_g w_g\right) =1-F_\mathcal{N}(\alpha v_r + \frac{1}{3}(1-\alpha)  a_g w_g)  \leq \frac{q_vr}{2p}$, by the symmetry of $y_i$, as $y_i \sim\mathcal{N}(0,1)$ for $i\leq p_0$. Hence, 
		\begin{equation}
			\mathbb{P}(H_i^v \; \text{rejected and}\; R^v=r) \leq \frac{q_vr}{p}\mathbb{P}(\tilde{R}^v=r-1).
		\end{equation}
		Plugging this into Equation (\ref{eqn:var_fdr})
		\begin{align}
			\text{vFDR} &= \sum_{r=1}^{p}\frac{1}{r} \sum_{i=1}^{p_0} \mathbb{P}(H_i^v \; \text{rejected and} \; R^v=r) \\
			&\leq \sum_{r=1}^{p}\frac{1}{r} \sum_{i=1}^{p_0} \frac{q_vr}{p}\mathbb{P}(\tilde{R}^v=r-1)\\
			&=\sum_{r=1}^{p}\sum_{i=1}^{p_0} \frac{q_v}{p} \mathbb{P}(\tilde{R} = r-1)\\
			&=\sum_{r=1}^{p} \frac{q_vp_0}{p}\mathbb{P}(\tilde{R} = r-1)\\
			&=\sum_{r\geq 1} \frac{q_vp_0}{p}\mathbb{P}(\tilde{R} = r-1), \;\; \text{by Lemma \ref{lemma:sgs_1} assumption} \\
			&= \frac{q_vp_0}{p},
		\end{align}
		which concludes the proof.
	\end{proof}
 \subsubsection{Lemma proofs.}\label{appendix:lemma_proofs}
	We now provides proofs for the lemmas used.
	To prove Lemma \ref{lemma:sgs_1}, we first prove a different lemma.
	
	\begin{lemma}\label{lemma:sgs_2}
		Consider nonincreasing and nonnegative sequences $y_1 \geq \cdots \geq y_p \geq 0, v_1 \geq \cdots \geq v_p \geq 0, w_1 \geq \cdots \geq w_m \geq 0$ and let $\hat{b}$ be the solution to the problem
		\begin{align}
			\min &f(b) := \frac{1}{2} \|y - b\|_2^2 + \alpha  \sum_{i=1}^{p}v_i b_i + (1-\alpha) \sum_{g=1}^{m} w_g \sqrt{p_g} \|b^{(g)}\|_2\\
			&\text{subject to}\; b_1 \geq \dots \geq b_p \geq 0, \sqrt{p_1}\|b^{(1)}\|_2 \geq \dots \geq \sqrt{p_m}\|b^{(m)}\|_2\geq 0.
		\end{align}
		Then, if the first $r$ $\hat{b}_i$ are positive for $i\in \{1,\dots,p\}$, then for every $j\leq r$:
		\begin{equation}\label{eqn:accept_condition}
			\sum_{i=j}^{r} y_i > \alpha  \sum_{i=j}^{r} v_i + \frac{1}{3}(1-\alpha)  h \sum_{g\in \mathbb{I}_j^v}a_gw_g,
		\end{equation}
		where $\mathbb{I}_j^v = \{i\in \{1,\dots,m\} \; |\;  \exists j \in G_i \cap \{j,\dots, r\}\}$ and $a_g = |G_g \cap \{j,\dots,r\}|$. Also, for every $j\geq r+1$:
		\begin{equation}
			\sum_{i=r+1}^{j}y_i \leq \alpha   \sum_{i=r+1}^{j}v_i+(1-\alpha)   \sum_{g\in \mathbb{J}_j^v} w_g\sqrt{p_g}\sqrt{\tilde{a}_g},
		\end{equation}
		where $\mathbb{J}_j^v =\{i\in \{1,\dots,m\} \; |\;  \exists j \in G_i \cap \{r+1,\dots, j\}\}$ and $\tilde{a}_g = |G_g \cap \{r+1,\dots,j\}|$.
	\end{lemma}
	\begin{proof}\label{eqn:reject_condition}
		For the first claim, consider a new feasible (but suboptimal) solution 
		\begin{equation}
			c_i = \begin{cases}
				\hat{b}_i - h, & i\in \{j,\dots,r\}\\
				\hat{b}_i, & \text{otherwise},
			\end{cases}
		\end{equation}
		where $h<\hat{b}_j$ is a small positive scalar. By optimality of $\hat{b}$, we have $f(\hat{b}) -f(c) \leq0$. Hence,
		\begin{align}
			f(\hat{b})-f(c)  &= \frac{1}{2} \|y-\hat{b}\|_2^2 - \frac{1}{2} \|y-c\|_2^2 + \alpha  \sum_{i=1}^{p} v_i (\hat{b}_i-c_i)  + (1-\alpha)  \sum_{g=1}^{m} w_g\sqrt{p_g} (\|\hat{b}^{(g)}\|_2-\|c^{(g)}\|_2 ) \\
			&= \frac{1}{2} \sum_{i=1}^{p} \left[(y_i -\hat{b}_i )^2 - (y_i - c_i)^2\right] + \alpha  \sum_{i=1}^{p} v_i (\hat{b}_i- c_i)+ (1-\alpha)  \sum_{g=1}^{m} w_g\sqrt{p_g} (\|\hat{b}^{(g)}\|_2-\|c^{(g)}\|_2).
		\end{align}
		By definition of $c$, it follows 
		\begin{equation*}
			f(\hat{b})-f(c) = \underbrace{\frac{1}{2}\sum_{i=j}^{r} \left((y_i - \hat{b}_i)^2 - (y_i - c_i)^2\right)}_{:=T_1} + \underbrace{\alpha  \sum_{i=j}^{r} v_i ( \hat{b}_i-c_i) }_{:=T_2} + \underbrace{(1-\alpha)  \sum_{g\in \mathbb{I}_j^v} w_g\sqrt{p_g} (\|\hat{b}^{(g)}\|_2-\|c^{(g)}\|_2) }_{:=T_3}.
		\end{equation*}
		Working through each term separately:
		
		\underline{$T_1$}:
		\begin{align}
			T_1 &= \frac{1}{2}\sum_{i=j}^{r} \left((y_i - \hat{b}_i)^2 - (y_i - c_i)^2\right)\\
			& = \frac{1}{2} \sum_{i=j}^{r}(y_i^2 + \hat{b}_i^2 - 2y_i \hat{b}_i - y_i^2 + 2y_i c_i - c_i^2)\\
			&= \frac{1}{2}\sum_{i=j}^{r}(\hat{b}_i^2 - c_i^2)+ \sum_{i=j}^{r} y_i(c_i - \hat{b}_i)\\
			&= \frac{1}{2}\sum_{i=j}^{r}(\hat{b}_i^2-(\hat{b}_i-h)^2)+ \sum_{i=j}^{r} y_i(\hat{b}_i-h-\hat{b}_i)\\
			&= \frac{1}{2}\sum_{i=j}^{r}(\hat{b}_i^2 -\hat{b}_i^2 -h^2 + 2\hat{b}_i h)- \sum_{i=j}^{r} y_i h\\
			&=-\frac{1}{2} \sum_{i=j}^{r}h^2 - h\sum_{i=j}^{r}(y_i - \hat{b}_i).
		\end{align}
		
		\underline{$T_2$}:
		\begin{equation}
			T_2 =\alpha  \sum_{i=j}^{r} v_i (\hat{b}_i - c_i)  = \alpha  h \sum_{i=j}^{r}v_i.
		\end{equation}
		
		\underline{$T_3$}: Here, we can apply Bound (i) from Lemma \ref{lemma:norm_inequality}. So, for $g \in \{1,\dots,m\}$
		\begin{equation}\label{eqn:norm_inequality_1}
			\|b^{(g)}\|_2 - \|c^{(g)}\|_2 \geq \frac{ha_g}{3\sqrt{p_g}}.
		\end{equation}
		Hence,
		\begin{equation}
			T_3 = (1-\alpha)  \sum_{g\in \mathbb{I}_j^v} w_g\sqrt{p_g} (\|\hat{b}^{(g)}\|_2-\|c^{(g)}\|_2)  \geq  \frac{1}{3}(1-\alpha)  h \sum_{g\in \mathbb{I}_j^v}a_gw_g.
		\end{equation}
		Combining the three terms back together, we have that 
		\begin{align}
			&0\geq  f(\hat{b})-f(c)\geq -\frac{1}{2} \sum_{i=j}^{r}h^2 - h\sum_{i=j}^{r}(y_i - \hat{b}_i) + \alpha  h \sum_{i=j}^{r}v_i +  \frac{1}{3}(1-\alpha)  h \sum_{g\in \mathbb{I}_j^v}a_g w_g  \\
			&\implies -\frac{1}{2} \sum_{i=j}^{r}h^2 - h\sum_{i=j}^{r}(y_i - \hat{b}_i) + \alpha  h \sum_{i=j}^{r}v_i +  \frac{1}{3}(1-\alpha)  h \sum_{g\in \mathbb{I}_j^v}a_gw_g\leq 0 \\
			&\implies  -h\sum_{i=j}^{r}(y_i - \hat{b}_i) + \alpha  h \sum_{i=j}^{r}v_i +  \frac{1}{3}(1-\alpha)  h \sum_{g\in \mathbb{I}_j^v}a_g w_g  \leq -\frac{1}{2} \sum_{i=j}^{r}h^2.
		\end{align}
		We divide by $h$ and then take the limit as $h\rightarrow 0$ to obtain
		\begin{equation}
			\sum_{i=j}^{r}y_i -  \sum_{i=j}^{r}\hat{b}_i - \alpha   \sum_{i=j}^{r}v_i - \frac{1}{3}(1-\alpha)  h \sum_{g\in \mathbb{I}_j^v}a_g w_g
			\geq 0.
		\end{equation}
		Now, by assumption we have $\sum_{i=j}^{r} \hat{b}_i >0$, so
		\begin{align}
			&\sum_{i=j}^{r}y_i - \alpha   \sum_{i=j}^{r}v_i - \frac{1}{3}(1-\alpha)  h \sum_{g\in \mathbb{I}_j^v}a_g w_g> 0\\
			&\implies \sum_{i=j}^{r}y_i > \alpha   \sum_{i=j}^{r}v_i +  \frac{1}{3}(1-\alpha)  h \sum_{g\in \mathbb{I}_j^v} a_g w_g,
		\end{align}
		which proves the first claim. The second case is similar, but we instead consider a solution 
		\begin{equation}
			c_i = \begin{cases}
				h, & i\in \{r+1,\dots,j\}\\
				\hat{b}_i, & \text{otherwise},
			\end{cases}
		\end{equation}
		where $0<h<\hat{b}_r$. So, as before
		\begin{equation*}
			f(\hat{b})-f(c) = \underbrace{\frac{1}{2}\sum_{i=r+1}^{j} \left((y_i - \hat{b}_i)^2 - (y_i - c_i)^2\right)}_{:=T_1} + \underbrace{\alpha  \sum_{i=r+1}^{j} v_i ( \hat{b}_i-c_i) }_{:=T_2} + \underbrace{(1-\alpha)  \sum_{g\in \mathbb{J}_j^v} w_g\sqrt{p_g} (\|\hat{b}^{(g)}\|_2-\|c^{(g)}\|_2) }_{:=T_3}.
		\end{equation*}
		As $\hat{b}_i = 0$ for $j\geq r+1$, the calculations are simpler:
  
		\underline{$T_1$}:
		\begin{align}
			T_1 &= \frac{1}{2}\sum_{i=r+1}^{j} \left((y_i - \hat{b}_i)^2 - (y_i - c_i)^2\right)\\
			& = \frac{1}{2} \sum_{i=r+1}^{j}(y_i^2  - y_i^2 + 2y_i h - h^2)\\
			&=  h\sum_{i=r+1}^{j}y_i  - \frac{1}{2}\sum_{i=r+1}^{j} h^2.\\
		\end{align}

		\underline{$T_2$}:
		\begin{equation}
			T_2 =\alpha  \sum_{i=r+1}^{j} v_i (\hat{b}_i - c_i)  = -\alpha  h \sum_{i=r+1}^{j}v_i.
		\end{equation}
		
		\underline{$T_3$}:
		Using the reverse triangle inequality, we obtain
  \begin{equation}\label{eqn:norm_inequality_2}
			\|\hat{b}^{(g)}\|_2-\|c^{(g)}\|_2  \geq -h\sqrt{\tilde{a}_g}.
		\end{equation}
		The key thing to note here is that $c_i=h$ and $b_i = 0$ for $i\in \{r+1,\dots,j\}$, so that we have replaced the zeros with a positive scalar. Therefore $\|\hat{b}^{(g)}\|_2-\|c^{(g)}\|_2  \leq 0$, so that we have not had to change sign (and therefore could not have used 0 as a tighter bound). Hence, 
		\begin{equation}
			T_3 =   (1-\alpha)  \sum_{g\in \mathbb{J}_j^v} w_g\sqrt{p_g} (\|\hat{b}^{(g)}\|_2-\|c^{(g)}\|_2)
			\geq -(1-\alpha)  h \sum_{g\in \mathbb{J}_j^v} w_g\sqrt{p_g}\sqrt{\tilde{a}_g}. 
		\end{equation}
		Therefore,
		\begin{align}
			0 \geq f(\hat{b})-f(c) \geq h\sum_{i=r+1}^{j}y_i  - \frac{1}{2}\sum_{i=r+1}^{j} h^2 -\alpha  h \sum_{i=r+1}^{j}v_i-(1-\alpha)  h \sum_{g\in \mathbb{J}_j^v}w_g\sqrt{p_g}\sqrt{\tilde{a}_g}.
		\end{align}
		Dividing by $h$ and taking the limit as $h\rightarrow 0$ gives
		\begin{equation}
			\sum_{i=r+1}^{j}y_i \leq \alpha   \sum_{i=r+1}^{j}v_i+(1-\alpha)   \sum_{g\in \mathbb{J}_j^v} w_g\sqrt{p_g}\sqrt{\tilde{a}_g},
		\end{equation}
		proving the result.
	\end{proof}
	\textit{Note}: We used Bounds (i) from Lemmas \ref{lemma:norm_inequality} and \ref{lemma:norm_inequality_2} to obtain Equations (\ref{eqn:accept_condition}) and (\ref{eqn:reject_condition}). If we instead use Bounds (ii) from these lemmas, we obtain 
	\begin{equation}\label{eqn:alt_conditions}
		\sum_{i=j}^{r} y_i > \alpha  \sum_{i=j}^{r} v_i + (1-\alpha) \sum_{g\in \mathbb{I}_j^v}\frac{w_g\sqrt{p_g} a_g}{2\|b^{(g)}\|_2},\;
		\sum_{i=r+1}^{j}y_i \leq \alpha   \sum_{i=r+1}^{j}v_i+(1-\alpha)   \sum_{g\in \mathbb{J}_j^v} \frac{w_g \sqrt{p_g}\tilde{a}_g}{2\|b^{(g)}\|_2}.
	\end{equation}
	
	These will be useful in the proof of Lemma \ref{lemma:sgs_split}.
	
	\textbf{\textit{Proof of Lemma \ref{lemma:sgs_1}}:}
	We now use Lemma \ref{lemma:sgs_2} to prove Lemma \ref{lemma:sgs_1}. Taking $j=r$ and $R=r$ in Equation (\ref{eqn:accept_condition}) and $j=r+1$ and $R=r$ in Equation (\ref{eqn:reject_condition}), we obtain the following two expressions
	\begin{equation}
		|y|_{(r)} > \alpha  v_r +\frac{1}{3}(1-\alpha)  a_g w_g \;\; \text{and} \;\; |y|_{(r+1)} \leq \alpha  v_{r+1} + (1-\alpha) w_{g'}\sqrt{p_{g'}}\sqrt{\tilde{a}_{g'}}, 
	\end{equation}
	where $\mathbb{I}_{j=r} = g$ for $r \in G_g$ and $\mathbb{I}_{j=r+1} = g'$ for $r+1 \in G_{g'}$. To prove Lemma \ref{lemma:sgs_1}, we first want to show $\{y: H_i^v \; \text{rejected and} \; R=r\} = \{y: \hat{b}_i \neq 0 \; \text{rejected and} \; R=r\} \subset\{y: |y_i|>\alpha  v_r + \frac{1}{3}(1-\alpha)  a_g w_g \; \text{and} \; R=r \}$. The first equality is by definition, so we are only proving the subset. If we fix an $i \in \{1,\dots,p\}$ and suppose $\hat{b}_i$ is nonzero (so that we reject $H_i^v$), then it must hold that $|y_i| \geq |y|_{(r)} >\alpha  v_r +\frac{1}{3}(1-\alpha)  a_g w_g,$ proving  $\{y: H_i^v \; \text{rejected and} \; R=r\} \subset\{y: |y_i|>\alpha  v_r +\frac{1}{3}(1-\alpha)  a_g w_g  \; \text{and} \; R=r \}$.
	
	Conversely, to show the other direction, assume that $|y_i|> \alpha  v_r +\frac{1}{3}(1-\alpha)  a_g w_g $ and $R=r$. Then, we must reject $H_i^v$, since $|y_i| > \alpha  v_{r+1}+\frac{1}{3}(1-\alpha)  a_{g'} w_{g'} \geq |y|_{(r+1)}$. This shows that $\{y: H_i^v \; \text{rejected and} \; R=r\} \supset \{y: |y_i|>\alpha  v_r + \frac{1}{3}(1-\alpha)  a_g w_g  \; \text{and} \; R=r \}$, proving Lemma \ref{lemma:sgs_1}.
	\pushQED{\qed} 
	\qedhere
	\popQED
	
	\textbf{\textit{Proof of Lemma \ref{lemma:sgs_split}}:}
	We first assume without loss of generality that $y\geq \boldsymbol{0}$. The solution to Equation (\ref{eqn:sgs_convex_proof}) has $r$ strictly positive values. We need to prove that if $y_1$ is rejected, then the solution to
	\begin{equation}\label{eqn:lemma_programme}
		\min_{\tilde{b}} g(\tilde{b}) := \frac{1}{2}\sum_{i=1}^{p-1}(\tilde{y}_i - \tilde{b}_i)^2 + \alpha  \sum_{i=1}^{p-1}\tilde{v}_i |\tilde{b}|_{(i)} + (1-\alpha) \sum_{g=1}^{\tilde{m}}\tilde{w}_g \sqrt{\tilde{p}_g} \|\tilde{b}^{(g)}\|_2,
	\end{equation}
	has exactly $r-1$ non-zero values. Here, the removed $y_1$ can correspond either to a singleton group, in which $\tilde{m} = m-1$ and $\tilde{w} = w$, or it is part of a larger group, such that $\tilde{m} = m-1$ and $\tilde{w} \in \mathbb{R}^{m-1}$. We need to prove that the optimal solution $\hat{b}$ to Equation (\ref{eqn:lemma_programme}) has both at least and at most $r-1$ non-zero entries. For both, we argue by contradiction by using new suboptimal solutions, as in Lemma \ref{lemma:sgs_2}.
	
	\textit{At least $r-1$ non-zero entries:}
	Suppose by contradiction that $\hat{b}$ has $j-1$ non-zero values, $j<r$. If we denote $I = \{i: \tilde{y}_{i}\geq \tilde{y}_{j} \; \text{and} \; \tilde{y}_{i}\leq\tilde{y}_{r-1} \}$, then we can denote a new suboptimal solution as
	\begin{equation}
		c_i = \begin{cases}
			h, & i\in I\\
			\hat{b}_i, & \text{otherwise},
		\end{cases}
	\end{equation}
	where $0<h<c_{(j-1)}$. By optimality, we require $g(\hat{b}) - g(c) \leq 0$. To prove the contradiction, we will show the contrary. So, denoting  $\tilde{\mathbb{I}}_j^v = \{i\in \{1,\dots,\tilde{m}\} \; |\;  \exists j \in G_i \cap I \}$, we have 
	\begin{align}
		g(\hat{b}) - g(c) &= h\sum_{i=j}^{r-1} \tilde{y}_{(i)} - \frac{1}{2} \sum_{i=j}^{r-1} h^2 - \alpha h \sum_{i=j}^{r-1} \tilde{v}_i + (1-\alpha)  \sum_{g \in \tilde{\mathbb{I}}_j^v} \tilde{w}_g \sqrt{\tilde{p}_g} (\|\hat{b}^{(g)}\|_2 - \|c^{(g)}\|_2) \\
		&\geq h\sum_{i=j}^{r-1} \tilde{y}_{(i)} - \frac{1}{2} \sum_{i=j}^{r-1} h^2 - \alpha h \sum_{i=j}^{r-1} \tilde{v}_i - (1-\alpha)  \sum_{g \in \tilde{\mathbb{I}}_j^v} \frac{\tilde{w}_g \sqrt{\tilde{p}_g}\tilde{a}_g}{2 \|\tilde{b}\|_2}, \;\; \text{by Equation (\ref{eqn:norm_inequality_2})} \\  
		&= h\sum_{i=j}^{r-1} \tilde{y}_{(i-1)} - \frac{1}{2} \sum_{i=j+1}^{r} h^2 - \alpha h \sum_{i=j+1}^{r} \tilde{v}_i - (1-\alpha)  \sum_{g \in \tilde{\mathbb{I}}_j^v} \frac{\tilde{w}_g \sqrt{\tilde{p}_g}\tilde{a}_g}{2 \|\tilde{b}\|_2}, \;\; \text{as} \; \tilde{v}_i = v_{i-1}, \\
		&\geq  h\sum_{i=j}^{r-1} \tilde{y}_{(i)} - \frac{1}{2} \sum_{i=j+1}^{r} h^2 - \alpha h \sum_{i=j+1}^{r} \tilde{v}_i - (1-\alpha)  \sum_{g \in \tilde{\mathbb{I}}_j^v} \frac{\tilde{w}_g \sqrt{\tilde{p}_g}\tilde{a}_g}{2 \|\tilde{b}\|_2}, \;\; \text{as} \; \tilde{y}_{(i-1)} \geq y_{(i)} , \\
		&\geq  h\sum_{i=j}^{r-1} \tilde{y}_{(i)} - \frac{1}{2} \sum_{i=j+1}^{r} h^2 - \alpha h \sum_{i=j+1}^{r} \tilde{v}_i - (1-\alpha)  \sum_{g \in \mathbb{I}_j^v} \frac{w_g \sqrt{p_g}a_g}{2 \|b\|_2},
	\end{align}
	where the final inequality is due to the fact that under $\tilde{\mathbb{I}}_j^v$ we have either the same number of summations as under $\mathbb{I}_j^v$, or one less. Now, by selecting $h$ small enough, we obtain $g(\hat{b}) - g(c) >0$ by Equation (\ref{eqn:alt_conditions}),
	giving the desired contradiction. 
	
	\textit{At most $r-1$ non-zero entries:} The proof here is similar. We again argue by contradiction that $\hat{b}$ has $j$ non-zero values, $j \geq r$. Denoting $I = \{i: \tilde{y}_{i}\geq \tilde{y}_{r} \; \text{and} \; \tilde{y}_{i}\leq\tilde{y}_{j} \}$, we denote a new suboptimal solution as 
	\begin{equation}
		c_i = \begin{cases}
			\hat{b}_i - h, & i\in I\\
			\hat{b}_i, & \text{otherwise},
		\end{cases}
	\end{equation}
	where $0<h<c_{(j)}$. By optimality, we require $g(\hat{b}) - g(c) \leq 0$. To prove the contradiction, we will again show the contrary. So, 
	\begin{align}
		g(\hat{b}) - g(c) &= -\frac{1}{2}\sum_{i=r}^{j} h^2 - h\sum_{i=r}^{j} (\tilde{y}_{(i)} - \hat{b}_{(i)}) + \alpha  h \sum_{i=r}^{j} \tilde{v}_i + (1-\alpha)  \sum_{g \in \mathbb{I}_j^v} \tilde{w}_g \sqrt{\tilde{p}_g}(\|b^{(g)}\|_2 - \|c^{(g)}\|_2) \\
		&\geq -\frac{1}{2}\sum_{i=r}^{j} h^2 - h\sum_{i=r}^{j} (\tilde{y}_{(i)} - \hat{b}_{(i)}) + \alpha  h \sum_{i=r}^{j} \tilde{v}_i. 
	\end{align}
	Now,  
	\begin{equation}
		\sum_{i=r}^{j} (\tilde{y}_{(i)} - \alpha  \tilde{v}_i) =\sum_{i=r+1}^{j+1} (y_{(i)} - \alpha  v_i)  \leq 0,
	\end{equation}
	by Equation (B.4) in \cite{Bogdan2015}. Hence, by selecting $h$ to be very small, we obtain $g(\hat{b}) - g(c) >0$, giving a contradiction and finishing the proof.
	\pushQED{\qed} 
	\qedhere
	\popQED
	\subsubsection{Choice of penalty sequence}
	We can use Lemma \ref{lemma:sgs_1} to define a penalty sequence for the variables. Our aim is to choose $v_r$ such that
	\begin{equation}
		\mathbb{P}(H_i^v \; \text{rejected}) =\mathbb{P}(|y_i|>\alpha  v_r +  \frac{1}{3}(1-\alpha)  a_g w_g)  \leq \frac{q_vr}{p}.
	\end{equation}
	So,
	\begin{align}
		&\mathbb{P}( |y_i|>\alpha  v_r +   \frac{1}{3}(1-\alpha)  a_g w_g ) \leq \frac{q_vr}{p}\\
		\implies&\mathbb{P}(y_i>\alpha  v_r +  \frac{1}{3}(1-\alpha)  a_g w_g + \mathbb{P}( -y_i>\alpha  v_r + (1-\alpha) w_g ) \leq \frac{q_vr}{p}.
	\end{align}
	As $y_i\sim \mathcal{N}(0,1)$, because $i\leq p_0$, we have by symmetry 
	\begin{align}
		&   \mathbb{P}(y_i>\alpha  v_r +   \frac{1}{3}(1-\alpha)  a_g w_g) \leq \frac{q_vr}{2p},\\
		&   \implies 1- F_\mathcal{N}(\alpha  v_r +\frac{1}{3}(1-\alpha)  a_g w_g) \leq \frac{q_vr}{2p},\label{eqn:strict_condition}
	\end{align}
	for each $r \in \{1,\dots,p\}$, where $F_\mathcal{N}(\cdot)$ is the standard normal cdf. Hence, we seek 
	\begin{equation}
		\alpha  v_r +  \frac{1}{3}(1-\alpha)  a_g w_g  = F_\mathcal{N}^{-1} \left(1-\frac{q_vr}{2p}\right).
	\end{equation}
	So,
	\begin{equation}
		v_r = \frac{1}{\alpha}F_\mathcal{N}^{-1} \left(1-\frac{q_vr}{2p}\right) -  \frac{1}{3\alpha}(1-\alpha)  a_g w_g, \; r\in \{1,\dots,p\}.
	\end{equation}
	However, as we have no knowledge of group $g$, we take the sequence over the maximum possible, to ensure definite FDR-control:
	\begin{equation}
		v_r^\text{max} = \max_{g=1,\dots,m} \left\{\frac{1}{\alpha}F_\mathcal{N}^{-1} \left(1-\frac{q_vr}{2p}\right) -\frac{1}{3\alpha}(1-\alpha)  a_g w_g\right\}, \; r\in \{1,\dots,p\}.
	\end{equation}
	\subsection{Group FDR proof}\label{appendix:grp_proof}
	\begin{proof}[Proof of Theorem \ref{thm:sgs_grp_fdr_proof}]
		The proof is generally very similar to that of Theorem \ref{thm:sgs_var_fdr_proof}. Using orthogonality, we can again rewrite the problem as in Equation (\ref{eqn:sgs_convex_proof}) and we again consider, without loss of generality, the scenario where $p=n$ and $y \sim \mathcal{N}(\beta, \mathbf{I}_p)$. We have that
		\begin{equation}
			\mathbb{P}(H_i^g \;  \text{rejected}) = \mathbb{P}(\|\hat{\beta}^{(i)}\|_2 \neq 0 \; | \; \exists j\in G_i \; \text{s.t.} \; \hat{\beta}_j \neq 0).
		\end{equation}
		
		We set things up so that $m_0$ $H_i^g$ hypotheses are null, i.e, $\| \beta^{(i)}\|_2 = 0$ for $i \in \zeta_g \subset \{1,\dots,m\}$. We do not assume these are the first $m_0$ hypothesis, as is done for the variable proof. This is done to ensure both results can co-occur. Further, assume that variables corresponding to the $p_0$ null variable hypothesis sit within the null groups. Hence, we can define the group FDR as
		\begin{equation}\label{eqn:grp_fdr}
			\text{gFDR} = \sum_{r=1}^{m}\frac{1}{r} \sum_{i=1}^{m_0} \mathbb{P}(H_i^g \; \text{rejected and} \; R^g=r).
		\end{equation}
		To find the key quantity, $\mathbb{P}(H_i^g \; \text{rejected and} \; R^g=r)$, we follow a similar strategy as for the variable FDR. We make use of the following two lemmas (which are proved later):
		
		\begin{lemma}\label{lemma:sgs_3_grp}
			Let $H_i^g$ be a null hypothesis, so that $i \in \zeta_g$, and let $r\geq 1$. Then, 
			\begin{equation}
				\left\{y: H_i^g \; \text{rejected and} \; R^g=r\right\} =\left\{y:  \sum_{i \in G_r} |y_i| > \alpha  \sum_{i \in G_r}v_i + (1-\alpha) w_r p_r  \; \text{and} \; R^g=r \right\}.
			\end{equation}
		\end{lemma}
		
		\begin{lemma}\label{lemma:sgs_split_grp}
			Consider applying SGS to $\tilde{y}$, which is $y$ with the observations from group $i$ removed, with weights $\tilde{w} = (w_2, \dots, w_p)$ and $\tilde{v} = v \backslash \{v_j: j \in G_i\}$, and let $\tilde{R}^g$ be the number of rejections generated. Then, for $r\geq 1$ and $i \in G_g$,
			\begin{equation}
				\left\{y:\sum_{i \in G_r} |y_i| > \alpha  \sum_{i \in G_r}v_i + (1-\alpha) w_r p_r  \; \text{and} \; R^g=r \right\} = \left\{y:   \sum_{i \in G_r} |y_i| > \alpha  \sum_{i \in G_r}v_i + (1-\alpha) w_r p_r   \; \text{and} \; \tilde{R}^g=r-1 \right\}.
			\end{equation}
		\end{lemma}
		Hence, using these lemmas, we calculate 
		\begin{align}
			\mathbb{P}(H_i^g \; \text{rejected and}\; R^g=r) &\leq  \mathbb{P}\left( \sum_{i \in G_r} |y_i| > \alpha  \sum_{i \in G_r}v_i + (1-\alpha) w_r p_r   \; \text{and}\; \tilde{R}^g=r-1\right) \\
			&=  \mathbb{P}\left( \sum_{i \in G_r} |y_i| > \alpha  \sum_{i \in G_r}v_i + (1-\alpha) w_r p_r \right)   \mathbb{P}\left(\tilde{R}^g=r-1\right), 
		\end{align}
		where the second step follows from independence of $y$ and $\tilde{y}$. Now, from the definition of $w_g$, we have, for $i \in G_g$
		\begin{equation}
			w_g \geq \frac{F_\text{FN}^{-1}(1-\frac{q_gr}{m})-\alpha  \sum_{i \in G_r}v_i }{(1-\alpha) p_r} \implies 1- F_\text{FN}\left(\alpha  \sum_{i \in G_r}v_i + (1-\alpha) p_r w_r\right) \leq \frac{q_gr}{m}. 
		\end{equation}
		Hence,
		\begin{equation}
			\mathbb{P}\left(\sum_{i \in G_r} |y_i|> \alpha  \sum_{i \in G_r}v_i + (1-\alpha) p_r w_r\right) \leq \frac{q_gr}{m}.
		\end{equation}
		Therefore, 
		\begin{equation}
			\mathbb{P}(H_i^g \; \text{rejected and}\; R^g=r) \leq \frac{qr}{m}\mathbb{P}(\tilde{R}^g=r-1).
		\end{equation}
		Plugging this into Equation (\ref{eqn:grp_fdr})
		\begin{equation}
			\text{gFDR} = \sum_{r=1}^{m}\frac{1}{r} \sum_{i=1}^{m_0} \mathbb{P}(H_i^g \; \text{rejected and} \; R^g=r) \leq \frac{q_gm_0}{m},
		\end{equation}
		which concludes the proof.
	\end{proof}
	We now provide the proofs for the lemmas. To prove Lemma \ref{lemma:sgs_3_grp}, we first prove a different lemma.
	\begin{lemma}\label{lemma:sgs_4}
		Consider nonincreasing and nonnegative sequences $  \sum_{i \in G_1} y_i \geq \cdots \geq \sum_{i \in G_m} y_i \geq 0, v_1 \geq \cdots \geq v_p \geq 0, w_1 \geq \cdots \geq w_m \geq 0$ and let $\hat{b}$ be the solution to the problem
		\begin{align}
			\min &f(b) := \frac{1}{2} \|y - b\|_2^2 + \alpha  \sum_{i=1}^{p}v_i b_i + (1-\alpha) \sum_{g=1}^{m} w_g \sqrt{p_g} \|b^{(g)}\|_2\\
			&\text{subject to}\; b_1 \geq \dots \geq b_p \geq 0, \sqrt{p_1}\|b^{(1)}\|_2 \geq \dots \geq \sqrt{p_m}\|b^{(m)}\|_2\geq 0.
		\end{align}
		Then, if there are exactly $r$ non-zero $\|\hat{b}^{(i)}\|_2$ for $i\in \{1,\dots,m\}$, then for every $j\leq r$:
		\begin{equation}\label{eqn:accept_condition_grp}
			\sum_{i\in \mathbb{I}_j^g} y_i > \alpha   \sum_{i\in \mathbb{I}_j^g}v_i + (1-\alpha)   \sum_{g=j}^r w_g p_g,
		\end{equation}
		where $\mathbb{I}_j^g = G_j \cup \dots \cup G_r \subset \{1,\dots,p\}$ and for every $j\geq r+1$:
		\begin{equation}\label{eqn:reject_condition_grp}
			\sum_{i\in \mathbb{J}_j^g} y_i \leq \alpha   \sum_{i\in \mathbb{J}_j^g}v_i + (1-\alpha)   \sum_{g=r+1}^{j} w_g p_g,
		\end{equation}
		where $\mathbb{J}_j^g = G_{r+1} \cup \dots \cup G_j \subset \{1,\dots,p\}$.
	\end{lemma}
	\begin{proof}
		Consider a new feasible solution (but suboptimal) solution
		\begin{equation}
			c_i = \begin{cases}
				\hat{b}_i - h, & i\in \mathbb{I}_j^g\\
				\hat{b}_i, & \text{otherwise},
			\end{cases}
		\end{equation}
		where $h$ is a small positive scalar. By optimality, we have $f(\hat{b}) -f(c) \leq0$. Hence, as before
		\begin{equation}
			f(\hat{b}) - f(c) = \frac{1}{2} \sum_{i=1}^{p} \left[(y_i - \hat{b}_i)^2 - (y_i - c)^2\right] + \alpha  \sum_{i=1}^{p} v_i (\hat{b}_i - c_i)+ (1-\alpha)  \sum_{g=1}^{m} w_g\sqrt{p_g} (\|\hat{b}^{(g)}\|_2 -\|c^{(g)}\|_2).
		\end{equation}
		By definition of $c$, it follows
		\begin{equation*}
			f(\hat{b}) - f(c) = \frac{1}{2}\sum_{i\in \mathbb{I}_j^g} \left((y_i - \hat{b}_i)^2 - (y_i - c)^2\right) + \alpha  \sum_{i\in \mathbb{I}_j^g} v_i (\hat{b}_i - c_i)  + (1-\alpha)  \sum_{g=j}^r w_g\sqrt{p_g} (\|\hat{b}^{(g)}\|_2 -\|c^{(g)}\|_2).
		\end{equation*}
		The first two terms are as in Lemma \ref{lemma:sgs_2}, but with different summation indices. Hence, we provide calculations only for the final term. For this term, we make use of Bound (i) in Lemma \ref{lemma:norm_inequality} with $m=p$, so that
		\begin{equation}
			\|\hat{b}^{(g)}\|_2 -\|c^{(g)}\|_2 \geq h \sqrt{p_g}.
		\end{equation}
		Therefore,
		\begin{align}
			T_3 &= (1-\alpha)  \sum_{g=j}^r w_g\sqrt{p_g} (\|\hat{b}^{(g)}\|_2 -\|c^{(g)}\|_2)  \\
			&\geq (1-\alpha)  h \sum_{g=j}^r w_g p_g.
		\end{align}
		Combining the three terms back together, we have that 
		\begin{equation}
			0\geq   f(\hat{b})-f(c)\geq  -\frac{1}{2}\sum_{i\in \mathbb{I}_j^g} h^2 - h \sum_{i\in \mathbb{I}^g_j} (y_i - \hat{b_i})  +\alpha  h \sum_{i\in \mathbb{I}^g_j}v_i +(1-\alpha)  h \sum_{g=j}^r w_g p_g.
		\end{equation}
		We divide by $h$ and then take the limit as $h\rightarrow 0$ to obtain
		\begin{equation}
			\sum_{i\in \mathbb{I}_j^g}y_i -  \sum_{g\in \mathbb{I}_j^g}\hat{b}_i - \alpha   \sum_{i\in \mathbb{I}_j^g}v_i-(1-\alpha)   \sum_{g=j}^r w_g p_g   \geq 0.
		\end{equation}
		Now, by assumption we have $\sum_{i\in \mathbb{I}_j^g} \hat{b}_i >0$, so
		\begin{align}
			&\sum_{i\in \mathbb{I}_j^g} y_i - \alpha   \sum_{i\in \mathbb{I}_j^g} v_i - (1-\alpha)   \sum_{g=j}^r w_g p_g > 0\\
			&\implies  \sum_{i\in \mathbb{I}_j^g} y_i > \alpha   \sum_{i\in \mathbb{I}_j^g}v_i + (1-\alpha)   \sum_{g=j}^r w_g p_g,
		\end{align}
		which proves the first claim. The second case is similar, but we instead consider a solution 
		\begin{equation}
			c = \begin{cases}
				h, & i\in \mathbb{J}_j^g\\
				\hat{b}_i, & \text{otherwise}.
			\end{cases}
		\end{equation}
		The calculation is the same as in Lemma \ref{lemma:sgs_2}, but with different indices and $\sqrt{\tilde{a}_g}$ replaced by $\sqrt{p_g}$. Hence, we obtain
		\begin{equation}
			\sum_{i\in \mathbb{J}_j^g} y_i \leq \alpha   \sum_{i\in \mathbb{J}_j^g}v_i + (1-\alpha)   \sum_{g=r+1}^{j} w_g p_g,
		\end{equation}
  proving the result.
	\end{proof}
	\textbf{\textit{Proof of Lemma \ref{lemma:sgs_3_grp}:}}
	We now use Lemma \ref{lemma:sgs_4} to prove Lemma \ref{lemma:sgs_3_grp}. Taking $j=r$ and $R=r$ in Equation (\ref{eqn:accept_condition_grp}) and $j=r+1$ and $R=r$ in Equation (\ref{eqn:reject_condition_grp}), we obtain the following two expressions
	\begin{equation}
		\sum_{i \in G_r} |y_i| > \alpha  \sum_{i \in G_r}v_i + (1-\alpha) w_r p_r  \;\; \text{and} \;\;   \sum_{i \in G_{r+1}} |y_i| \leq \alpha  \sum_{i \in G_{r+1}} v_i+ (1-\alpha) w_{r+1} p_{r+1}.
	\end{equation}
	We first want to show $\{y: H_i^g \; \text{rejected and} \; R=r\} \subset\{y:   \sum_{i \in G_r}|y_i| > \alpha  \sum_{i \in G_r}v_i + (1-\alpha) w_r p_r \; \text{and} \; R=r \}$. If we fix a group, $i \in \{1,\dots,m\}$, and suppose $\|\hat{b}\|_2 \neq 0$, then $\sum_{i \in G_1} |y_i| \geq \sum_{i \in G_r} |y|_{(i)} > \alpha  \sum_{i \in G_r}v_i + (1-\alpha) w_r p_r ,$ proving $\{y: H_i^g \; \text{rejected and} \; R=r\} \subset\{y: \sum_{i \in G_r}|y_i| > \alpha  \sum_{i \in G_r}v_i + (1-\alpha) w_r p_r \; \text{and} \; R=r \}$.
	
	To show the other direction, assume that $\sum_{i \in G_r}|y_i| > \alpha  \sum_{i \in G_r}v_i + (1-\alpha) w_r p_r$ and $R=r$. Then, we must reject $H_i^g$, since $\sum_{i \in G_r}|y_i| >\sum_{i \in G_{r+1}}|y_i|$. This shows that $\{y: H_i^g \; \text{rejected and} \; R=r\} \supset \{y:   \sum_{i \in G_r} |y_i| > \alpha  \sum_{i \in G_r}v_i + (1-\alpha) w_r p_r   \; \text{and} \; R=r \}$, proving Lemma \ref{lemma:sgs_3_grp}.
	
	\textbf{\textit{Proof of Lemma \ref{lemma:sgs_split_grp}:}} Assume without loss of generality that $y > 0$. The solution to Equation (\ref{eqn:sgs_convex_proof}) has $r$ non-zero groups. We aim to prove that if $\{y_i: i \in G_1\}$ is rejected, then the solution to
	\begin{equation}
		\min_{\tilde{b}} g(\tilde{b}) := \frac{1}{2}\sum_{i \in \mathcal{G}\backslash G_1} (\tilde{y}_i - \tilde{b}_i)^2 + \alpha  \sum_{i \in \mathcal{G}\backslash G_1} \tilde{v}_i |\tilde{b}|_{i} + (1-\alpha) \sum_{g=1}^{m-1} \tilde{w}_g \sqrt{\tilde{p}_g} \|\tilde{b}^{(g)}\|_2,
	\end{equation}
	has exactly $r-1$ non-zero groups. To prove this, we will prove it has at least and at most $r-1$ non-zero groups. We again use proof by contradiction, as in Lemma \ref{lemma:sgs_split}.

	\textit{At least $r-1$ non-zero groups:} Suppose by contradiction that $\hat{b}$ has $j-1$ non-zero groups, for $j<r$. Let $I = \{G_g: \sum_{i \in G_g}\tilde{y}_i \geq \sum_{i \in G_j}\tilde{y}_i \; \text{and} \; \sum_{i \in G_g}\tilde{y}_i \leq \sum_{i \in G_{r-1}}\tilde{y}_i \}.$ Denoting a new suboptimal solution as
	\begin{equation}
		c_i = \begin{cases}
			h, & i\in I\\
			\hat{b}_i, & \text{otherwise},
		\end{cases}
	\end{equation}
	where $0< h< \min\{G_{j-1}\}.$ By optimality, we should have $g(\hat{b}) - g(c) \leq 0$. However, 
	\begin{equation}
		g(\hat{b}) - g(c) \geq h\sum_{i \in I}\tilde{y}_i - \frac{1}{2}\sum_{i \in I}h^2 - \alpha  h \sum_{i \in I}\tilde{v}_i - (1-\alpha) h \sum_{g=j}^{r-1}\tilde{p}_g \tilde{w}_g.
	\end{equation}
	The proof here is very similar to that of Lemma \ref{lemma:sgs_split}, so we will only observe the following three facts 
	\begin{itemize}
		\item $\sum_{i \in I} \tilde{y}_i \geq \sum_{i \in I'} y_i,$ where $I' =  \{G_g: \sum_{i \in G_g}y_i \geq \sum_{i \in G_{j+1}}y_i \; \text{and} \; \sum_{i \in G_g}y_i \leq \sum_{i \in G_{r}}y_i\}.$
		\item $\tilde{w}_{i} = w_{i+1}$, by design.
		\item $\alpha  \sum_{i \in I} \tilde{v}_i \leq \alpha  \sum_{i \in I'} v_i$, as we are summing over more penalty terms in the latter.
	\end{itemize}
	Using these, and by setting $h$ very small, we can apply Equation (\ref{eqn:accept_condition_grp}) to show that $g(\hat{b}) - g(c) > 0$, which is a contradiction, so that we must have at least $r-1$ non-zero groups.
	
	\textit{At most $r-1$ non-zero groups:} Suppose by contradiction that $\hat{b}$ has $j$ non-zero groups, with $j\geq r$. We again define an indicator set $I = \{G_g: \sum_{i \in G_g}\tilde{y}_i \geq \sum_{i \in G_r}\tilde{y}_i \; \text{and} \; \sum_{i \in G_g}\tilde{y}_i \leq \sum_{i \in G_{j}}\tilde{y}_i \}$, and a new suboptimal solution
	\begin{equation}
		c_i = \begin{cases}
			\hat{b}_i - h, & i\in I\\
			\hat{b}_i, & \text{otherwise},
		\end{cases}
	\end{equation}
	with $0<h<\min\{G_j\}$.
	Now,
	\begin{align}
		g(\hat{b}) - g(c) &> -\frac{1}{2} \sum_{i \in I}h^2 - h\sum_{i \in I} \tilde{y}_i + \alpha  h \sum_{i \in I} \tilde{v}_i + (1-\alpha) h \sum_{g=r}^{j} \tilde{w}_g\tilde{p}_g \\
		& \geq  -\frac{1}{2} \sum_{i \in I}h^2 - h\sum_{i \in I} \tilde{y}_i  + (1-\alpha) h \sum_{g=r}^{j} \tilde{w}_g\tilde{p}_g \\
		& =  -\frac{1}{2} \sum_{i \in I}h^2 - h\sum_{i \in I'} y_i  + (1-\alpha) h \sum_{g=r+1}^{j+1} w_g p_g, \;\; \text{by definition},
	\end{align}
	where $I' =  \{G_g: \sum_{i \in G_g}y_i \geq \sum_{i \in G_{r+1}}y_i \; \text{and} \; \sum_{i \in G_g}y_i \leq \sum_{i \in G_{j+1}}y_i\}$. Now, by looking at the proof of Lemma \ref{lemma:sgs_4}, we see that we can bound $\alpha  \sum_{i\in \mathbb{I}^g_j}v_i\geq 0$, so that instead of obtaining Equation (\ref{eqn:accept_condition_grp}), we have
	\begin{equation}
		\sum_{i\in \mathbb{I}_j^g} y_i - (1-\alpha)   \sum_{g=r+1}^j w_g p_g\leq 0.
	\end{equation}
	By picking $h$ to be very small, from this we see that we must have $g(\hat{b}) - g(c) > 0$, which is a contradiction, so that the solution has at most $r-1$ non-zero groups, proving the lemma.
 	\pushQED{\qed} 
	\qedhere
	\popQED
	\subsubsection{Choice of penalty sequence}
	We can now use Lemma \ref{lemma:sgs_3_grp} to define a penalty sequence for the groups. Our aim is to choose $w_r$ such that
	\begin{equation}
		\mathbb{P}(H_i^g \; \text{rejected}) =\mathbb{P}\left(  \sum_{i \in G_r} |y_i| > \alpha  \sum_{i \in G_r}v_i + (1-\alpha) w_r p_r\right) \leq \frac{q_gr}{p}.
	\end{equation}
	For $r \in \{1,\dots,m\}$, this is given by
	\begin{equation}
		w_r = \frac{F_\text{FN}^{-1}(1-\frac{q_gr}{p})-\alpha  \sum_{i \in G_r}v_i }{(1-\alpha) p_r}.
	\end{equation}
 Again, taking the maximum gives
	\begin{equation}
		w_r =\max_{g=1,\dots,m}\left\{\frac{F_\text{FN}^{-1}(1-\frac{q_gr}{p})-\alpha  \sum_{i \in G_r}v_i }{(1-\alpha) p_g}\right\}.
	\end{equation}
 \newpage
	\subsection{Norm results}
	\begin{lemma}\label{lemma:norm_inequality}
		For a vector $x \in \mathbb{R}^+$, $p>0$, suppose we create another vector 
		\begin{equation}
			y = \begin{cases}
				x_i - h, & i\in M\\
				x_i, & \text{otherwise},
			\end{cases}
		\end{equation}
		where $0<h<\min_{i \in M} x_i, M \subset \{1,\dots,p\}, |M|=m>0$. Then, the following two bounds hold
		\begin{align}
			&(i) \; \; \|x\|_2 - \|y\|_2 \geq h\sqrt{p} - h\sqrt{p-m}\geq \frac{hm}{3\sqrt{p}} \geq 0. \\
			&(ii) \; \; \|x\|_2 - \|y\|_2 \geq \frac{h}{2} \frac{\|x\|_{1,M}}{\|x\|_2} \geq 0.
		\end{align}
		For Bound (ii), we require a slightly stronger assumption on $h$; that is, $0<h\leq  \bar{x}/2$. 
	\end{lemma}
	\begin{proof}
		For Bound (i): For $i \notin M$, we rewrite $y_i = x_i - h + h$ and denote vectors $\tilde{h} \in \mathbb{R}^p$, where $\tilde{h}_i = h, \forall i$, and $\eta \in \mathbb{R}^p$ such that $\eta_i = h$ for $i \notin M$ and $0$ otherwise. Then, we can rewrite $y$ as $y = x - \tilde{h} + \eta$.
		Using the triangle inequality, we have
		\begin{equation}
			\|y\|_2 = \|x - \tilde{h} + \eta\|_2 \leq  \|x - \tilde{h} \|_2 - \|\eta\|_2.
		\end{equation}
		Therefore,
		\begin{equation}
			\|x\|_2 - \|y\|_2 \geq \|x\|_2 - \|x-\tilde{h}\| - \|\eta\|_2.
		\end{equation}
		For any $i \in \{1,\dots,p\}$, it holds $x_i-\tilde{h}_i = x_i(1-h/x_i) \leq x_i (1-h/\max_{i \in \{1,\dots,p\}}x_i) =x_i (1-h/\|x\|_{\infty})$, so that  $\|x - \tilde{h} \|_2 \leq \|x (1-h/\|x\|_{\infty})\|_2$. Therefore,
		\begin{equation}
			\|x\|_2 - \|x-\tilde{h}\| + \|\eta\|_2 \geq \|x\|_2 - \|x(1-h/\|x\|_{\infty})\|_2 - \|\eta\|_2.
		\end{equation}
		Now, as $(1-h/\|x\|_{\infty})\in \mathbb{R}$ is a real scalar, the reverse triangle inequality becomes an equality, so that
		\begin{align}
			\|x\|_2 - \|x(1-h/\|x\|_{\infty})\|_2 &= \|x - x(1-h/\|x\|_{\infty})\|_2 \\
			&= \|x h/\|x\|_{\infty}\|_2 \\
			&= \frac{h\|x\|_2}{\|x\|_{\infty}}.
		\end{align}
		Additionally, we have $\|\eta\|_2 = \sqrt{\sum_{i\notin M}h^2} = \sqrt{(p-m)h^2} = h \sqrt{p-m}$. Therefore, 
		\begin{align}
			\|x\|_2 - \|y\|_2 &\geq \frac{h\|x\|_2}{\|x\|_{\infty}}- h\sqrt{p-m} \\
			&=h\left(\frac{\|x\|_2}{\|x\|_{\infty}} -\sqrt{p-m} \right).
		\end{align}
		As $\|x\|_\infty \geq p^{-1/2} \|x\|_2$, it follows
		\begin{equation}
			\|x\|_2 - \|y\|_2 \geq h\left(\frac{\|x\|_2}{p^{-1/2} \|x\|_2} -\sqrt{p-m} \right) = h\sqrt{p} - h\sqrt{p-m}.
		\end{equation}
		Rewriting $\sqrt{p} - \sqrt{p-m}$ as $\frac{m}{\sqrt{p}  + \sqrt{p-m}}$, we obtain
		\begin{equation}
			\|x\|_2 - \|y\|_2 \geq \frac{hm}{\sqrt{p}  + \sqrt{p-m}} \geq \frac{hm}{2\sqrt{p}  + \sqrt{m}} \geq \frac{hm}{3\sqrt{p}  } \geq 0.
		\end{equation}

		For Bound (ii): We can rewrite $y$ as
		\begin{align}
			\|y\|_2 &=  \left(\sum_{i\notin M}x_i^2 + \sum_{i\in M}(x_i - h)^2 \right)^{1/2} \\
			&=\left(\sum_{i = 1}^{p}x_i^2 - 2h \sum_{i\in M} x_i + mh^2\right)^{1/2} \\
			&=\left(\|x\|_2^2 + h\left(mh -2 \sum_{i\in M} x_i\right)\right)^{1/2} \\
			&=\left(\|x\|_2^2\left(1+ \frac{h}{\|x\|_2^2}\left(mh -2 \sum_{i\in M} x_i\right)\right)\right)^{1/2} \\
			&=\|x\|_2\left(1+\underbrace{ \frac{h}{\|x\|_2^2}\left(mh -2 \sum_{i\in M} x_i\right)}_{=: t}\right)^{1/2}.
		\end{align}
		Now, to apply Bernoulli's inequality, we require
		\begin{align}
			&t =  \frac{h}{\|x\|_2^2}\left(mh -2 \sum_{i\in M} x_i\right) \geq -1\\
			&\implies h\left(mh -2 \sum_{i\in M} x_i\right) \geq -\|x\|_2^2 \\
			&\implies \sum_{i \in M}x_i^2 - 2h \sum_{i \in M}x_i +mh^2 \geq 0. 
		\end{align}
		It is clear that $mh^2 \geq 0$, so we only require 
		\begin{align}
			&\sum_{i \in M}(x_i^2 - 2h x_i)\geq 0 \\
			&\implies \sum_{i \in M}x_i^2 \geq 2h\sum_{i \in M} x_i \\
			&\implies \sum_{i \in M}x_i \geq 2hm \\
			& \implies h \leq  \frac{1}{2m}\sum_{i \in M}x_i = \bar{x}/2.
		\end{align}
		Hence, assuming $h \leq \bar{x}/2$, we apply the Bernoulli inequality to obtain
		\begin{align}
			\left(1+ \frac{h}{\|x\|_2^2}\left(mh -2 \sum_{i\in M} x_i\right)\right)^{1/2} \leq 1+\frac{h}{2\|x\|_2^2}\left(mh -2 \sum_{i\in M} x_i\right).
		\end{align}
		Therefore, 
		\begin{equation}
			\|y\|_2 \leq \|x\|_2 + \frac{h}{2\|x\|_2}\left(mh -2 \sum_{i\in M} x_i\right).
		\end{equation}
		Hence,
		\begin{align}
			\|x\|_2 - \|y\|_2 &\geq  \frac{h}{2\|x\|_2}\left(2 \sum_{i\in M} x_i - mh\right) \\
			&=  \frac{h}{2\|x\|_2}\left(\sum_{i\in M} x_i +\sum_{i\in M} x_i- mh\right) \\
			&\geq  \frac{h}{2\|x\|_2}\left(\sum_{i\in M} x_i\right), \;\; \text{as} \sum_{i \in M}x_i -mh \geq 0 \\
			&= \frac{h \|x\|_{1,M}}{2\|x\|_2}, \;\; \text{where} \; \|x\|_{1,M} = \sum_{i\in M} x_i,
		\end{align}
		proving Bound (ii).
	\end{proof}
	
	\begin{lemma}\label{lemma:norm_inequality_2}
		Suppose we have a vector $x \in \mathbb{R}^+$, $p>0$, where $M = \{i: x_i = 0\}, |M|=m$, and suppose further that we create another vector 
		\begin{equation}
			y = \begin{cases}
				h, & i\in M\\
				x_i, & \text{otherwise},
			\end{cases}
		\end{equation}
		where $0<h<\min_{i \in M} x_i$, and $h\leq 1$. Then, the following bound holds
		\begin{equation}
			0 \geq  \|x\|_2 - \|y\|_2 \geq -\frac{h m}{2 \|x\|_2}.
		\end{equation}
	\end{lemma}
	\begin{proof}
		The proof is similar to that of Lemma \ref{lemma:norm_inequality}. We again rewrite $y$ as
		\begin{align}
			\|y\|_2 &=  \left(\sum_{i\notin M}x_i^2 + \sum_{i\in M}h^2 \right)^{1/2} \\
			&=\left(\|x\|_2^2 + h^2 m\right)^{1/2} \\
			&=\|x\|_2\left(1+ \frac{h^2m}{\|x\|_2^2}\right)^{1/2}.
		\end{align}
		Now, as $f(x) = \sqrt{x}$ is a concave function, we can bound it from above using a linear approximation (obtained by use of Taylor's expansion) to give
		\begin{equation}
			\left(1+ \frac{h^2m}{\|x\|_2^2}\right)^{1/2} \leq 1+\frac{h^2m}{2\|x\|_2^2}.
		\end{equation}
		Hence, 
		\begin{align}
			&\|y\|_2 \leq \|x\|_2 + \frac{h^2m}{2\|x\|_2} \\
			&\implies \|x\|_2 - \|y\|_2 \geq \frac{-h^2m}{2\|x\|_2} \geq \frac{-hm}{2\|x\|_2},
		\end{align}
		where the last inequality follows from $h\leq 1$, proving the result.
	\end{proof}
 \newpage
\section{Simulation study}\label{appendix:simulation_study}
	\begin{figure}[H]
	\vspace{-5pt}
	\includegraphics[width=1\textwidth]{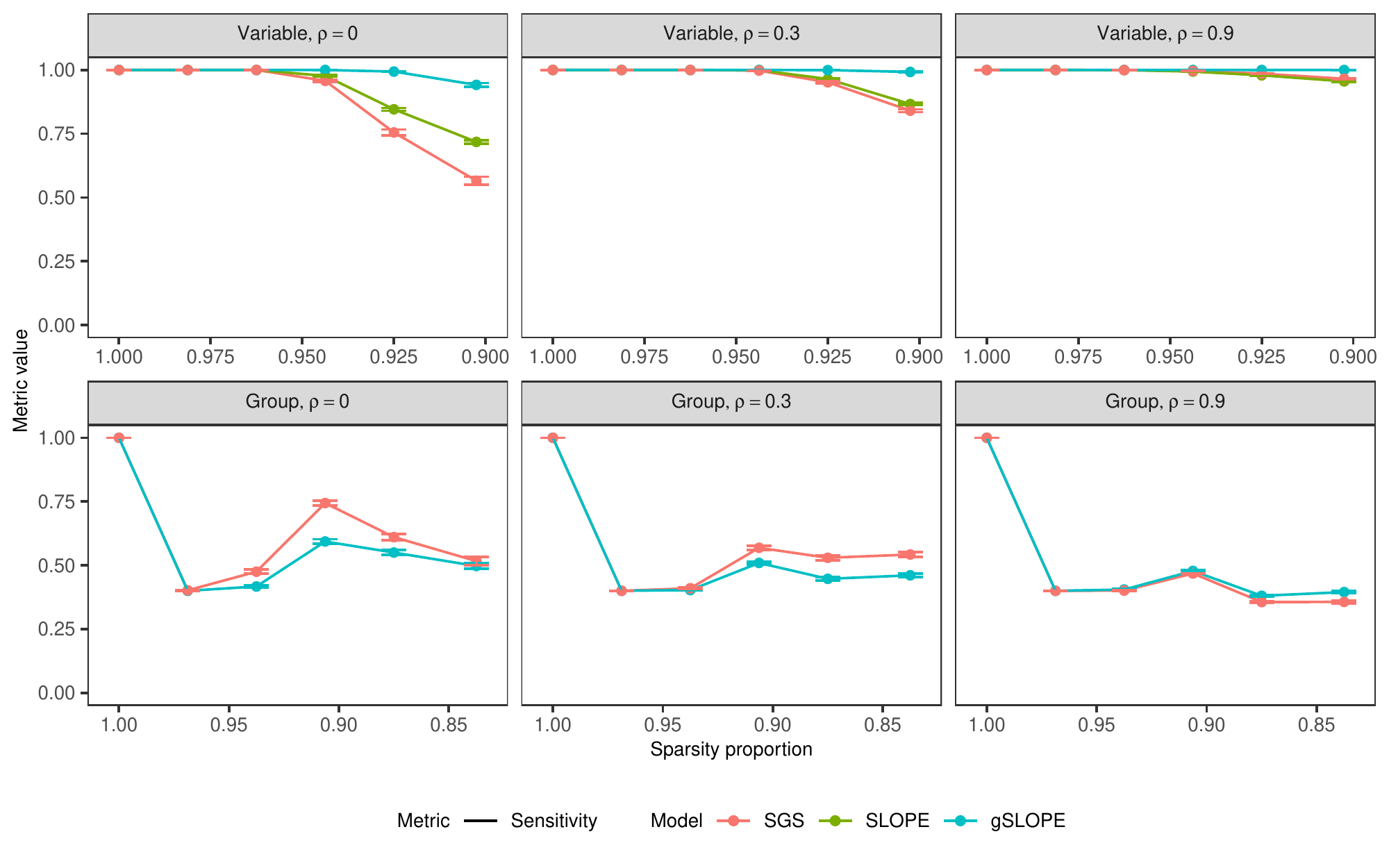}
	\vspace{-25pt}
	\caption[width=0.8\textwidth]{Sensitivity shown as a function of decreasing sparsity proportion, for the SLOPE-based models. This is shown for the different correlation cases and split by the type of selection, with standard errors shown. 100 MC repetitions performed per sparsity proportion and correlation case.}
	\label{fig:slope_models_non_orthog_bd_fixedsnr_sens}
\end{figure}
	\begin{figure}[H]
	\includegraphics[width=1\textwidth]{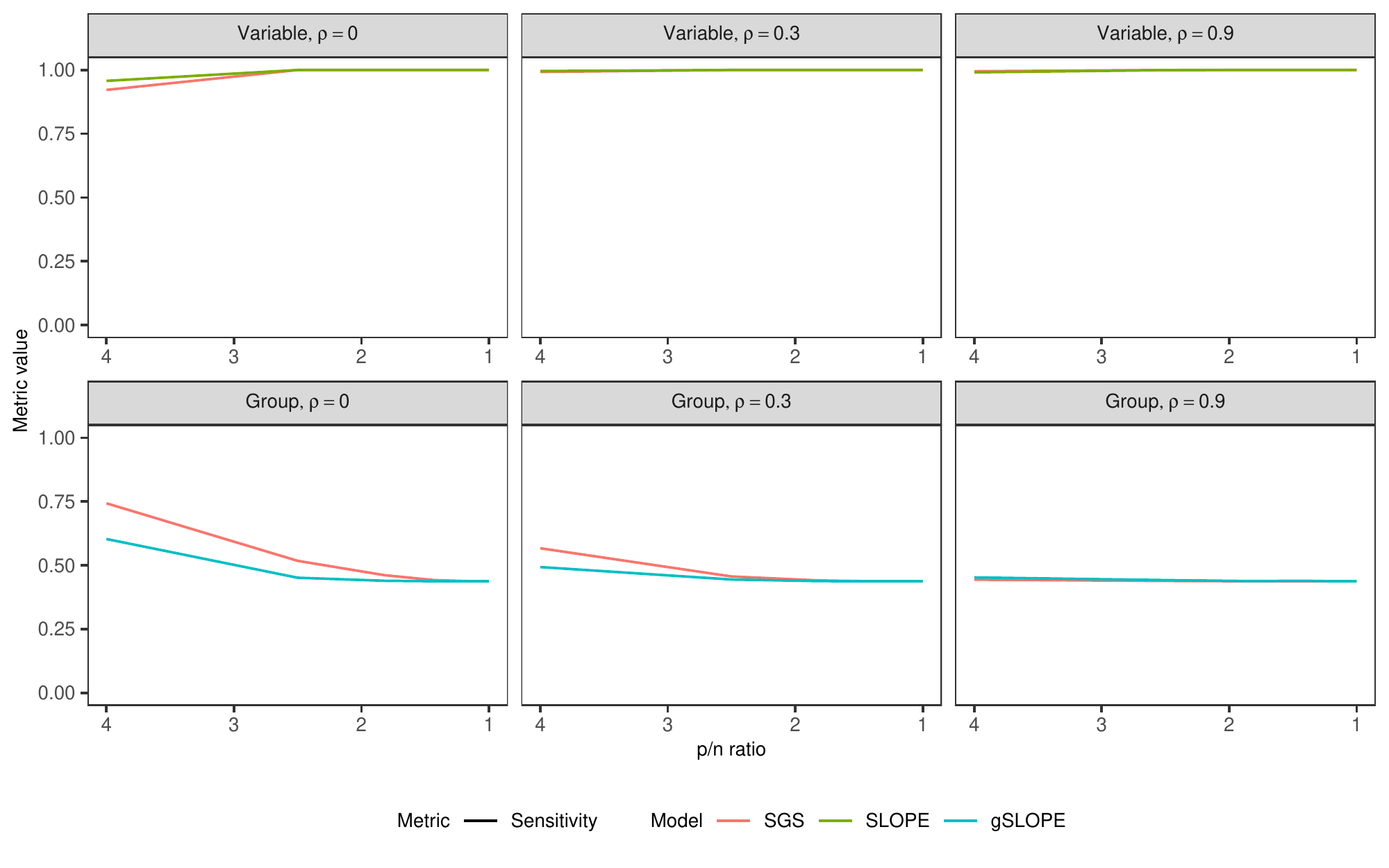}
	\vspace{-25pt}
	\caption[width=0.8\textwidth]{Sensitivity shown as a function of decreasing $p/n$ ratio, for the SLOPE-based models. This is shown for the different correlation cases and split by the type of selection. 100 MC repetitions performed per $p/n$ ratio and correlation case.}
	\vspace{-5pt}
	\label{fig:slope_models_non_orthog_bd_in_fixedsnr_sens}
\end{figure}
\section{Real data}\label{appendix:real_data}
	\begin{figure}[H]
		\vspace{-5pt}
		\includegraphics[width=1\textwidth]{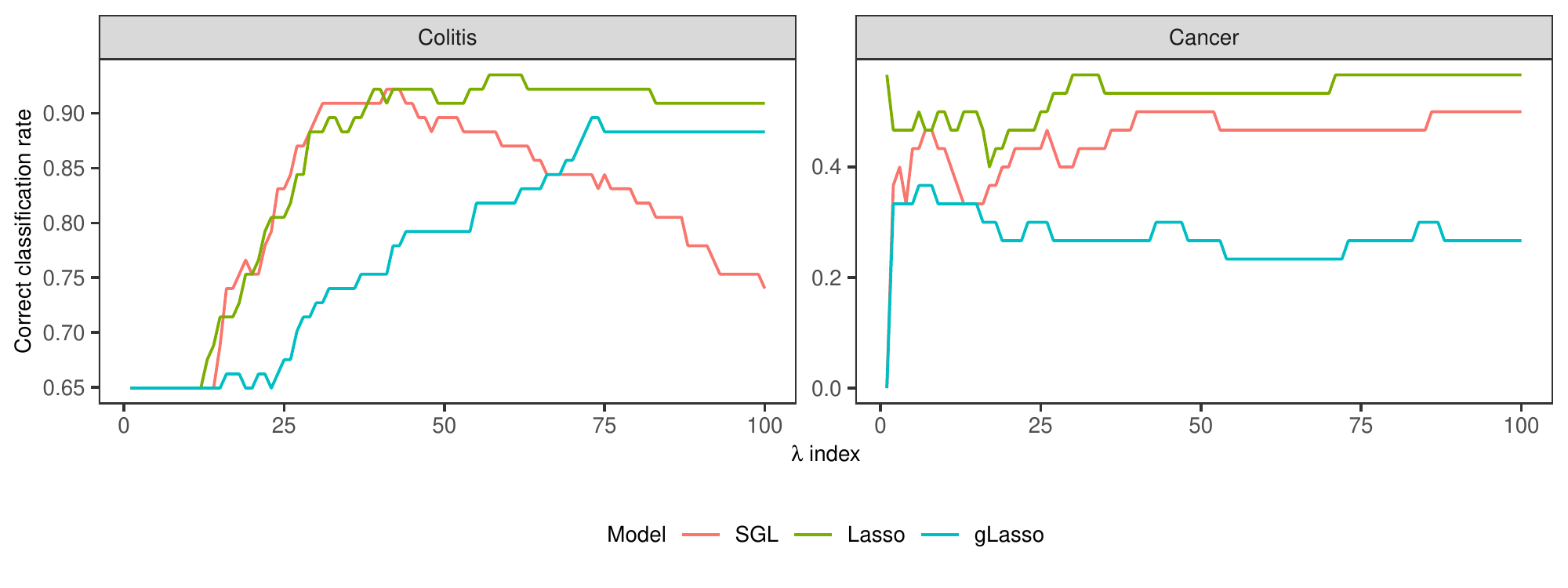}
		\vspace{-25pt}
		\caption[width=0.8\textwidth]{Correct classification rate (\%) ($\uparrow$) shown for SGL, the lasso, and gLasso applied to the colitis and cancer datasets, along a $100$-$\lambda$ regularisation path.}
		\vspace{-10pt}
		\label{fig:real_data_lasso_models}
	\end{figure}
 \begin{table}[H]
		\centering
		\begin{tabular}{r|llc}
			\hline
			Symbol&Gene name&Pathway&$\hat{\beta}$\\
			\hline
			NCK2&NCK adaptor protein 2&MIR6867\_5P&$-0.432$\\ 
			 SUZ12&Suppressor of zeste 12 homolog (Drosophila)&MIR607&$0.357$ \\
			GOLGA8N	&Golgin subfamily A member 8N&MIR3662&$0.283$\\  
			ARPC5L&Actin related protein 2/3 complex, subunit 5‐like&MIR4659A\_3P\_MIR4659B\_3P&$0.256$ 	 \\ 
			BASP1&Brain abundant, membrane attached signal protein 1&LET\_7A\_3P&$-0.179$\\ 
            C5AR1 & Complement component 5a receptor 1&MIR153\_5P&$-0.158$ \\
            TMEM158&Transmembrane protein 158&MIR5582\_3P&$-0.107$ \\
            APP& Amyloid beta (A4) precursor protein&MIR3662&$-0.0692$ \\
            RAP1A & RAP1A, member of RAS oncogene family&MIR3662&$-0.00356$ \\
			\hline
		\end{tabular}
		\caption{The nine active genes as found by the optimal SGS solution for the colitis dataset, given with their estimated coefficient value.}
		\label{tbl:colitis_solution}
	\end{table}	
  \begin{table}[H]
		\centering
		\begin{tabular}{r|llc}
			\hline
			Symbol&Gene name&Pathway&$\hat{\beta}$\\
			\hline
			COX6A1&Cytochrome C Oxidase Subunit 6A1&M40014&$-0.678$\\ 
			 SUSD3&Sushi Domain Containing 3&M40023&$-0.665$ \\
			TRIM46	&Tripartite Motif Containing 46&M39067&$-0.656$\\  
			MMP10&Matrix Metallopeptidase 10&M41652&$-0.638$ 	 \\ 
			CROCC&Ciliary Rootlet Coiled-Coil, Rootletin&M39136&$-0.360$\\ 
            CD320 & CD320 Molecule&M39018&$0.336$ \\
            RAP1GAP2&RAP1 GTPase Activating Protein 2&M40014&$-0.315$ \\
            SLC37A1& Solute Carrier Family 37 Member 1&M39064&$0.307$ \\
            ACCS & 1-Aminocyclopropane-1-Carboxylate Synthase Homolog (Inactive)&M45728&$-0.279$\\ 
            CABLES2& Cdk5 And Abl Enzyme Substrate 2&M39070&$0.236$  \\
			\hline
		\end{tabular}
		\caption{The top ten active genes as found by the optimal SGS solution for the breast cancer dataset, given with their estimated coefficients.}
		\label{tbl:cancer_solution}
	\end{table}	
 \begin{figure}[H]
		\vspace{-5pt}
		\includegraphics[width=1\textwidth]{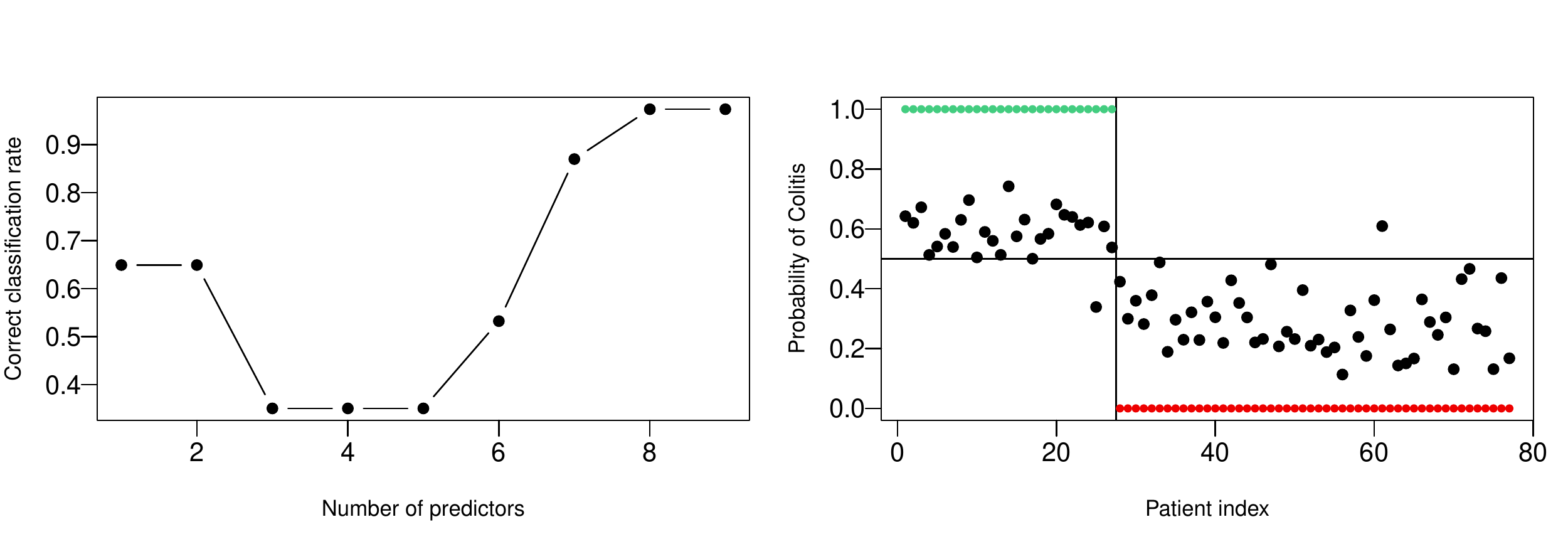}
		\vspace{-25pt}
		\caption[width=0.8\textwidth]{Left: correct classification rate (\%) as a function of the number of predictors in the model, for the optimal SGS model applied to the colitis dataset. The genes enter the model in order of their effect size. Right: the probability of a patient having colitis, according to the fitted SGS model. The decision boundaries are shown and the patients are grouped into whether they have the disease. Two misidentifications can be observed.}
		\vspace{-10pt}
		\label{fig:colitis_solution}
	\end{figure}
	\begin{table}[H]
		\centering
		\begin{tabular}{ll|c|cccc}
			\hline
			% & model & spar\_1 & spar\_2 & spar\_3 & spar\_4 & spar\_5 & spar\_6 & avg\_mse & avg\_mae \\ 
			\multirow{2}{*}{} && \multicolumn{1}{c}{SGS} & \multicolumn{4}{c}{Dataset information}\\
			Dataset& Gene set& Peak classification (\%) &$\#$ genes & $\#$ pathways & Pathway sizes & Avg. pathway size \\
			\hline
			\multirow{ 9}{*}{Colitis} &C1& 93.5& 12321&292 &[1,470] &42 \\
			&C2& 94.8& 12091& 1193& [1,888]& 10\\
			&C3& 97.4& 12031& 1408& [1,723]& 9\\
			&C4& 94.8& 8482& 613& [1,287]& 14\\
			&C5& 93.5& 11555& 614& [1,1034]& 19\\
			&C6& 93.5& 8749& 185& [1,169]&47\\
			&C7& 96.1& 12084& 936& [1,172]& 13\\
			&C8& 94.8& 11027& 601& [1,1007]& 18\\
			&H& 97.4& 3988& 50& [8,193]& 80\\
			\hline
			\multirow{ 9}{*}{Cancer} &C1& 63.3& 7233& 287& [1,338] &25\\
			&C2& 63.3& 7145& 1041& [1,449]& 7\\
			&C3& 66.7& 7088&1132& [1,449] &6\\
			&C4& 56.6& 4106& 475& [1,140] &9\\
			&C5& 60.0& 6636& 548& [1,546]&12\\
			&C6& 63.3& 4529& 183& [2,84]& 25\\
			&C7& 60.0& 7163& 896& [1,83]&8\\
			&C8& 66.7& 6375& 550& [1,533]&12 \\
			&H& 53.3& 583& 217& [1,18]&3\\
			\hline
		\end{tabular}
		\caption{Peak correct classification rate (\%) ($\uparrow$) for SGS applied to all gene sets, alongside dataset information for each gene set.}
		\label{tbl:real_data_full_results}
	\end{table}

\end{document}